\documentclass{LMCS}

\def\doi{8 (1:17) 2012}
\lmcsheading%
{\doi}
{1--50}
{}
{}
{Sep.~21, 2011}
{Mar.~\phantom02, 2012}
{}

\newif\ifnocomment
\nocommenttrue

\usepackage{lmodern}
\usepackage{mathpartir}
\usepackage{mathbbol}
\usepackage{calrsfs}
\usepackage{amssymb,amsmath}
\usepackage{stmaryrd}
\usepackage{empheq}
\usepackage{color}
\usepackage{ifpdf}
\usepackage{wasysym}
\usepackage{upgreek}
\usepackage{bussproofs}
\usepackage{ifthen}
\usepackage{enumerate,hyperref}
\def\ScoreOverhang{0pt}
\def\proofSkipAmount{}
\usepackage{listings}

\definecolor{lstbackground}{gray}{0.9}
\definecolor{lstdarkblue}{rgb}{0,0,0.5}

\lstdefinestyle{PSEUDOC}
  {language=C,
   print=true,
   numbers=left,
   numberstyle=\tiny,
   columns=fullflexible,
   basicstyle=\ttfamily,
   identifierstyle=\it\ttfamily,
   keywordstyle=\bfseries,
   morekeywords={bool,let,new,map,foreach},
   commentstyle=\color{lstdarkblue},
   mathescape=true,
   literate=
     {\ !=\ }{ $\ne$ }3
  }

\lstnewenvironment{PseudoC}
  {\lstset{style=PSEUDOC}}
  {}

\lstdefinestyle{SINGSHARP}
  {language=C,
   print=true,
   columns=fullflexible,
   basicstyle=\ttfamily,
   identifierstyle=\ttfamily,
   keywordstyle=\ttfamily\color{lucablue},
   morekeywords={class,message,ExHeap,Claims,linear,in,out,receive,send,delete,expose,null,claims,dispose,contract,initial,final,state,open,close,imp,exp,new},
   commentstyle=\color{lstdarkblue},
   mathescape=true,
    literate= 
    {->}{$\to$}1
    {\ ->}{ $\to$}3
    {\ ->\ }{ $\to$ }3
  }

\lstnewenvironment{SingSharp}[1][]
  {\lstset{style=SINGSHARP,#1}}
  {}

\newcommand{\Code}{\lstinline[style=SINGSHARP]}

\newif\iflong
\longtrue
\longfalse

\newif\ifcomments
\commentstrue
\commentsfalse

\newif\ifproofs
\proofstrue
\proofsfalse

\ifcomments
\newcommand{\REVISION}[1]{{\color{magenta}#1}}
\else
\newcommand{\REVISION}[1]{#1}
\fi

\newcommand{\hypo}[1]{\textrm{\color{lucared}(#1)}}
\newcommand{\STEP}[2][]{
  \begin{iteMize}{$-$}{}
  \item
    \ifthenelse{\equal{#1}{}}{}{\hypo{#1} }%
    #2
  \end{iteMize}
}

\newcommand{\eoe}{
  \hfill\hbox{$\blacksquare$}
}

\let\oldmarginpar\marginpar
\renewcommand{\marginpar}[1]{\oldmarginpar[\raggedleft\tiny #1]%
{\raggedright\tiny #1}}

\newcommand{\marginnote}[2]{
  \ifcomments
  \marginpar{\sf\textbf{#1}: #2}
  \fi
}
\newcommand{\Luca}[1]{\marginnote{Luca}{\color{blue}#1}}

\definecolor{lucared}{rgb}{0.5,0,0}
\definecolor{lucalert}{rgb}{0.8,0,0}
\definecolor{lucagreen}{rgb}{0,0.3,0}
\definecolor{lucablue}{rgb}{0,0,0.8}

\newcommand{\rulename}[1]{\text{{\sc(#1)}}}

\newcommand{\natset}{\mathbb{N}}

\newcommand{\Fsub}{$F_{{<}{:}}$}
\newcommand{\Sharp}[1]{\textsf{#1}{${}^{\#}$}}

\newcommand{\PolySing}{\Sharp{PolySing}}
\newcommand{\Sing}{\Sharp{Sing}}

\newcommand{\Process}{\ProcessP}
\newcommand{\ProcessP}{P}
\newcommand{\ProcessQ}{Q}
\newcommand{\ProcessR}{R}

\newcommand{\PointerSet}{\mathtt{Pointers}}
\newcommand{\Pointer}{\PointerA}
\newcommand{\PointerA}{a}
\newcommand{\PointerB}{b}
\newcommand{\PointerC}{c}
\newcommand{\PointerD}{d}

\newcommand{\Type}{\TypeT}
\newcommand{\TypeT}{t}
\newcommand{\TypeS}{s}

\newcommand{\SessionType}{\SessionTypeT}
\newcommand{\SessionTypeT}{T}
\newcommand{\SessionTypeS}{S}

\newcommand{\Name}{\NameA}
\newcommand{\NameA}{u}
\newcommand{\NameB}{v}

\newcommand{\Value}{\ValueV}
\newcommand{\ValueV}{\mathsf{v}}
\newcommand{\ValueW}{\mathsf{w}}

\newcommand{\VarSet}{\mathtt{Variables}}
\newcommand{\Var}{\VarX}
\newcommand{\VarX}{x}
\newcommand{\VarY}{y}
\newcommand{\VarZ}{z}

\newcommand{\RecVar}{\RecVarX}
\newcommand{\RecVarX}{{X}}
\newcommand{\RecVarY}{{Y}}

\newcommand{\Tag}{\mathtt{m}}

\newcommand{\Memory}{\mu}

\newcommand{\Queue}{\mathfrak{Q}}

\newcommand{\BoundContext}{\Updelta}
\newcommand{\EmptyBoundContext}{{\emptyset}}
\newcommand{\RecContext}{\Upsigma}
\newcommand{\Context}{\Upgamma}

\newcommand{\Qualifier}{q}

\newcommand{\EmptyMemory}{\emptyset}

\newcommand{\xmsg}[3]{#1\langle#2\rangle(#3)}
\newcommand{\msg}[2]{#1(#2)}

\newcommand{\EmptyQueue}{\varepsilon}
\newcommand{\shared}[1]{\co{#1}}

\newcommand{\system}[2]{(#1;#2)}

\newcommand{\parop}{\mathbin{|}}
\newcommand{\choice}{\oplus}

\newcommand{\idle}{\textbf{\color{lucagreen}0}}
\newcommand{\primitive}[1]{\text{\ttfamily\color{lucablue}#1}}
\newcommand{\openChannel}{\primitive{open}}
\newcommand{\closeChannel}{\primitive{close}}
\newcommand{\rec}{\primitive{rec}~}

\newcommand{\iteb}[3]{\primitive{if}~#1~\begin{array}[t]{@{}l@{}}\primitive{then}~#2 \\ \primitive{else}~#3\end{array}}

\newcommand{\xsend}[4]{#1!\xmsg{#2}{#3}{#4}}
\newcommand{\xreceive}[4]{#1?\xmsg{#2}{#3}{#4}}
\newcommand{\send}[3]{#1!\msg{#2}{#3}}
\newcommand{\receive}[3]{#1?\msg{#2}{#3}}

\newcommand{\varmap}{\mathsf{m}}

\newcommand{\tbool}{\mathtt{Bool}}

\newcommand{\tvar}{\tvarA}
\newcommand{\tvarA}{\alpha}
\newcommand{\tvarB}{\beta}
\newcommand{\tvarC}{\gamma}
\newcommand{\tvarD}{\delta}

\newcommand{\tmsg}[3]{#1\langle#2\rangle(#3)}

\newcommand{\Ref}[1]{{\ast}#1}
\newcommand{\Open}{{\color{lucared}\bullet}}

\newcommand{\qlin}{{\color{lucared}\mathsf{lin}}}
\newcommand{\qun}{{\color{lucared}\mathsf{un}}}

\newcommand{\trec}{\mathtt{rec}~}
\newcommand{\SessionEnd}{{\color{lucagreen}\mathsf{end}}}
\newcommand{\InternalChoice}[1]{\{{!}#1\}}
\newcommand{\ExternalChoice}[1]{\{{?}#1\}}

\newcommand{\dual}{\bowtie}
\newcommand{\drel}{\mathcal{D}}
\newcommand{\srel}{\mathcal{S}}

\newcommand{\wrel}{\mathcal{W}}

\newcommand{\red}[1]{\rightarrow_{#1}}
\newcommand{\nred}[1]{\arrownot\red{#1}}
\newcommand{\wred}[1]{\Rightarrow_{#1}}
\newcommand{\eqdef}{\stackrel{\text{def}}{=}}

\newcommand{\subt}{\leqslant}
\newcommand{\asubt}{\leqslant_{\mathsf{\color{lucared}a}}}

\newcommand{\wbound}{\mathrel{::}}
\newcommand{\wbb}[3]{#1\vdash#2\wbound#3}

\renewcommand{\land}{\mathrel{\&}}
\newcommand{\seq}{\equiv}
\newcommand{\wfdash}{\Vdash}

\newcommand{\co}[1]{\overline{#1}}
\newcommand{\dom}{\mathtt{dom}}
\newcommand{\subst}[2]{\{#1/#2\}}
\newcommand{\psubst}[2]{\{\!\{#1/#2\}\!\}}
\newcommand{\ftv}{\mathtt{ftv}}
\newcommand{\btv}{\mathtt{btv}}
\newcommand{\fpv}{\mathtt{fpv}}
\newcommand{\bpv}{\mathtt{bpv}}
\newcommand{\fn}{\mathtt{fn}}
\newcommand{\bn}{\mathtt{bn}}

\newcommand{\tail}{\mathtt{tail}}

\newcommand{\reachable}[2]{\mathtt{reach}(#1,#2)}
\newcommand{\weight}[1]{\|#1\|}
\newcommand{\xweight}[2]{\|#2\|_{#1}}
\newcommand{\trees}{\mathtt{trees}}
\newcommand{\instances}{\mathtt{instances}}
\newcommand{\aweight}{\mathtt{W}}

\theoremstyle{plain}
\newtheorem{proposition}{Proposition}[section]
\newtheorem{lemma}{Lemma}[section]
\newtheorem{theorem}{Theorem}[section]

\theoremstyle{definition}
\newtheorem{definition}{Definition}[section]
\newtheorem{example}{Example}[section]

\theoremstyle{remark}

\begin{document}

\title{Typing Copyless Message Passing}
\thanks{This work was partially supported by MIUR (PRIN 2008 DISCO)}

\author[V.~Bono]{Viviana Bono}
\address{Dipartimento di Informatica, Universit\`a degli Studi di Torino, Torino, Italy}
\email{\{bono,padovani\}@di.unito.it}  

\author[L.~Padovani]{Luca Padovani}	
\address{\vskip -6 pt}





\begin{abstract}
  We present a calculus that models a form of process interaction
  based on copyless message passing, in the style of Singularity
  OS. The calculus is equipped with a type system ensuring that
  well-typed processes are free from memory faults, memory leaks, and
  communication errors.
  The type system is essentially linear, but we show that linearity
  alone is inadequate, because it leaves room for scenarios where
  well-typed processes leak significant amounts of memory.
  We address these problems basing the type system upon an original
  variant of session types.
\end{abstract}

\maketitle

\section{Introduction}
\label{sec:intro}

\newcommand{\foo}{\Code{foo}}

Communicating systems pervade every modern computing environment
ranging from lightweight threads in multi-core architectures to Web
services deployed over wide area networks.
\emph{Message passing} is a widespread communication paradigm adopted
in many such systems. In this paradigm, it is usually the case that a
message traveling on a channel is \emph{copied} from the source to
the destination. This is inevitable in a distributed setting, where
the communicating parties are loosely coupled, but some small-scale
systems grant access to a shared address space. In these cases it is
possible to conceive a different communication paradigm --
\emph{copyless message passing} -- where only \emph{pointers} to
messages are copied from the source to the destination.
The Singularity Operating System (Singularity OS for
short)~\cite{SingularityOverview05,HuntLarus2007} is a notable example
of system that adopts the copyless paradigm.
In Singularity OS, processes have access to their own local memory as
well as to a region called \emph{exchange heap} that is shared by all
processes in the system and that is explicitly managed (objects on the
exchange heap are not garbage collected, but are explicitly allocated
and deallocated by processes).
Inter-process communication solely occurs by means of message passing
over \emph{channels} allocated on the exchange heap and messages are
themselves pointers to the exchange heap.

The copyless paradigm has obvious performance advantages, because it
may dramatically decrease the overhead caused by copying (possibly
large) messages.  At the same time, it fosters the proliferation of
subtle programming errors due to the explicit handling of pointers and
the sharing of data.
For this reason, Singularity processes
must respect an \emph{ownership invariant}: at any given point in
time, each object allocated on the exchange heap is owned by exactly
one process.
In addition, inter-process communication is regulated by so-called
\emph{channel contracts} which specify, for each channel, the
sequences of interactions that are expected to occur.
Overall, these features are meant to prevent \emph{memory faults} (the
access to non-owned/deallocated/uninitialized objects on the exchange
heap), \emph{memory leaks} (the accumulation of unreachable allocated
objects on the exchange heap), and communication errors which could
cause the abnormal termination of processes and trigger the previous
kinds of errors.

In this paper we attempt at providing a formal foundation to the
copyless paradigm from a type-theoretic point of view, along the
following lines:
\begin{iteMize}{$\bullet$}
\item We develop a process calculus that captures the essential
  features of Singularity OS and formalizes a substantial fragment of
  \Sing{}, the programming language specifically designed for the
  development of programs that run in Singularity OS.
  We provide a formal characterization of \emph{well-behaved systems},
  those that are free from memory faults, memory leaks, and
  communication errors.

\item We develop a type system ensuring that well-typed systems are
  well behaved. The type system is fundamentally based on the
  \emph{linear usage} of pointers and on \emph{endpoint types}, a
  variant of session
  types~\cite{Honda93,HondaVasconcelosKubo98,Vasconcelos09} tailored
  to the communication model of Singularity OS.
  We provide evidence that session types are a natural and expressive
  formalization of channel contracts.

\item We show that the combination of linearity and endpoint types is
  insufficient for preserving the ownership invariant, but also that
  endpoint types convey enough information to tighten the type system
  so as to guarantee its soundness.
  This allows us to give an indirect soundness proof of the current
  Singularity OS implementation.
\end{iteMize}

\noindent The rest of the paper is organized as follows.
In Section~\ref{sec:singularity} we take a quick tour of \Sing{} and
we focus on its peculiar features in the context of Singularity OS
that we are going to study more formally in the subsequent sections.
In Section~\ref{sec:types} we define syntax and semantics (in terms of
subtyping) of the type language for our type system. We also give a
number of examples showing how to represent the \Sing{} types and
channel contracts encountered in Section~\ref{sec:singularity} into
our type language.
Section~\ref{sec:processes} presents the syntax and reduction
semantics of the process calculus and ends with the formal definition
of well-behaved systems. Since we want to model the copyless paradigm,
our calculus includes an explicit representation of the exchange heap
and of the objects allocated therein. Names in the language represent
pointers to the exchange heap rather than abstract communication
channels.
Section~\ref{sec:type_system} begins showing that a traditionally
conceived type system based on linearity and behavioral types may
leave room for violations of the ownership invariant. We then devise a
type-theoretic approach to solve the problem, we present the type
rules for the exchange heap and the process calculus and the soundness
results of the type system.
In Section~\ref{sec:algorithms} we define algorithms for deciding the
subtyping relation and for implementing the type checking rules
presented in the previous section.
We relate our work with relevant literature in
Section~\ref{sec:related}, where we also detail similarities and
differences between this paper and two earlier
versions~\cite{BonoMessaPadovani11,BonoPadovani11} that have appeared
in conference and workshop proceedings.
We conclude in Section~\ref{sec:conclusions} with a summary of our
work.
For the sake of readability, proofs and additional technical material
relative to Sections~\ref{sec:types}, \ref{sec:type_system},
and~\ref{sec:algorithms} have been moved into
Appendixes~\ref{sec:extra_types}, \ref{sec:extra_type_system},
and~\ref{sec:extra_algorithms} respectively.


\section{A Taste of \Sing{}}
\label{sec:singularity}

\newcommand{\TagData}{\Code{Data}}
\newcommand{\TagEos}{\Code{Eos}}
\newcommand{\TagArg}{\Code{Arg}}
\newcommand{\TagRes}{\Code{Res}}

\begin{figure}
\begin{minipage}{0.9\textwidth}
\begin{SingSharp}[frame=single,numbers=left,numberstyle=\tiny\color{lucared}]
  void map<$\tvarA$,$\tvarB$>(imp<Mapper<$\tvarA$,$\tvarB$>:WAIT_ARG> in ExHeap mapper,
                  [Claims] imp<Stream<$\tvarA$>:START> in ExHeap source,
                  [Claims] exp<Stream<$\tvarB$>:START> in ExHeap target) {
    switch receive {
      case source.Data($\tvarA$ in ExHeap x):
        mapper.Arg(x);
        switch receive {
          case mapper.Res($\tvarB$ in ExHeap y):
            target.Data(y);
            map<$\tvarA$,$\tvarB$>(mapper, source, target);
        }

      case source.Eos():
        target.Eos();
        source.Close();
        target.Close();
    }
  }
\end{SingSharp}
\end{minipage}
\caption{\label{fig:example}\strut An example of \Sing{} code.}
\end{figure}

In this section we take a closer look at \Sing{}, the programming
language specifically designed for the development of programs that
run in Singularity OS. We do so by means of a simple, yet rather
comprehensive example that shows the main features of the language and
of its type system.
In the discussion that follows it is useful to keep in mind that
Singularity channels consist of pairs of related \emph{endpoints},
called the \emph{peers} of the channel. Messages sent over one peer
are received from the other peer, and vice versa. Each peer is
associated with a FIFO buffer containing the messages sent to that
peer that have not been received yet. Therefore, communication is
asynchronous (send operations are non-blocking) and process
synchronization must be explicitly implemented by means of suitable
handshaking protocols.

The code snippet in Figure~\ref{fig:example} defines a polymorphic
function \Code{map} that transforms a stream of data of type $\tvarA$
into a stream of data of type $\tvarB$ through a provided
mapper.\footnote{This function can be thought of as the
  communication-oriented counterpart of the higher-order,
  list-processing \Code{map} function defined in the standard library
  of virtually all functional programming languages.}
The function accepts two type arguments $\tvarA$ and $\tvarB$ and
three proper arguments: a \Code{mapper} endpoint that allows
communication with a process that performs the actual processing of
data; a \Code{source} endpoint from which data to be processed is
read; a \Code{target} endpoint to which processed data is
forwarded. For the time being, we postpone the discussion of the type
annotations of these arguments and focus instead on the operational
semantics of the function. We will come back to types shortly, when we
discuss static analysis.
The \Code{switch receive} construct (lines~4--17) is used to receive
messages from an endpoint, and to dispatch the control flow to various
cases depending on the kind of message that is received. Each
\Code{case} block specifies the endpoint from which a message is
expected and the tag of the message.
In this example, two kinds of messages can be received from the
\Code{source} endpoint: either a \TagData-tagged message (lines~5--11)
or a \TagEos-tagged message (lines~13--16).
A \TagData-tagged message contains a chunk of data to be processed,
which is bound to the local variable \Code{x} (line~5). The data is
sent in an \TagArg-tagged message on the \Code{mapper} endpoint for
processing (line~6), the result is received from the same endpoint as
a \TagRes-tagged message, stored in the local variable \Code{y}
(line~8) and forwarded on the \Code{target} endpoint as another
outgoing \TagData-tagged message (line~9). Finally, the \Code{map}
function is invoked recursively so that further data can be processed
(line~10).
An \TagEos-tagged message flags the fact that the incoming stream of
data is finished (line~13). When this happens, the same kind of
message is sent on the \Code{target} endpoint (line~14) and both the
\Code{source} and the \Code{target} endpoints are closed (lines~15
and~16).

We now illustrate the meaning of the type annotations and their
relevance with respect to static analysis.
The \Code{in ExHeap} annotations state that all the names in this
example denote pointers to objects allocated on the exchange
heap. Some of these objects (like those pointed to by \Code{source}
and \Code{target}) represent communication endpoints, others (those
pointed to by \Code{x} and \Code{y}) represent data contained in
messages. Static analysis of \Sing{} programs aims at providing strong
guarantees on the absence of errors deriving from communications and
the usage of heap-allocated objects.

Regarding communications, the correctness of this code fragment
relies on the assumption that the process(es) using the peer endpoints
of \Code{mapper}, \Code{source}, and \Code{target} are able to deal
with the message types as they are received/sent from within
\Code{map}. For instance, \Code{map} assumes to receive a
\TagRes-tagged message \emph{after} it has sent an \TagArg-tagged
message on \Code{mapper}. It also assumes that only \TagData-tagged
and \TagEos-tagged messages can be received from \Code{source} and
sent to \Code{target}, and that after an \TagEos-tagged message is
received no further message can be received from it. No classical type
associated with \Code{mapper} or \Code{source} or \Code{target} is
able to capture these temporal dependencies between such different
usages of the same object at different times. The designers of \Sing{}
have consequently devised \emph{channel contracts} describing the
allowed communication patterns on a given endpoint. Consider, for
example, the polymorphic contracts \Code{Mapper<$\tvarA$,$\tvarB$>}
and \Code{Stream<$\tvarA$>} below:

\begin{tabular}{@{}l@{\qquad\qquad}l@{}}
\begin{minipage}[t]{0.4\textwidth}
\begin{SingSharp}
contract Mapper<$\tvarA$,$\tvarB$> {
  message Arg($\tvarA$ in ExHeap);
  message Res($\tvarB$ in ExHeap);
  state WAIT_ARG
  { Arg? -> SEND_RES; }
  state SEND_RES
  { Res! -> WAIT_ARG; }
}
\end{SingSharp}
\end{minipage}
&
\begin{minipage}[t]{0.5\textwidth}
\begin{SingSharp}
contract Stream<$\tvarA$> {
  message Data($\tvarA$ in ExHeap);
  message Eos();
  state START
  { Data! -> START;
    Eos! -> END; }
  state END { }
}
\end{SingSharp}
\end{minipage}
\end{tabular}

A contract is made of a finite set of \emph{message specifications}
and a finite set of \emph{states} connected by
\emph{transitions}. Each message specification begins with the
\Code{message} keyword and is followed by the \emph{tag} of the
message and the type of its arguments. For instance, the
\Code{Stream<$\tvarA$>} contract defines the \TagData-tagged message
with an argument of type $\tvarA$ and the \TagEos-tagged message with
no arguments.
The state of the contract determines the state in which the endpoint
associated with the contract is and this, in turn, determines which
messages can be sent/received. The same contract can have multiple
states, each with a possibly different set of messages that can be
sent/received, therefore capturing the behavioral nature of
endpoints. In \Code{Stream<$\tvarA$>} we have a \Code{START} state
from which two kinds of message can be sent: if a \TagData-tagged
message is sent, the contract remains in the \Code{START} state; if a
\TagEos-tagged message is sent, the contract transits to the
\Code{END} state from which no further transitions are possible.
Communication errors are avoided by associating the two peers of a
channel with types that are complementary, in that they specify
complementary actions. This is achieved in \Sing{} with the
\Code{exp<C:$s$>} and \Code{imp<C:$s$>} type constructors that, given
a contract \Code{C} and a state $s$ of \Code{C}, respectively denote
the so-called \emph{exporting} and \emph{importing} views of \Code{C}
when in state $s$. For the sake of hindsight, it is useful to think of
the exporting view as of the type of the \emph{provider} of the
behavior specified in the contract, and of the importing view as of
the type of the \emph{consumer} of the behavior specified in the
contract.
On the one hand, the \Code{map} function in Figure~\ref{fig:example}
accepts a \Code{mapper} argument of type
\Code{imp<Map<$\tvarA$,$\tvarB$>:WAIT_ARG>} since it consumes the
mapping service accessible through the \Code{mapper} endpoint and a
\Code{source} argument of type \Code{imp<Stream<$\tvarA$>:START>}
since it consumes the source stream of data to be processed.
On the other hand, the function accepts a \Code{target} argument of
type \Code{exp<Stream<$\tvarB$>:START>} since it produces a new stream
of data on the \Code{target} endpoint.
In the code fragment in Figure~\ref{fig:example}, the endpoint
\Code{target} has type \Code{exp<Stream<$\tvarB$>:START>} on line~9,
the output of a \TagData-tagged message is allowed by the exporting
view of \Code{Stream<$\tvarB$>} in this state, and the new type of
\Code{target} on line~10 is again \Code{exp<Stream<$\tvarB$>:START>}.
Its type turns to \Code{exp<Stream<$\tvarA$>:END>} from line~14 to
line~15, when the \TagEos-tagged message is received.
The endpoint \Code{mapper} has type
\Code{imp<Mapper<$\tvarA$,$\tvarB$>:WAIT_ARG>} on line~6. The
importing view of \Code{Mapper<$\tvarA$,$\tvarB$>} allows sending a
\TagArg-tagged message in this state, hence the type of \Code{mapper}
turns to \Code{imp<Mapper<$\tvarA$,$\tvarB$>:SEND_RES>} in lines~7 and
back to type \Code{imp<Mapper<$\tvarA$,$\tvarB$>:WAIT_ARG>} from
line~8 to line~9.

A major complication of the copyless paradigm derives from the fact
that communicated objects are not copied from the sender to the
receiver, but rather pointers to allocated objects are passed
around. This can easily invalidate the ownership invariant if special
attention is not payed to whom is entitled to access which objects.
Given these premises, it is natural to think of a type discipline
controlling the \emph{ownership} of allocated objects, whereby at any
given time every allocated object is owned by one (and only one)
process. Whenever (the pointer to) an allocated object is sent as a
message, its ownership is also transferred from the sender to the
receiver.
In the example of Figure~\ref{fig:example}, the function \Code{map}
becomes the owner of data \Code{x} in line~5. When \Code{x} is sent on
endpoint \Code{mapper}, the ownership of \Code{x} is transferred from
\Code{map} to whichever process is receiving messages on
\Code{mapper}'s peer endpoint. Similarly, \Code{map} acquires the
ownership of \Code{y} on line~8, and ceases it in the subsequent
line. Overall it seems like \Code{map} is well balanced, in the sense
that everything it acquires it also released. In fact, as
\Code{mapper}, \Code{source}, and \Code{target} are also allocated on
the exchange heap, we should care also for \Code{map}'s
arguments. Upon invocation of \Code{map}, the ownership of these three
arguments transfers from the caller to \Code{map}, but when \Code{map}
terminates, only the ownership of \Code{mapper} returns to the caller,
since \Code{source} and \Code{target} are closed (and deallocated)
within \Code{map} on lines~15 and~16. This is the reason why the types
of \Code{source} and \Code{target} in the header of \Code{map} are
annotated with a \Code{[Claims]} clause indicating that \Code{map}
retains the ownership of these two arguments even after it has
returned.

From the previous discussion it would seem plausible to formalize
\Sing{} using a process calculus equipped with a suitable session type
system. Session types capture very well the sort of protocols
described by \Sing{} contracts and one could hope that, by imposing a
\emph{linear} usage on entities, the problems regarding the ownership
of heap-allocated objects would be easily solved.
In practice, things are a little more involved than this because,
somewhat surprisingly, linearity alone is \emph{too weak} to guarantee
the absence of \emph{memory leaks}, which occur when every reference
to an heap-allocated object is lost.
We devote the rest of this section to illustrating this issue through
a couple of simple examples. Consider the function:
\begin{SingSharp}
  void foo([Claims] imp<C:START> in ExHeap e,
            [Claims] exp<C:START> in ExHeap f)
  { e.Arg(f); e.Close(); }
\end{SingSharp}
which accepts two endpoints \Code{e} and \Code{f} allocated in the
exchange heap, sends endpoint \Code{f} as an \TagArg-tagged message on
\Code{e}, and closes \Code{e}. The \Code{[Claims]} annotations in the
function header are motivated by the fact that one of the two
arguments is sent away in a message, while the other is properly
deallocated within the function. Yet, this function may produce a leak
if \Code{e} and \Code{f} are the peer endpoints of the same
channel. If this is the case, only the \Code{e} endpoint is properly
deallocated while every reference to \Code{f} is lost.
Note that the \Code{foo} function behaves correctly with respect to
the \Sing{} contract
\begin{SingSharp}
  contract C {
    message Arg(exp<C:START> in ExHeap);
    state START { Arg? -> END; }
    state END { }
  }
\end{SingSharp}
whose only apparent anomaly is the implicit recursion in the type of
the argument of the \TagArg{} message, which refers to the contract
\Code{C} being defined.
A simple variation of \Code{foo} and \Code{C}, however, is equally
dangerous and does not even need this form of implicit recursion:
\begin{SingSharp}
  void bar([Claims] imp<D:START> in ExHeap e,
            [Claims] exp<D:START> in ExHeap f)
  { e.Arg<exp<D:START>>(f); e.Close(); }
\end{SingSharp}

In this case, the \TagArg-tagged message is polymorphic (it accepts a
linear argument of \emph{any} type) and the contract \Code{D} is
defined as:
\begin{SingSharp}
  contract D {
    message Arg<$\tvar$>($\tvar$ in ExHeap);
    state START { Arg? -> END; }
    state END { }    
  }
\end{SingSharp}

These examples show that, although it makes sense to allow the types
\Code{exp<C:START>} and \Code{exp<D:START>} in general, their specific
occurrences in the definition of \Code{C} and in the body of
\Code{bar} are problematic. We will see why this is the case in
Section~\ref{sec:type_system} and we shall devise a purely
type-theoretic framework that avoids these problems.
\REVISION{Remarkably, the \Code{foo} function is ill typed also in
  \Sing{}~\cite{Fahndrich06}, although the motivations for considering
  \Code{foo} dangerous come from the implementation details of
  ownership transfer rather than from the memory leaks that \Code{foo}
  can produce (see Section~\ref{sec:related} for a more detailed
  discussion).}


\section{Types}
\label{sec:types}

\newcommand{\innervars}{\BoundContext_{\textsc{i}}}
\newcommand{\outervars}{\BoundContext_{\textsc{o}}}

\begin{table}
\caption{\label{tab:type_syntax}\strut Syntax of types.}
\framebox[\textwidth]{
\begin{math}
\displaystyle
\begin{array}[t]{r@{\quad}rcl@{\quad}l}
  \textbf{Type} & \TypeT
    & ::= & \Qualifier~\SessionType & \text{(qualified endpoint type)} \\
  \\
  \textbf{Qualifier} & \Qualifier
    & ::= & \qlin & \text{(linear)} \\
  & &  |  & \qun & \text{(unrestricted)} \\
  \\
  \textbf{Endpoint Type} & \SessionType
    & ::= & \SessionEnd & \text{(termination)} \\
  & &  |  & \tvar & \text{(type variable)} \\
  & &  |  & \InternalChoice{\tmsg{\Tag_i}{\tvar_i}{\Type_i}.\SessionType_i}_{i\in I} & \text{(internal choice)} \\
  & &  |  & \ExternalChoice{\tmsg{\Tag_i}{\tvar_i}{\Type_i}.\SessionType_i}_{i\in I} & \text{(external choice)} \\
  & &  |  & \trec\tvar.\SessionType & \text{(recursive type)} \\
\end{array}
\end{math}
}
\end{table}

We introduce some notation for the type language: we assume an
infinite set of \emph{type variables} ranged over by $\tvarA$,
$\tvarB$, $\dots$; we use $\TypeT$, $\TypeS$, $\dots$ to range over
types, $\Qualifier$ to range over qualifiers, and $\SessionTypeT$,
$\SessionTypeS$, $\dots$ to range over endpoint types.
The syntax of types and endpoint types is defined in
Table~\ref{tab:type_syntax}.
An endpoint type describes the allowed behavior of a process with
respect to a particular endpoint. The process may send messages over
the endpoint, receive messages from the endpoint, and deallocate the
endpoint.
The endpoint type $\SessionEnd$ denotes an endpoint on which no
input/output operation is possible and that can only be deallocated.
An internal choice
$\InternalChoice{\tmsg{\Tag_i}{\tvar_i}{\Type_i}.\SessionType_i}_{i\in
  I}$ denotes an endpoint on which a process may send any message with
tag $\Tag_i$ for $i\in I$. The message has a \emph{type parameter}
$\tvar_i$, which the process can instantiate with any endpoint type
(but we will impose some restrictions in
Section~\ref{sec:type_system}), and an argument of type $\Type_i$.
Depending on the tag $\Tag_i$ of the message, the endpoint can be used
thereafter according to the endpoint type $\SessionTypeT_i$.
In a dual manner, an external choice
$\ExternalChoice{\tmsg{\Tag_i}{\tvar_i}{\TypeT_i}.\SessionType_i}_{i\in
  I}$ denotes and endpoint from which a process must be ready to
receive any message with tag $\Tag_i$ for $i\in I$. Again, $\tvar_i$
is the type parameter of the message and $\Type_i$ denotes the type of
the message's argument. Depending on the tag $\Tag_i$ of the received
message, the endpoint is to be used according to $\SessionTypeT_i$.
The duality between internal and external choices regards not only the
dual send/receive behaviors of processes obeying these types, but also
the quantification of type parameters in messages, which we can think
universally quantified in internal choices (the sender chooses how to
instantiate the type variable) and existentially quantified in
external choices (the receiver does not know the type with which the
type variable has been instantiated).
In endpoint types
$\InternalChoice{\tmsg{\Tag_i}{\tvar_i}{\TypeT_i}.\SessionType_i}_{i\in
  I}$ and
$\ExternalChoice{\tmsg{\Tag_i}{\tvar_i}{\TypeT_i}.\SessionType_i}_{i\in
  I}$ we assume that $\Tag_i = \Tag_j$ implies $i = j$. That is, the
tag $\Tag_i$ of the message that is sent or received identifies a
unique continuation $\SessionType_i$.
Terms $\trec\tvar.\SessionTypeT$ can be used to specify recursive
behaviors, as usual.
The role of type variables $\tvar$ is twofold, depending on whether
they are bound by a recursion $\trec\tvar.\SessionTypeT$ or by a
prefix $\tmsg\Tag\tvar\Type$ in a choice: they either represent
recursion points, like $\tvar$ in
$\trec\tvar.{!}\tmsg\Tag\tvarB\Type.\tvar$, or abstracted endpoint
types, like $\tvar$ in
${!}\tmsg\Tag\tvar{\qlin~{?}\tmsg{\Tag'}\tvarB\Type.\tvar}.\SessionEnd$.
We will see plenty of examples of both usages in the following.

Even though the type system focuses on linear objects allocated on the
exchange heap, the type language must be expressive enough to describe
Singularity OS entities like system-wide services or \Sing{} functions
and procedures. For this reason, we distinguish \emph{linear}
resources from \emph{unrestricted} ones and, along the lines
of~\cite{Vasconcelos09,GiuntiVasconcelos10}, we define types as
qualified endpoint types. A qualifier is either `$\qlin$', denoting a
\emph{linear endpoint type} or `$\qun$', denoting an
\emph{unrestricted endpoint type}. Endpoints with a linear type must
be owned by exactly one process at any given time, whereas endpoints
with an unrestricted type can be owned by several (possibly zero)
processes at the same time.
Clearly, not every endpoint type can be qualified as unrestricted, for
the type system relies fundamentally on linearity in order to enforce
its properties. In the following we limit the use of the `$\qun$'
qualifier to endpoint types of the form
$\trec\tvar.\InternalChoice{\tmsg{\Tag_i}{\tvar_i}{\Type_i}.\tvar}_{i\in
  I}$, whose main characteristic is that they do not change over time
(each continuation after an output action is $\tvar$, that is the
whole endpoint type itself). In a sense, they are not behavioral
types, which intuitively explains why they can be safely qualified as
unrestricted.

Here are some conventions regarding types and endpoint types:
\begin{iteMize}{$\bullet$}
\item we sometimes use an infix notation for internal and external
  choices and write
\[
{!}\tmsg{\Tag_1}{\tvar_1}{\Type_1}.\SessionTypeT_1
  \oplus \cdots \oplus
  {!}\tmsg{\Tag_n}{\tvar_n}{\Type_n}.\SessionTypeT_n
\text{\qquad instead of\qquad}
\InternalChoice{\tmsg{\Tag_i}{\tvar_i}{\Type_i}.\SessionTypeT_i}_{i\in\{1,\dots,n\}}
\]
and
\[
{?}\tmsg{\Tag_1}{\tvar_1}{\Type_1}.\SessionTypeT_1
  + \cdots +
  {?}\tmsg{\Tag_n}{\tvar_n}{\Type_n}.\SessionTypeT_n
\text{\qquad instead of\qquad}
\ExternalChoice{\tmsg{\Tag_i}{\tvar_i}{\Type_i}.\SessionTypeT_i}_{i\in\{1,\dots,n\}}
\]

\item we omit the type variable specification $\langle\tvar\rangle$
  when useless (if the type variable occurs nowhere else) and write,
  for example, ${!}\Tag(\Type).\SessionTypeT$;

\item for the sake of simplicity, we formally study (endpoint) types
  where messages carry exactly one type/value argument, but we will be
  more liberal in the examples;

\item we write $\qlin(\Type)$ and $\qun(\Type)$ to mean that $\Type$
  is respectively linear and unrestricted.
\end{iteMize}

\begin{table}
\caption{\label{tab:wf}\strut Well-formedness rules for endpoint
  types.}
\framebox[\textwidth]{
\begin{math}
\displaystyle
\begin{array}{c}
\inferrule[\rulename{WF-End}]{}{
  \outervars; \innervars \wfdash \SessionEnd
}
\qquad
\inferrule[\rulename{WF-Var}]{
  \tvar \in \outervars \setminus \innervars
}{
  \outervars; \innervars \wfdash \tvar
}
\qquad
\inferrule[\rulename{WF-Rec}]{
  \outervars, \tvar; \innervars \setminus \{ \tvar \} \wfdash \SessionType
}{
  \outervars; \innervars \wfdash \trec\tvar.\SessionType
}
\\\\
\inferrule[\rulename{WF-Prefix}]{
  \dagger\in\{{!},{?}\}
  \\
  (\outervars \cup \innervars), \tvar_i; \emptyset \wfdash \SessionTypeS_i
  ~{}^{(i\in I)}
  \\
  \outervars; \innervars, \tvar_i \wfdash \SessionTypeT_i
  ~{}^{(i\in I)}
}{
  \outervars; \innervars \wfdash {\dagger}\{\xmsg{\Tag_i}{\tvar_i}{\Qualifier_i~\SessionTypeS_i}.\SessionTypeT_i\}_{i\in I}
}
\end{array}
\end{math}
}
\end{table}

We have standard notions of free and bound type variables for
(endpoint) types. The binders are $\trec$ and
$\xmsg\Tag\tvar\Type$. In particular, $\trec\tvar.\SessionTypeT$ binds
$\tvar$ in $\SessionTypeT$ and
${\dagger}\xmsg\Tag\tvar\Type.\SessionTypeT$ where ${\dagger} \in \{
{!}, {?} \}$ binds $\tvar$ in $\Type$ and in $\SessionTypeT$.
We will write $\ftv(\SessionType)$ and $\btv(\SessionType)$ for the
set of free and bound type variables of $\SessionType$.
We require that type variables bound by a recursion $\trec$ must be
guarded by a prefix (therefore a non-contractive endpoint type such as
$\trec\tvar.\tvar$ is forbidden) and that type variables bound in
$\langle\tvar\rangle$ as in ${!}\tmsg\Tag\tvar\Type.\SessionTypeT$ can
only occur in $\Type$ and within the prefixes of $\SessionTypeT$.
We formalize this last requirement as a well-formedness predicate for
types denoted by a judgment $\outervars; \innervars \wfdash
\SessionType$ and inductively defined by the axioms and rules in
Table~\ref{tab:wf}.
The set $\outervars$ contains so-called \emph{outer variables} (those
that can occur everywhere) while the set $\innervars$ contains
so-called \emph{inner variables} (those that can occur only within
prefixes).
Here and in the following we adopt the convention that $\BoundContext,
\BoundContext'$ denotes $\BoundContext \cup \BoundContext'$ when
$\BoundContext \cap \BoundContext' = \emptyset$ and is undefined
otherwise; we also write $\BoundContext, \tvar$ instead of
$\BoundContext, \{ \tvar \}$.
We say that $\SessionType$ is well formed with respect to
$\BoundContext$, written $\BoundContext \wfdash \SessionType$, if
$\BoundContext; \emptyset \wfdash \Type$ is derivable.
Well formedness restricts the expressiveness of types, in particular
endpoint types such as ${!}\xmsg\Tag\tvar\Type.\tvar$ and
${?}\xmsg\Tag\tvar\Type.\tvar$ are not admitted because ill formed. We
claim that ill-formed endpoint types have little practical utility: a
process using an endpoint with type ${!}\xmsg\Tag\tvar\Type.\tvar$
knows the type with which $\tvar$ is instantiated while no process is
capable of using an endpoint with type ${?}\tmsg\Tag\tvar\Type.\tvar$
since nothing can be assumed about the endpoint type with which
$\tvar$ is instantiated.

In what follows we consider endpoint types modulo renaming of bound
variables and \REVISION{the law $\trec\tvar.\SessionType =
  \SessionType\subst{\trec\tvar.\SessionType}{\tvar}$ where
  $\SessionType\subst{\trec\tvar.\SessionType}{\tvar}$ is the
  capture-avoiding substitution of $\trec\tvar.\SessionType$ in place
  of every free occurrence of $\tvar$ in $\SessionType$.}  Whenever we
want to reason on the structure of endpoint types, we will use a
syntactic equality operator $\seq$.  Therefore we have
$\trec\tvar.\SessionType \not\seq
\SessionType\subst{\trec\tvar.\SessionType}{\tvar}$ (recall that
$\SessionType$ cannot be $\tvar$ for contractivity).

\newcommand{\Tmap}{\SessionType_{\mathtt{map}}}
\newcommand{\TMapper}{\SessionType_{\mathtt{Mapper}}}
\newcommand{\TStream}{\SessionType_{\mathtt{Stream}}}

\begin{example}
\label{ex:contracts}
Consider the contracts \Code{Mapper<$\tvarA$,$\tvarB$>} and
\Code{Stream<$\tvarA$>} presented in Section~\ref{sec:singularity}.
We use the endpoint types
\[
\begin{array}{rcl}
  \TMapper(\tvarA, \tvarB) & = &
  \trec\tvarC.
  {?}\msg{\mathtt{Arg}}{\qlin~\tvarA}.
  {!}\msg{\mathtt{Res}}{\qlin~\tvarB}.
  \tvarC
  \\
  \TStream(\tvarA) & = &
  \trec\tvarC.(
  {!}\msg{\mathtt{Data}}{\qlin~\tvarA}.\tvarC
  \oplus
  {!}\msg{\mathtt{Eos}}{}.\SessionEnd
  )
\end{array}
\]
to denote the \Sing{} types
\Code{exp<Mapper<$\tvarA$,$\tvarB$>:WAIT_ARG>} and
\Code{exp<Stream<$\tvarA$>:START>} respectively.
Recursion models loops in the contracts and each state of a contract
corresponds to a particular subterm of $\TMapper(\tvarA,\tvarB)$ and
$\TStream(\tvarA)$. For instance, the \Sing{} type
\Code{exp<Mapper<$\tvarA$,$\tvarB$>:SEND_RES>} is denoted by the
endpoint type
${!}\msg{\mathtt{Res}}{\qlin~\tvarB}.\TMapper(\tvarA,\tvarB)$.
The type of message arguments are embedded within the endpoint types,
like in session types but unlike \Sing{} where they are specified in
separate \Code{message} directives.
The $\qlin$ qualifiers correspond to the \Code{in ExHeap} annotations
and indicate that these message arguments are \emph{linear} values.

Observe that both endpoint types are open, as the type variables
$\tvarA$ and $\tvarB$ occur free in them. We will see how to embed
these endpoint types into a properly closed type for \Code{map} in
Example~\ref{ex:map_type}.
\eoe
\end{example}

Duality is a binary relation between endpoint types that describe
complementary actions. Peer endpoints will be given dual endpoint
types, so that processes accessing peer endpoints will interact
without errors: if one of the two processes sends a message of some
kind, the other process is able to receive a message of that kind; if
one process has finished using an endpoint, the other process has
finished too.


\begin{definition}[duality]
\label{def:duality}
We say that $\drel$ is a \emph{duality relation} if $(\SessionTypeT,
\SessionTypeS) \in {\drel}$ implies either
\begin{iteMize}{$\bullet$}
\item $\SessionTypeT = \SessionTypeS = \SessionEnd$, or

\item $\SessionTypeT =
  \ExternalChoice{\tmsg{\Tag_i}{\tvar_i}{\Type_i}.\SessionTypeT_i}_{i\in
    I}$ and $\SessionTypeS =
  \InternalChoice{\tmsg{\Tag_i}{\tvar_i}{\Type_i}.\SessionTypeS_i}_{i\in I}$
  and $(\SessionTypeT_i, \SessionTypeS_i) \in {\drel}$ for every $i\in
  I$, or

\item $\SessionTypeT =
  \InternalChoice{\tmsg{\Tag_i}{\tvar_i}{\Type_i}.\SessionTypeT_i}_{i\in
    I}$ and $\SessionTypeS =
  \ExternalChoice{\tmsg{\Tag_i}{\tvar_i}{\Type_i}.\SessionTypeS_i}_{i\in
    I}$ and $(\SessionTypeT_i, \SessionTypeS_i) \in {\drel}$ for every
  $i\in I$.
\end{iteMize}

\noindent We write $\dual$ for the largest duality relation and we say that
$\SessionTypeT$ and $\SessionTypeS$ are \emph{dual} if $\SessionTypeT
\dual \SessionTypeS$.
\end{definition}

\REVISION{
  We will see that every well-formed endpoint type $\SessionTypeT$ has
  a dual -- that we denote by $\co\SessionTypeT$ -- which is
  intuitively obtained from $\SessionTypeT$ by swapping $?$'s with
  $!$'s. The formal definition of $\co\SessionTypeT$, however, is
  complicated by the possible occurrence of recursion variables within
  prefixes.
  As an example, the dual of the endpoint type $\SessionTypeT =
  \trec\tvarA.{!}\tmsg{\Tag}{\tvarB}{\tvarA}.\SessionEnd$ is
  \emph{not} $\SessionTypeS =
  \trec\tvarA.{?}\tmsg{\Tag}{\tvarB}{\tvarA}.\SessionEnd$ but rather
  $\trec\tvarA.{?}\tmsg{\Tag}{\tvarB}{\SessionTypeT}.\SessionEnd$.
  This is because, by unfolding the recursion in $\SessionTypeT$, we
  obtain $\SessionTypeT =
  {!}\tmsg{\Tag}{\tvarB}{\SessionTypeT}.\SessionEnd$ whose dual,
  ${?}\tmsg{\Tag}{\tvarB}{\SessionTypeT}.\SessionEnd$, is clearly
  different from $\SessionTypeS =
  {?}\tmsg{\Tag}{\tvarB}{\SessionTypeS}.\SessionEnd$ (duality does not
  change the type of message arguments).

  To provide a syntactic definition of dual endpoint type, we use an
  \emph{inner substitution operator} $\psubst{\cdot}{\cdot}$ such that
  $\SessionTypeT\psubst{\SessionTypeS}{\tvar}$ denotes $\SessionTypeT$
  where every free occurrence of $\tvar$ \emph{within the prefixes of
    $\SessionTypeT$} has been replaced by $\SessionTypeS$. Free
  occurrences of $\tvar$ that do not occur within a prefix of
  $\SessionTypeT$ are not substituted.
  For example, we have
  $({!}\tmsg\Tag\tvarB\tvarA.\tvarA)\psubst{\SessionTypeS}{\tvarA} =
  {!}\tmsg\Tag\tvarB\SessionTypeS.\tvarA$.
  Then, the \emph{dual} of an endpoint type $\SessionTypeT$ is defined
  inductively on the structure of $\SessionTypeT$, thus:
\[
\begin{array}{rcl}
  \co{\SessionEnd} & = & \SessionEnd \\
  \co{\tvar} & = & \tvar \\
  \co{\trec\tvar.\SessionType} & = &
  \trec\tvar.\co{\SessionType\psubst{\trec\tvar.\SessionType}{\tvar}}
  \\
  \co{\InternalChoice{\tmsg{\Tag_i}{\tvar_i}{\Type_i}.\SessionTypeT_i}_{i\in
      I}} & = & \ExternalChoice{\tmsg{\Tag_i}{\tvar_i}{\Type_i}.\co{\SessionTypeT_i}}_{i\in
    I} \\
  \co{\ExternalChoice{\tmsg{\Tag_i}{\tvar_i}{\Type_i}.\SessionTypeT_i}_{i\in
      I}} & = & \InternalChoice{\tmsg{\Tag_i}{\tvar_i}{\Type_i}.\co{\SessionTypeT_i}}_{i\in
    I}
\end{array}
\]
} 

Here are some important facts about well-formed endpoint types and
duality:

\begin{proposition}
  \label{prop:wf_types}
  The following properties hold:
\begin{enumerate}[\em(1)]
\item $\co{\co{\SessionType}} = \SessionType$.

\item $\emptyset \wfdash \SessionType$ implies that $\SessionType
  \dual \co\SessionType$ and $\emptyset \wfdash \co\SessionType$.

\item $\BoundContext; \{ \tvar \} \wfdash \SessionTypeT$ and
  $\BoundContext \wfdash \SessionTypeS$ imply $\BoundContext \wfdash
  \SessionTypeT\subst\SessionTypeS\tvar$.

\item $\emptyset; \{\tvar\} \wfdash \SessionTypeT$ and $\emptyset
  \wfdash \SessionTypeS$ imply
  $\co{\SessionTypeT\subst{\SessionTypeS}{\tvar}} =
  \co{\SessionTypeT}\subst{\SessionTypeS}{\tvar}$.
\end{enumerate}
\end{proposition}

\REVISION{Item~(1) states that $\co{\,\cdot\,}$ is an involution.}
Item~(2) states that $\co\SessionTypeT$ is well formed and dual of
$\SessionTypeT$ when $\SessionTypeT$ is well formed. Item~(3) states
the expected property of well-formedness preservation under
substitution of well-formed endpoint types. Finally, item~(4) shows
that duality does not affect the inner variables of an endpoint type
and that, in fact, duality and substitution commute.

\begin{example}
  In Example~\ref{ex:contracts} we have defined the endpoint types
  $\TMapper(\tvarA, \tvarB)$ and $\TStream(\tvarA)$ denoting the
  \Code{exp<Mapper<$\tvarA$,$\tvarB$>:WAIT_ARG>} and
  \Code{exp<Stream<$\tvarA$>:START>} types in \Sing{}.
  The dual endpoint types of $\TMapper(\tvarA, \tvarB)$ and
  $\TStream(\tvarA)$ are
\[
\begin{array}{rcl}
  \co{\TMapper(\tvarA, \tvarB)} & = &
  \trec\tvarC.
  {!}\msg{\mathtt{Arg}}{\qlin~\tvarA}.
  {?}\msg{\mathtt{Res}}{\qlin~\tvarB}.
  \tvarC
  \\
  \co{\TStream(\tvarA)} & = &
  \trec\tvarC.(
  {?}\msg{\mathtt{Data}}{\qlin~\tvarA}.\tvarC
  +
  {?}\msg{\mathtt{Eos}}{}.\SessionEnd
  )
\end{array}
\]
and they denote the \Code{imp<Mapper<$\tvarA$,$\tvarB$>:WAIT_ARG>} and
\Code{imp<Stream<$\tvarA$>:START>} types in \Sing{}.
\eoe
\end{example}

\begin{example}[function types]
\label{ex:map_type}
While \Sing{} is a procedural language, our formalization is based on
a process algebra. Therefore, some \Sing{} entities like functions and
function types that are not directly representable must be encoded. A
function can be encoded as a process that waits for the arguments and
sends the result of the computation. Callers of the function will
therefore send the arguments and receive the result. Following this
intuition, the type
\[
  \Tmap(\tvarA,\tvarB)
  =
  {!}\msg{\mathtt{Arg}}{\qlin~\co{\TMapper(\tvarA,\tvarB)}}.
  {!}\msg{\mathtt{Arg}}{\qlin~\co{\TStream(\tvarA)}}.
  {!}\msg{\mathtt{Arg}}{\qlin~\TStream(\tvarB)}.
  {?}\msg{\mathtt{Res}}{}.
  \SessionEnd
\]
seems like a good candidate for denoting the type of \Code{map} in
Figure~\ref{fig:example}. This type allows a caller of the function to
supply (send) three arguments having type
$\co{\TMapper(\tvarA,\tvarB)}$, $\co{\TStream(\tvarA)}$, and
$\TStream(\tvarB)$ in this order. The $\qlin$ qualifiers indicates
that all the arguments are linear.
Since \Code{map} returns nothing, the $\mathtt{Res}$-tagged message
does not carry any useful value, but it models the synchronous
semantics of function invocation.

This encoding of the type of \Code{map} does not distinguish arguments
that are \emph{claimed} by \Code{map} from others that are not. The
use of the $\qlin$ qualifier in the encoding is mandated by the fact
that the arguments are allocated in the exchange heap, but in this way
the caller process permanently loses the ownership of the
\Code{mapper} argument, and this is not the intended semantics of
\Code{map}.
We can model the temporary ownership transfer as a pair of linear
communications, by letting the (encoded) \Code{map} function return
any argument that is not claimed. Therefore, we patch the above
endpoint type as follows:
\[
  \Tmap(\tvarA,\tvarB)
  =
  {!}\msg{\mathtt{Arg}}{\qlin~\co{\TMapper(\tvarA,\tvarB)}}.
  [{\cdots}].
  {?}\msg{\mathtt{Arg}}{\qlin~\co{\TMapper(\tvarA,\tvarB)}}.
  {?}\msg{\mathtt{Res}}{}.
  \SessionEnd
\]

The endpoint type $\Tmap(\tvarA,\tvarB)$ describes the protocol for
one particular invocation of the \Code{map} function.
A proper encoding of the type of \Code{map}, which allows for multiple
invocations and avoids interferences between independent invocations,
is the following:
\[
\Type_{\mathtt{map}}
=
  \qun
  ~
  \trec\tvarC.
  {!}\xmsg{\mathtt{Invoke}}{\tvarA,\tvarB}{\qlin~\co{\Tmap(\tvarA,\tvarB)}}.
  \tvarC
\]
Prior to invocation, a caller is supposed to create a fresh channel
which is used for communicating with the process modeling the
function. One endpoint, of type $\Tmap(\tvarA,\tvarB)$, is retained by
the caller, the other one, of type $\co{\Tmap(\tvarA,\tvarB)}$, is
sent upon invocation to the process modeling \Code{map}.
The recursion in $\Type_{\mathtt{map}}$ permits multiple invocation of
\Code{map}, and the $\qun$ qualifier indicates that \Code{map} is
\emph{unrestricted} and can be invoked simultaneously and
independently by multiple processes in the system.
\eoe
\end{example}

The most common way to increase flexibility of a type system is to
introduce a \emph{subtyping} relation $\subt$ that establishes an
(asymmetric) compatibility between different types: any value of type
$\TypeT$ can be safely used where a value of type $\TypeS$ is expected
when $\TypeT \subt \TypeS$.
In the flourishing literature on session types several notions of
subtyping have been put
forward~\cite{GayHole05,Gay08,CastagnaDezaniGiachinoPadovani09,Vasconcelos09,Padovani09}.
We define subtyping in pretty much the same way as
in~\cite{GayHole05,Gay08}.

\begin{definition}[subtyping]
\label{def:subt}
Let $\leq$ be the least preorder on qualifiers such that $\qun \leq
\qlin$.
We say that $\srel$ is a \emph{coinductive subtyping}
if:
\begin{iteMize}{$\bullet$}
\item $(\Qualifier~\SessionTypeT, \Qualifier'~\SessionTypeS) \in
  {\srel}$ implies $\Qualifier \leq \Qualifier'$ and $(\SessionTypeT,
  \SessionTypeS) \in {\srel}$, and

\item $(\SessionTypeT, \SessionTypeS) \in {\srel}$ implies either:
\begin{enumerate}[(1)]
\item $\SessionTypeT = \SessionTypeS = \SessionEnd$, or

\item $\SessionTypeT = \SessionTypeS = \tvar$, or

\item $\SessionTypeT =
  \ExternalChoice{\tmsg{\Tag_i}{\tvar_i}{\TypeT_i}.\SessionTypeT_i}_{i\in
    I}$ and $\SessionTypeS =
  \ExternalChoice{\tmsg{\Tag_i}{\tvar_i}{\TypeS_i}.\SessionTypeS_i}_{i\in
    J}$ with $I\subseteq J$ and $(\TypeT_i, \TypeS_i) \in {\srel}$ and
  $(\SessionTypeT_i, \SessionTypeS_i) \in {\srel}$ for every $i\in I$,
  or

\item $\SessionTypeT =
  \InternalChoice{\tmsg{\Tag_i}{\tvar_i}{\TypeT_i}.\SessionTypeT_i}_{i\in
    I}$ and $\SessionTypeS =
  \InternalChoice{\tmsg{\Tag_i}{\tvar_i}{\TypeS_i}.\SessionTypeS_i}_{i\in
    J}$ with $J\subseteq I$ and $(\TypeS_i, \TypeT_i) \in {\srel}$ and
  $(\SessionTypeT_i, \SessionTypeS_i) \in {\srel}$ for every $i\in J$.
\end{enumerate}
\end{iteMize}
We write $\subt$ for the largest coinductive subtyping.
\end{definition}

Items~(1) and~(2) account for reflexivity of subtyping when
$\SessionTypeT$ and $\SessionTypeS$ are both $\SessionEnd$ or the same
type variable;
items~(3) and~(4) are the usual covariant and contravariant rules for
inputs and outputs respectively. Observe that subtyping is always
covariant with respect to the continuations.
Two types $\Qualifier_1~\SessionTypeT$ and
$\Qualifier_2~\SessionTypeS$ are related by subtyping if so are
$\SessionTypeT$ and $\SessionTypeS$ and if $\Qualifier_1$ is no more
stringent than $\Qualifier_2$. In particular, it is safe to use an
unrestricted value where a linear one is expected.

The reader may verify that subtyping is a pre-order:

\begin{proposition}
\label{prop:subt_transitive}
$\subt$ is reflexive and transitive.
\end{proposition}
\begin{proof}[Proof sketch]
  The proofs of both properties are easy exercises. In the case of
  transitivity it suffices to show that
\[
{\srel}
\eqdef
\{ (\TypeT_1, \Type_2) \mid \exists \TypeS: \TypeT_1 \subt \TypeS \land \TypeS \subt \TypeT_2 \}
\cup
\{ (\SessionTypeT_1, \SessionType_2) \mid \exists \SessionTypeS: \SessionTypeT_1 \subt \SessionTypeS \land \SessionTypeS \subt \SessionTypeT_2 \}
\]
is a coinductive subtyping.
\end{proof}

%
The following property shows that duality is contravariant with
respect to subtyping. It is a standard property of session type
theories, except that in our case it holds only when the two endpoint
types being related have no free type variables occurring at the top
level (outside any prefix), for otherwise their duals are undefined
(Proposition~\ref{prop:wf_types}).

\begin{proposition}
\label{prop:subt_dual}
Let $\emptyset \wfdash \SessionTypeT$ and $\emptyset \wfdash
\SessionTypeS$. Then $\SessionTypeT \subt \SessionTypeS$ if and only
if $\co\SessionTypeS \subt \co\SessionTypeT$.
\end{proposition}

\begin{example}
\label{ex:subtyping}
\newcommand{\enc}[1]{\llbracket#1\rrbracket} In
Example~\ref{ex:map_type} we have suggested a representation for the
function type $\TypeS \to \TypeT$ as the type $\enc{\TypeS \to
  \TypeT}$ defined thus:
\[
  \enc{\TypeS \to \TypeT}
  =
  \qun~\trec\tvar.{!}\msg{\mathtt{Invoke}}{
    \qlin~{?}\msg{\mathtt{Arg}}{\TypeS}.
    {!}\msg{\mathtt{Res}}{\TypeT}.
    \SessionEnd
  }.\tvar
\]
It is easy to verify that $\enc{\TypeS_1 \to \TypeT_1} \subt
\enc{\TypeS_2 \to \TypeT_2}$ if and only if $\TypeS_2 \subt \TypeS_1$
and $\TypeT_1 \subt \TypeT_2$. That is, the subtyping relation between
encoded function types is consistent with the standard subtyping
between function types, which is contravariant in the domain and
covariant in the co-domain.

Another way to interpret an endpoint having type
\[
  \trec\tvar.{!}\msg{\mathtt{Invoke}}{\Type}.\tvar
\]
is as an object with one method $\mathtt{Invoke}$. Sending a
$\mathtt{Invoke}$-tagged message on the endpoint means invoking the
method (incidentally, this is the terminology adopted in SmallTalk),
and after the invocation the object is available again with the same
interface. We can generalize the type above to
\[
  \trec\tvar.\InternalChoice{\msg{\Tag_i}{\Type_i}.\tvar}_{i\in I}
\]
for representing objects with multiple methods $\Tag_i$. According to
the definition of subtyping we have
\[
  \trec\tvar.\InternalChoice{\msg{\Tag_i}{\Type_i}.\tvar}_{i\in I}
  \subt
  \trec\tvar.\InternalChoice{\msg{\Tag_j}{\Type_j}.\tvar}_{j\in J}
\]
whenever $J \subseteq I$, which corresponds the same notion of
subtyping used in object-oriented language (it is safe to use an
object offering more methods where one offering fewer methods is
expected).
\eoe
\end{example}


\section{Syntax and Semantics of Processes}
\label{sec:processes}

We assume the existence of an infinite set $\PointerSet$ of
\emph{linear pointers} (or simply \emph{pointers}) ranged over by
$\PointerA$, $\PointerB$, $\dots$, of an infinite set $\VarSet$ of
\emph{variables} ranged over by $\VarX$, $\VarY$, $\dots$, and of an
infinite set of \emph{process variables} ranged over by $\RecVarX$,
$\RecVarY$, $\dots$.
We define the set $\shared\PointerSet$ of \emph{unrestricted pointers}
as $\shared\PointerSet = \{ \shared\Pointer \mid \Pointer \in
\PointerSet \}$.
We assume $\PointerSet$, $\shared\PointerSet$, and $\VarSet$ be
pairwise disjoint, we let $\NameA$, $\NameB$, $\dots$ range over
\emph{names}, which are elements of $\PointerSet \cup
\shared\PointerSet \cup \VarSet$, and we let $\ValueV$, $\ValueW$,
$\dots$ range over \emph{values}, which are elements of $\PointerSet
\cup \shared\PointerSet$.
%

Processes, ranged over by $\ProcessP$, $\ProcessQ$, $\dots$, are
defined by the grammar in Table~\ref{tab:syntax}.
The calculus of processes is basically a monadic pi calculus equipped
with tag-based message dispatching and primitives for handling
heap-allocated endpoints. The crucial aspect of the calculus is that
names are pointers to the heap and channels are concretely represented
as structures allocated on the heap.
Pointers can be either linear or unrestricted: a linear pointer must
be owned by exactly one process at any given point in time; an
unrestricted pointer can be owned by several (possibly zero) processes
at any time.  In practice the two kinds of pointers are
indistinguishable and range over the same address space, but in the
calculus we decorate unrestricted pointers with a bar to reason
formally on the different ownership invariants.
The term $\idle$ denotes the idle process that performs no action.
The term $\openChannel(\PointerA:\SessionTypeT,
\PointerB:\SessionTypeS).\Process$ denotes a process that creates a
\emph{linear channel}, represented as a pair of endpoints $\PointerA$
of type $\SessionTypeT$ and $\PointerB$ of type $\SessionTypeS$, and
continues as $\Process$. We will say that $\PointerB$ is the peer
endpoint of $\PointerA$ and vice-versa.
The term $\openChannel(\Pointer:\SessionTypeT).\Process$ denotes a
process that creates an \emph{unrestricted channel}, represented as an
endpoint $\Pointer$ of type $\SessionTypeT$ along with an unrestricted
pointer $\shared\Pointer$ of type $\co\SessionTypeT$, and continues as
$\Process$.
The term $\closeChannel(\Name)$ denotes a process closing and
deallocating the endpoint $\Name$.
The term $\xsend\NameA\Tag\SessionType\NameB.\Process$ denotes a
process that sends a message $\xmsg\Tag\SessionType\NameB$ on the
endpoint $\NameA$ and continues as $\Process$. The message is made of
a \emph{tag} $\Tag$ along with its \emph{parameter} $\NameB$. The
endpoint type $\SessionType$ instantiates the type variable in the
type of $\NameA$. For consistency with the type language we only
consider monadic communications where every message has exactly one
type/value parameter. The generalization to polyadic communications,
which we will occasionally use in the examples, does not pose
substantial problems.
The term $\sum_{i\in I}
\xreceive\Name{\Tag_i}{\tvar_i}{\Var_i:\Type_i}.\Process_i$ denotes a
process that waits for a message from the endpoint $\Name$. The tag
$\Tag_i$ of the received message determines the continuation
$\Process_i$ where the variable $\Var_i$ is instantiated with the
parameter of the message.
Sometimes we will write
$\xreceive\Name{\Tag_1}{\tvar_1}{\Var_1:\Type_1}.\Process_1 + \cdots +
\xreceive\Name{\Tag_n}{\tvar_n}{\Var_n:\Type_n}.\Process_n$ in place
of $\sum_{i=1}^n
\xreceive\Name{\Tag_i}{\tvar_i}{\Var_i:\Type_i}.\Process_i$.\footnote{\REVISION{We
    require the endpoint $\Name$ to be the same in all branches of the
    receive construct, while \Code{switch receive} in \Sing{} allows
    waiting for messages coming from \emph{different} endpoints. This
    generalization would not affect our formalization in any
    substantial way, save for slightly more complicated typing
    rules.}}
The term $\ProcessP\choice\ProcessQ$ denotes a process that internally
decides whether to behave as $\ProcessP$ or as $\ProcessQ$. We do not
specify the actual condition that determines the decision, as this is
irrelevant for our purposes. To improve readability, in some of the
examples we will use a more concrete syntax.
As usual, terms $\rec\RecVar.\Process$ and $\RecVar$ serve to denote
recursive processes, while $\ProcessP \parop \ProcessQ$ denotes the
parallel composition of $\ProcessP$ and $\ProcessQ$.

\begin{table}
  \caption{\label{tab:syntax}\strut Syntax of processes.}
\framebox[\textwidth]{
\begin{math}
\displaystyle
\begin{array}[t]{@{}r@{\quad}rcl@{\quad}l@{}}
  \textbf{Process} & \Process
    & ::= & \idle & \text{(idle)} \\
  & & | & \closeChannel(\Name) & \text{(close endpoint)} \\
  & & | & \openChannel(\Pointer:\SessionType, \Pointer:\SessionType).\Process & \text{(open linear channel)} \\
  & & | & \openChannel(\Pointer:\SessionType).\Process & \text{(open unrestricted channel)} \\
  & & | & \xsend\Name\Tag\SessionType\Name.\Process & \text{(send)} \\
  & & | & \sum_{i\in I} \xreceive\Name{\Tag_i}{\tvar_i}{\Var_i:\Type_i}.\Process_i & \text{(receive)} \\
  & & | & \Process\choice\Process & \text{(conditional process)} \\
  & & | & \Process \parop \Process & \text{(parallel composition)} \\
  & & | & \RecVar & \text{(process variable)} \\
  & & | & \rec\RecVar.\Process & \text{(recursive process)} \\
\end{array}
\end{math}
}
\end{table}

Table~\ref{tab:names} collects the definitions of \emph{free names}
$\fn(\cdot)$ and \emph{bound names} $\bn(\cdot)$ for processes.
Beware that a process $\openChannel(\Pointer:\SessionType).\Process$
implicitly binds $\shared\Pointer$ in addition to $\Pointer$ in
$\Process$.
In the same table we also define the sets of \emph{free type
  variables} $\ftv(\cdot)$ and of \emph{bound type variables} of a
process. Note that the set of bound type variables only includes those
variables occurring in input prefixes of the process, not the type
variables bound within endpoint types occurring in the process.
The construct $\rec\RecVar.\Process$ is the only binder for process
variables. The sets of \emph{free process variables} $\fpv(\cdot)$ and
of \emph{bound process variables} $\bpv(\cdot)$ are standard.
We identify processes up to renaming of bound names/type
variables/process variables and let $\Process\subst{\Value}{\Var}$,
$\Process\subst{\SessionType}{\tvar}$, and
$\Process\subst{\ProcessQ}{\RecVar}$ denote the standard
capture-avoiding substitutions of variables/type variables/process
variables with values/endpoint types/processes.


\begin{table}
  \caption{\label{tab:names}\strut Free and bound names/type variables in processes.}
\framebox[\textwidth]{
\begin{math}
\displaystyle
\begin{array}{r@{~}c@{~}l}
  \fn(\idle) = \fn(\RecVar) & = & \emptyset \\
  \fn(\closeChannel(\Name)) & = & \{ \Name \} \\
  \fn(\openChannel(\PointerA:\SessionTypeT, \PointerB:\SessionTypeS).\Process) & = &
  \fn(\Process) \setminus \{ \PointerA, \PointerB \} \\
  \fn(\openChannel(\Pointer:\SessionTypeT).\Process) & = &
  \fn(\Process) \setminus \{ \Pointer, \shared\Pointer \} \\
  \fn(\xsend\NameA\Tag\SessionType\NameB.\Process) & = &
  \{ \NameA, \NameB \} \cup \fn(\Process) \\
  \fn(\sum_{i\in I} \xreceive\Name{\Tag_i}{\tvar_i}{\Var_i:\Type_i}.\Process_i) & = &
  \{ \Name \} \cup \bigcup_{i\in I} (\fn(\Process_i) \setminus \{ \Var_i \}) \\
  \fn(\ProcessP\choice\ProcessQ) =
  \fn(\ProcessP \parop \ProcessQ) & = & \fn(\ProcessP) \cup \fn(\ProcessQ) \\
  \fn(\rec\RecVar.\Process) & = & \fn(\Process) \\
  \\
  \bn(\idle) = \bn(\closeChannel(\Name)) = \bn(\RecVar) & = & \emptyset \\
  \bn(\openChannel(\PointerA:\SessionTypeT, \PointerB:\SessionTypeS).\Process) & = &
  \{ \PointerA, \PointerB \} \cup \bn(\Process) \\
  \bn(\openChannel(\Pointer:\SessionTypeT).\Process) & = &
  \{ \Pointer, \shared\Pointer \} \cup \bn(\Process) \\
  \bn(\xsend\NameA\Tag\SessionType\NameB.\Process) = \bn(\rec\RecVar.\Process) & = &
  \bn(\Process) \\
  \bn(\sum_{i\in I} \xreceive\Name{\Tag_i}{\tvar_i}{\Var_i:\Type_i}.\Process_i) & = &
  \bigcup_{i\in I} (\{\Var_i\} \cup \bn(\Process_i)) \\
  \bn(\ProcessP \choice \ProcessQ) = \bn(\ProcessP \parop \ProcessQ) & = &
  \bn(\ProcessP) \cup \bn(\ProcessQ) \\
  \\
  \ftv(\idle) = \ftv(\closeChannel(\Name)) = \ftv(\RecVar) & = & \emptyset \\
  \ftv(\openChannel(\PointerA:\SessionTypeT, \PointerB:\SessionTypeS).\Process) & = &
  \ftv(\SessionTypeT) \cup \ftv(\SessionTypeS) \cup \ftv(\Process) \\
  \ftv(\openChannel(\Pointer:\SessionTypeT).\Process) =
  \ftv(\xsend\NameA\Tag\SessionType\NameB.\Process) & = &
  \ftv(\SessionType) \cup \ftv(\Process) \\
  \ftv(\sum_{i\in I} \xreceive\Name{\Tag_i}{\tvar_i}{\Var_i:\Type_i}.\Process_i) & = &
  \bigcup_{i\in I} ((\ftv(\Type_i) \cup \ftv(\Process_i)) \setminus \{ \tvar_i \}) \\
  \ftv(\ProcessP\choice\ProcessQ) =
  \ftv(\ProcessP \parop \ProcessQ) & = & \ftv(\ProcessP) \cup \ftv(\ProcessQ) \\
  \ftv(\rec\RecVar.\Process) & = & \ftv(\Process) \\
  \\
  \btv(\idle) = \btv(\closeChannel(\Name)) = \btv(\RecVar) & = & \emptyset \\
  \btv(\openChannel(\PointerA:\SessionTypeT, \PointerB:\SessionTypeS).\Process) =
  \btv(\openChannel(\Pointer:\SessionTypeT).\Process) & = &
  \btv(\Process) \\
  \btv(\xsend\NameA\Tag\SessionType\NameB.\Process) = \btv(\rec\RecVar.\Process) & = &
  \btv(\Process) \\
  \btv(\sum_{i\in I} \xreceive\Name{\Tag_i}{\tvar_i}{\Var_i:\Type_i}.\Process_i) & = &
  \bigcup_{i\in I} (\{\tvar_i\} \cup \btv(\Process_i)) \\
  \btv(\ProcessP \choice \ProcessQ) = \btv(\ProcessP \parop \ProcessQ) & = &
  \btv(\ProcessP) \cup \btv(\ProcessQ) \\
\end{array}
\end{math}
}
\end{table}

\newcommand{\ESource}{\mathit{source}}
\newcommand{\ETarget}{\mathit{target}}
\newcommand{\EMapper}{\mathit{mapper}}

\begin{example}
\label{ex:map_process}
\newcommand{\TagInvoke}{\Code{Invoke}}
Let us encode the \Code{map} function in Figure~\ref{fig:example}
using the syntax of our process calculus. As anticipated in
Example~\ref{ex:map_type}, the idea is to represent \Code{map} as a
process that permanently accepts invocations and handles them. For
this reason we need an endpoint, say $\PointerC$, to which invocation
requests are sent and we define the $\mathrm{MAP}(\PointerC)$ process
thus:
\[
\begin{array}{@{}rcl@{}}
  \mathrm{MAP}(\PointerC) & = &
  \rec\RecVarX.
  \xreceive\PointerC{\mathtt{Invoke}}{\tvarA,\tvarB}{\VarZ : \qlin~\co{\Tmap(\tvarA,\tvarB)}}.
  (\RecVarX \parop \mathrm{BODY}(\tvarA,\tvarB,\VarZ))
  \\
  \mathrm{BODY}(\tvarA,\tvarB,\VarZ) & = &
  \receive\VarZ{\mathtt{Arg}}{\EMapper:\qlin~\co{\TMapper(\tvarA,\tvarB)}}. \\ & &
  \receive\VarZ{\mathtt{Arg}}{\ESource:\qlin~\co{\TStream(\tvarA)}}. \\ & &
  \receive\VarZ{\mathtt{Arg}}{\ETarget:\qlin~\TStream(\tvarB)}. \\ & &
  \rec\RecVarY.
  \begin{array}[t]{@{}l@{}l@{}}
    ( & \receive\ESource{\mathtt{Data}}{\VarX : \qlin~\tvarA}.
    \send\EMapper{\mathtt{Arg}}\VarX. \\
    & \receive\EMapper{\mathtt{Res}}{\VarY : \qlin~\tvarB}.
    \send\ETarget{\mathtt{Data}}\VarY.
    \RecVarY
    \\
    + {\,} &
    \receive\ESource{\mathtt{Eos}}{}.
    \send\ETarget{\mathtt{Eos}}{}. \\
    & \send\VarZ{\mathtt{Arg}}\EMapper.
    \send\VarZ{\mathtt{Res}}{}. \\
    & (\closeChannel(\VarZ) \parop \closeChannel(\ESource) \parop \closeChannel(\ETarget)))
  \end{array}
\end{array}
\]
The process $\mathrm{MAP}(\PointerC)$ repeatedly reads
\TagInvoke-tagged messages from $\PointerC$. Each message carries
another endpoint $\VarZ$ that represents a private session established
between the caller and the callee, whose purpose is to make sure that
no interference occurs between independent invocations of the
service. Note that $\VarZ$ has type $\co{\Tmap(\tvarA,\tvarB)}$, the
dual of $\Tmap(\tvarA,\tvarB)$, since it is the endpoint handed over
by the caller from which the callee will \emph{receive} the arguments
and \emph{send} the result.
The body of the \Code{map} function is encoded by the
$\mathrm{BODY}(\tvarA,\tvarB,\VarZ)$ process, which begins by reading
the three arguments $\EMapper$, $\ESource$, and $\ETarget$.
Then, the process enters its main loop where messages are received
from $\ESource$, processed through $\EMapper$, and finally sent on
$\ETarget$. Overall the structure of the process closely follows that
of the code in Figure~\ref{fig:example}, where the branch operator is
used for modeling the \Code{switch receive} construct.
The only remarkable difference occurs after the input of a
\TagEos-tagged message, where the $\EMapper$ argument is returned to
the caller so as to model the temporary ownership transfer that was
implicitly indicated by the lack of the \Code{[Claims]} annotation in
\Code{map}.
At this point the $\VarZ$ endpoint serves no other purpose and is
closed along with $\ESource$ and $\ETarget$.
\eoe
\end{example}

\begin{table}
  \caption{\label{tab:runtime}\strut Syntax of heaps and queues.}
\framebox[\textwidth]{
\begin{math}
\displaystyle
\begin{array}{cc}
\begin{array}[t]{@{}r@{\quad}rcl@{\quad}l@{}}
  \textbf{Heap} & \Memory
    & ::= & \EmptyMemory & \text{(empty)} \\
  & & | & \Pointer \mapsto [\Pointer, \Queue] & \text{(endpoint)} \\
  & & | & \Memory, \Memory & \text{(composition)} \\
\end{array}
&
\begin{array}[t]{@{}r@{\quad}rcl@{\quad}l@{}}
  \textbf{Queue} & \Queue
    & ::= & \EmptyQueue & \text{(empty)} \\
  & & | & \xmsg\Tag\SessionType\Value & \text{(message)} \\
  & & | & \Queue :: \Queue & \text{(composition)} \\
\end{array}
\end{array}
\end{math}
}
\end{table}

To state the operational semantics of processes we need a formal
definition of the \emph{exchange heap} (or simply \emph{heap}), which
is given in Table~\ref{tab:runtime}.
\emph{Heaps}, ranged over by $\Memory$, are term representations of
finite maps from pointers to heap objects:
the term $\EmptyMemory$ denotes the empty heap, in which no object is
allocated;
the term $\PointerA \mapsto [\PointerB, \Queue]$ denotes a heap made
of an endpoint located at $\PointerA$. The endpoint is a structure
containing another pointer $\PointerB$ and a \emph{queue} $\Queue$ of
messages waiting to be read from $\PointerA$.
Heap compositions $\Memory,\Memory'$ are defined only when the domains
of the heaps being composed, which we denote by $\dom(\Memory)$ and
$\dom(\Memory')$, are disjoint.
We assume that heaps are equal up to commutativity and associativity
of composition and that $\EmptyMemory$ is neutral for composition.
\emph{Queues}, ranged over by $\Queue$, are finite ordered sequences
of messages $\xmsg{\Tag_1}{\SessionType_1}{\Value_1} :: \cdots ::
\xmsg{\Tag_n}{\SessionType_n}{\Value_n}$, where a message
$\xmsg\Tag\SessionType\Value$ is identified by its tag $\Tag$, the
endpoint type $\SessionType$ with which its type argument has been
instantiated, and its value argument $\Value$.  We build queues from
the empty queue $\EmptyQueue$ and concatenation of messages by means
of $::$. We assume that queues are equal up to associativity of $::$
and that $\EmptyQueue$ is neutral for $::$.
The $\SessionType$ component in the enqueued messages must be
understood as a technical annotation that helps reasoning on the
formal properties of the model. In particular, it does not imply that
a practical implementation of the calculus must necessarily provide a
runtime representation of endpoint types.\footnote{\Sing{} \emph{does}
  require a runtime representation of endpoint types because its
  expression language is equipped with a dynamic cast operator.}

We define the operational semantics of processes as the combination of
a structural congruence relation, which equates processes we do not
want to distinguish, and a reduction relation.
Structural congruence, denoted by $\equiv$, is the least congruence
relation defined by the axioms in Table~\ref{tab:cong} and closed
under parallel composition.
Essentially, the axioms state that~$\parop$ is commutative,
associative, and has $\idle$ as neutral element.

\begin{table}
\caption{\label{tab:cong}\strut Structural congruence.}
\framebox[\textwidth]{
\begin{math}
\displaystyle
\begin{array}{c}
\inferrule[\rulename{S-Idle}]{}{
  \Process \parop \idle \equiv \Process
}
\qquad
\inferrule[\rulename{S-Comm}]{}{
  \ProcessP \parop \ProcessQ \equiv \ProcessQ \parop \ProcessP
}
\qquad
\inferrule[\rulename{S-Assoc}]{}{
  \ProcessP \parop (\ProcessQ \parop \ProcessR)
  \equiv
  (\ProcessP \parop \ProcessQ) \parop \ProcessR
}
\end{array}
\end{math}
}
\end{table}

\begin{table}
\caption{\label{tab:reduction}\strut Operational semantics of processes.}
\framebox[\textwidth]{
\begin{math}
\displaystyle
\begin{array}{@{}c@{}}
\inferrule[\rulename{R-Open Linear Channel}]{}{
  \system\Memory{\openChannel(\PointerA:\SessionTypeT,\PointerB:\SessionTypeS).\Process}
  \red{}
  \system{
    \Memory,
    \PointerA\mapsto[\PointerB,\EmptyQueue],
    \PointerB\mapsto[\PointerA,\EmptyQueue]
  }{
    \Process
  }
}
\\\\
\inferrule[\rulename{R-Open Unrestricted Channel}]{}{
  \system\Memory{\openChannel(\Pointer:\SessionType).\Process}
  \red{}
  \system{
    \Memory,
    \Pointer\mapsto[\Pointer,\EmptyQueue]
  }{
    \Process
  }
}
\qquad
\inferrule[\rulename{R-Choice Left}]{}{
  \system\Memory{\ProcessP \choice \ProcessQ}
  \red{}
  \system\Memory\ProcessP
}
\qquad
\inferrule[\rulename{R-Choice Right}]{}{
  \system\Memory{\ProcessP \choice \ProcessQ}
  \red{}
  \system\Memory\ProcessQ
}
\\\\
\inferrule[\rulename{R-Send Linear}]{}{
  \system{
    \Memory,
    \PointerA \mapsto [\PointerB, \Queue],
    \PointerB \mapsto [\PointerA, \Queue']
  }{
    \xsend\PointerA\Tag\SessionType\Value.\Process
  }
  \red{}
  \system{
    \Memory,
    \PointerA \mapsto [\PointerB, \Queue],
    \PointerB \mapsto [\PointerA, \Queue'::\xmsg\Tag\SessionType\Value]
  }{
    \Process
  }
}
\\\\
\inferrule[\rulename{R-Send Unrestricted}]{}{
  \system{
    \Memory,
    \Pointer \mapsto [\Pointer, \Queue]
  }{
    \xsend{\shared\Pointer}\Tag\SessionType\Value.\Process
  }
  \red{}
  \system{
    \Memory,
    \Pointer \mapsto [\Pointer, \Queue::\xmsg\Tag\SessionType\Value]
  }{
    \Process
  }
}
\\\\
\inferrule[\rulename{R-Receive}]{
  k\in I
}{
  \textstyle
  \system{
    \Memory,
    \PointerA \mapsto [\PointerB, \xmsg{\Tag_k}\SessionType\Value::\Queue]
  }{
    \sum_{i\in I} \xreceive\PointerA{\Tag_i}{\tvar_i}{\Var_i:\Type_i}.\Process_i
  }
  \red{}
  \system{
    \Memory,
    \PointerA \mapsto [\PointerB, \Queue]
  }{
    \Process_k\subst{\SessionType}{\tvar_k}\subst{\Value}{\Var_k}
  }
}
\\\\
\inferrule[\rulename{R-Rec}]{}{
  \system\Memory{\rec\RecVar.\Process}
  \red{}
  \system\Memory{\Process\subst{\rec\RecVar.\Process}\RecVar}
}
\\\\
\inferrule[\rulename{R-Par}]{
  \system\Memory\ProcessP \red{} \system{\Memory'}{\ProcessP'}
}{
  \system\Memory{\ProcessP \parop \ProcessQ}
  \red{}
  \system{\Memory'}{\ProcessP' \parop \ProcessQ}
}
\qquad
\inferrule[\rulename{R-Struct}]{
  \ProcessP \equiv \ProcessP'
  \\
  \system\Memory{\ProcessP'} \red{} \system{\Memory'}{\ProcessQ'}
  \\
  \ProcessQ' \equiv \ProcessQ
}{
  \system\Memory\ProcessP \red{} \system{\Memory'}\ProcessQ
}
\end{array}
\end{math}
}
\end{table}

Processes communicate by means of endpoints that are allocated on the
heap. Consequently, the reduction relation defines the transitions of
\emph{systems} rather than of processes, where a system is a pair
$\system\Memory\Process$ of a heap $\Memory$ and a process $\Process$.
The reduction relation $\red{}$ is inductively defined in
Table~\ref{tab:reduction}; we comment on the rules in the following
paragraphs.
Rule~\rulename{R-Open Linear Channel} creates a new linear channel,
which consists of two fresh endpoints with empty queues and mutually
referring to each other. The mutual references are needed since the
messages sent using one of the endpoints will be enqueued into the
other peer.
Rule~\rulename{R-Open Unrestricted Channel} creates a new unrestricted
channel, which consists of a single endpoint with empty queue. The
reference in the endpoint is initialized with a pointer to itself.
This way, by inspecting the $\PointerB$ component of an endpoint
$\PointerA \mapsto [\PointerB, \Queue]$ it is possible to understand
whether the endpoint belongs to a linear or to an unrestricted
channel, as we respectively have either $\PointerA \ne \PointerB$ or
$\PointerA = \PointerB$. This distinction is necessary in the
reductions defining the semantics of outputs, as we will see shortly.
In both~\rulename{R-Open Linear Channel} and~\rulename{R-Open
  Unrestricted Channel} we implicitly rename bound names to make sure
that the newly introduced pointers do not already occur in
$\dom(\Memory)$, for otherwise the heap in the resulting system would
be undefined.
Rules~\rulename{R-Choice Left} and~\rulename{R-Choice Right} describe
the standard reduction of conditional processes.
Rules~\rulename{R-Send Linear} and~\rulename{R-Send Unrestricted}
describe the output of a message $\xmsg\Tag\SessionType\Value$ on the
endpoint $\PointerA$ of a linear channel and on the endpoint
$\shared\PointerA$ of an unrestricted channel, respectively. In the
former case, the message is enqueued at the end of $\PointerA$'s peer
endpoint queue. In the latter case, the message is enqueued in the
only available queue.
Rule~\rulename{R-Receive} describes the input of a message from the
endpoint $\PointerA$. The message at the front of $\PointerA$'s queue
is removed from the queue, its tag is used for selecting some branch
$k\in I$, and its type and value arguments instantiate the type
variable $\tvar_k$ and variable $\Var_k$. If the queue is not empty
and the first message in the queue does not match any of the tags $\{
\Tag_i \mid i \in I\}$, then no reduction occurs and the process is
stuck.
Rule~\rulename{R-Rec} describes the usual unfolding of a recursive
process.
Rule~\rulename{R-Par} closes reductions under parallel
composition. Observe that the heap is treated globally, even when it
is only a sub-process to reduce.
Finally, rule~\rulename{R-Struct} describes reductions modulo
structural congruence. 
%
There is no reduction for $\closeChannel(\PointerA)$ processes. In
principle, $\closeChannel(\PointerA)$ should deallocate the endpoint
located at $\PointerA$ and remove the association for $\PointerA$ from
the heap. In the formal model it is technically convenient to treat
$\closeChannel(\PointerA)$ processes as persistent because, in this
way, we keep track of the pointers that have been properly
deallocated. We will see that this information is crucial in the
definition of well-behaved processes (Definition~\ref{def:wb}).  A
process willing to deallocate a pointer $\PointerA$ and to continue as
$\Process$ afterwards can be modeled as
$\closeChannel(\PointerA) \parop \Process$.
In the following we write $\wred{}$ for the reflexive, transitive
closure of $\red{}$ and we write $\system\Memory\Process \nred{}$ if
there exist no $\Memory'$ and $\Process'$ such that
$\system\Memory\Process \red{} \system{\Memory'}{\Process'}$.

In this work we characterize well-behaved systems as those that are
free from faults, leaks, and communication errors: a \emph{fault} is
an attempt to use a pointer not corresponding to an allocated object
or to use a pointer in some way which is not allowed by the object it
refers to; a \emph{leak} is a region of the heap that some process
allocates and that becomes unreachable because no reference to it is
directly or indirectly available to the processes in the system; a
\emph{communication error} occurs if some process receives a message
of unexpected type.
We conclude this section formalizing these properties.  To do so, we
need to define the reachability of a heap object with respect to a set
of \emph{root} pointers. Intuitively, a process $\Process$ may
directly reach any object located at some pointer in the set
$\fn(\Process)$ (we can think of the pointers in $\fn(\Process)$ as of
the local variables of the process stored on its stack); from these
pointers, the process may reach other heap objects by reading messages
from the endpoints it can reach, and so forth.

\newcommand{\Pointers}{A}

\newcommand{\xreach}[1][]{\prec_{#1}}
\newcommand{\reach}[1][]{\preccurlyeq_{#1}}

\begin{definition}[reachable pointers]
\label{def:reachable}
We say that $\PointerC$ is \emph{reachable} from $\PointerA$ in
$\Memory$, notation $\PointerC \xreach[\Memory] \PointerA$, if
$\PointerA \mapsto [\PointerB, \Queue ::
\xmsg\Tag\SessionType\PointerC :: \Queue'] \in \Memory$.
We write $\reach[\Memory]$ for the reflexive, transitive closure of
$\xreach[\Memory]{}{}$.
Let $\reachable{\Pointers}{\Memory} = \{ \PointerC \in \PointerSet
\mid \exists\PointerA\in\Pointers: \PointerC \reach[\Memory] \PointerA
\}$.
\end{definition}

Observe that $\reachable\Pointers\Memory \subseteq \PointerSet$ for
every $\Pointers \subseteq \PointerSet \cup \shared\PointerSet$ and
$\Memory$. Also, according to this definition nothing is reachable
from an unrestricted pointer. The rationale is that we will use
$\reachable\cdot\cdot$ only to define the ownership invariant, for
which the only pointers that matter are the linear ones.
We now define well-behaved systems formally.

\begin{definition}[well-behaved process]
\label{def:wb}
We say that $\ProcessP$ is \emph{well behaved} if
$\system\EmptyMemory\ProcessP \wred{} \system\Memory{\ProcessQ}$
implies:
\begin{enumerate}[(1)]
\item $\dom(\Memory) = \reachable{\fn(\ProcessQ)}\Memory$;

\item $\ProcessQ \equiv \Process_1 \parop \Process_2$ implies
  $\reachable{\fn(\Process_1)}\Memory \cap
  \reachable{\fn(\Process_2)}\Memory = \emptyset$;

\item $\ProcessQ \equiv \Process_1 \parop \Process_2$ and
  $\system\Memory{\Process_1} \nred{}$ where $\Process_1$ does not
  have unguarded parallel compositions imply either $\Process_1 =
  \idle$ or $\Process_1 = \closeChannel(\Pointer)$ or $\Process_1 =
  \sum_{i\in I}
  \xreceive\Pointer{\Tag_i}{\tvar_i}{\Var_i:\Type_i}.\Process_i$ and,
  \REVISION{in the last two cases}, $\Pointer \mapsto [\PointerB,
  \EmptyQueue] \in \Memory$.
\end{enumerate}
\end{definition}

In words, a process $\ProcessP$ is well behaved if every residual of
$\ProcessP$ reachable from a configuration where the heap is empty
satisfies a number of conditions.
Conditions~(1) and~(2) require the absence of faults and leaks.
Indeed, condition~(1) states that every allocated pointer in the heap
is reachable by one process, and that every reachable pointer
corresponds to an object allocated in the heap.
Condition~(2) states that processes are isolated, namely that no
linear pointer is reachable from two or more distinct
processes. Because of the definition of reachable pointers, though, it
may be possible that two or more processes share the same unrestricted
pointer. Since processes of the form $\closeChannel(\Pointer)$ are
persistent, this condition also requires the absence of faults
deriving from multiple deallocations of the same endpoint or from the
use of deallocated endpoints.
Condition~(3) requires the absence of communication errors, namely
that if $\system\Memory\ProcessQ$ is stuck (no reduction is possible),
then it is because every non-terminated process in $\ProcessQ$ is
waiting for a message on an endpoint having an empty queue. This
configuration corresponds to a genuine deadlock where every process in
some set is waiting for a message that is to be sent by another
process in the same set.
\REVISION{Condition~(3) also ensures the absence of so-called
  \emph{orphan messages}: no message accumulates in the queue of
  closed endpoints.}
We only consider initial configurations with an empty heap for two
reasons: first, we take the point of view that initially there are no
allocated objects; second, since we will need a well-typed predicate
for heaps and we do not want to verify heap well-typedness at runtime,
we will make sure that the empty heap is trivially well typed.

We conclude this section with a few examples of ill-behaved processes
to illustrate the sort of errors we aim to avoid with our static type
system:
\begin{iteMize}{$\bullet$}
\item The process $\openChannel(\PointerA : \SessionTypeT, \PointerB :
  \SessionTypeS).\idle$ violates condition~(1), since it allocates two
  endpoints $\PointerA$ and $\PointerB$ and forgets them, thus
  generating a leak.

\item The process $\openChannel(\PointerA : \SessionTypeT, \PointerB :
  \SessionTypeS).(\closeChannel(\PointerA) \parop
  \closeChannel(\PointerA) \parop \closeChannel(\PointerB))$ violates
  condition~(2), since it deallocates the same endpoint $\PointerA$
  twice. This is an example of fault.

\item The process $\openChannel(\PointerA : \SessionTypeT, \PointerB :
  \SessionTypeS).(\send\PointerA{\Tag}{}.\closeChannel(\PointerA)
  \parop \receive\PointerB{\Tag'}{}.\closeChannel(\PointerB))$
  violates condition~(3), since it reduces to a parallel composition
  of subprocesses where one has sent an $\Tag$-tagged message, but the
  other one was expecting an $\Tag'$-tagged message.

\item The process $\openChannel(\PointerA : \SessionTypeT, \PointerB :
  \SessionTypeS).\send{\shared\PointerA}\Tag{}.\receive\PointerB\Tag{}.(\closeChannel(\PointerA) \parop
  \closeChannel(\PointerB))$ violates condition~(3), since it reduces
  to a stuck process that attempts at sending an $\Tag$-tagged message
  using the unrestricted pointer $\shared\PointerA$, while in fact
  $\PointerA$ is a linear pointer.
\end{iteMize}


\section{Type System}
\label{sec:type_system}

\subsection{Weighing Types}

We aim at defining a type system such that well-typed processes are
well behaved. In session type systems, from which we draw inspiration,
each action performed by a process using a certain endpoint must be
matched by a corresponding action in the type associated with the
endpoint, and the continuation process after that action must behave
according the continuation in the endpoint type.
Following this intuition, the reader may verify that the process
$\mathrm{BODY}$ (Example~\ref{ex:map_process}) uses the endpoint
$\VarZ$ correctly with respect to the endpoint type
$\co{\Tmap(\tvarA,\tvarB)}$ (Example~\ref{ex:map_type}). Analogous
observations can be made for the other endpoints ($\EMapper$,
$\ESource$, $\ETarget$) received from $\VarZ$ and subsequently used in
$\mathrm{BODY}$.
Linearity makes sure that a process owning an endpoint \emph{must} use
the endpoint (according to its type), or it must delegate it to
another process. Endpoints cannot be simply forgotten and this is
essential in guaranteeing the absence of leaks. In
Example~\ref{ex:map_process} there is a number of endpoints involved:
$\PointerC$ is owned permanently by $\mathrm{MAP}$; $\VarZ$ is owned
by $\mathrm{BODY}$ until an \TagEos-tagged message is received, at
which point it is deallocated; $\ESource$ and $\ETarget$ are acquired
by $\mathrm{BODY}$ and deallocated when no longer in use; finally,
$\EMapper$ is acquired by $\mathrm{BODY}$ from the caller and returned
to the caller when $\mathrm{BODY}$ ends. Overall, $\mathrm{MAP}$ is
evenly balanced as far as the ownership of linear endpoints is
concerned.

Nonetheless, as we have anticipated in Section~\ref{sec:singularity},
there are apparently well-typed processes that lead to a violation of
the ownership invariant. A first example is the process
\begin{equation}
\label{eq:micidiale}
  \ProcessP
  =
  \openChannel(\PointerA : \SessionTypeT_1, \PointerB : \SessionTypeT_2).
  \send\PointerA\Tag\PointerB.
  \closeChannel(\PointerA)
\end{equation}
where
\[
  \SessionTypeT_1 = {!}\msg\Tag{\qlin~\SessionTypeT_2}.\SessionEnd
  \text{\qquad and\qquad}
  \SessionTypeT_2 = \trec\tvar.{?}\msg\Tag{\qlin~\tvar}.\SessionEnd
  \,.
\]

The process $\ProcessP$ begins by creating two endpoints $\PointerA$
and $\PointerB$ with dual endpoint types. The fact that
$\SessionTypeT_1 = \co{\SessionTypeT_2}$ ensures the absence of
communication errors, as each action performed on one endpoint is
matched by a corresponding co-action performed on the corresponding
peer. After its creation, endpoint $\PointerB$ is sent over endpoint
$\PointerA$. Observe that, according to $\SessionTypeT_1$, the process
is entitled to send an $\Tag$-tagged message with argument of type
$\SessionTypeT_2$ on $\PointerA$ and $\PointerB$ has precisely that
type. After the output operation, the process no longer owns endpoint
$\PointerB$ and endpoint $\PointerA$ is deallocated. Apparently,
$\ProcessP$ behaves correctly while in fact it generates a leak, as we
can see from its reduction:
\[
  \system\EmptyMemory\ProcessP
  \red{}
  \system{
    \PointerA \mapsto [\PointerB, \EmptyQueue],
    \PointerB \mapsto [\PointerA, \EmptyQueue]
  }{
    \send\PointerA\Tag\PointerB.
    \closeChannel(\PointerA)
  }
  \red{}
  \system{
    \PointerA \mapsto [\PointerB, \EmptyQueue],
    \PointerB \mapsto [\PointerA, \msg{\Tag}{\PointerB}]
  }{
    \closeChannel(\PointerA)
  }
\]

In the final, stable configuration we have
$\reachable{\fn(\closeChannel(\PointerA))}\Memory =
\reachable{\{\PointerA\}}\Memory = \{ \PointerA \}$ (recall that
$\PointerB$ is not reachable from $\PointerA$ even though its peer is)
while $\dom(\Memory) = \{ \PointerA, \PointerB \}$. In particular, the
endpoint $\PointerB$ is no longer reachable and this configuration
violates condition~(1) of Definition~\ref{def:wb}.
Additionally, if there were some mechanism for accessing $\PointerB$
(for example, by peeking into the endpoint located at $\PointerA$) and
for reading the message from $\PointerB$'s queue, this would
compromise the typing of $\PointerB$: the endpoint type associated
with $\PointerB$ is $\SessionTypeT_2$, but as we remove the message
from its queue it turns to $\SessionEnd$. The $\PointerB$ in the
message, however, would retain the now obsolete type
$\SessionTypeT_2$, with potentially catastrophic consequences.
A closer look at the heap in the reduction above reveals that the
problem lies in the cycle involving $\PointerB$: it is as if the
$\PointerB \mapsto [\PointerA, \msg\Tag\PointerB]$ region of the heap
needs not be owned by any process because it ``owns itself''.  With
respect to other type systems for session types, we must tighten our
typing rules and make sure that no cycle involving endpoint queues is
created in the heap. In the process above this problem would not be
too hard to detect, as the fact that $\PointerA$ and $\PointerB$ are
peer endpoints is apparent from the syntax of the process.
In general, however, $\PointerA$ and $\PointerB$ might have been
acquired in previous communications (think of the \Code{foo} and
\Code{bar} functions in Section~\ref{sec:singularity}, where nothing
is known about the arguments \Code{e} and \Code{f} save for their
type) and they may not even be peers. For example, the process
\[
  \openChannel(\PointerA : \SessionTypeT_1, \PointerC: \SessionTypeT_2).
  \openChannel(\PointerB : \SessionTypeT_1, \PointerD: \SessionTypeT_2).
  \send\PointerA\Tag\PointerD.
  \send\PointerB\Tag\PointerC.
  (\closeChannel(\PointerA) \parop \closeChannel(\PointerB))
\]
creates a leak with a cycle of length 2 even though no endpoint is
ever sent over its own peer.

Our approach for attacking the problem stems from the observation that
infinite values (once the leak configuration has been reached the
endpoint $\PointerB$ above fits well in this category) usually inhabit
recursive types and the endpoint type $\SessionTypeT_2$ indeed
exhibits an odd form of recursion, as the recursion variable $\tvar$
occurs within the only prefix of $\SessionTypeT_2$.  Forbidding this
form of recursion in general, however, would \hypo1 unnecessarily
restrict our language and \hypo2 it would not protect us completely
against leaks.
Regarding \hypo1, we can argue that an endpoint type $\SessionTypeT_2'
= \trec\tvar.{!}\msg\Tag\tvar.\SessionEnd$ (which begins with an
output action) would never allow the creation of cycles in the heap
despite its odd recursion. The reason is that, if we are sending an
endpoint $\PointerB : \SessionTypeT_2'$ over $\PointerA :
\SessionTypeT_2'$, then the peer of $\PointerA$ must have the dual
type $\co{\SessionTypeT_2'} =
{?}\msg\Tag{\SessionTypeT_2'}.\SessionEnd$ (which begins with an input
action) and therefore must be different from~$\PointerB$.
Regarding \hypo2, consider the following variation of the process
$\ProcessP$ above
\begin{equation}
\label{eq:micidiale2}
  \ProcessQ = \openChannel(\PointerA : \SessionTypeS_1, \PointerB : \SessionTypeS_2).
              \xsend{\PointerA}\Tag{\SessionTypeS_2}{\PointerB}.
              \closeChannel(\PointerA)
\end{equation}
where
\[
  \SessionTypeS_1 = {!}\xmsg\Tag\tvar{\qlin~\tvar}.\SessionEnd
  \text{\qquad and\qquad}
  \SessionTypeS_2 = {?}\xmsg\Tag\tvar{\qlin~\tvar}.\SessionEnd
  \,.
\]

Once again, $\SessionTypeS_1$ and $\SessionTypeS_2$ are dual endpoint
types and process $\ProcessQ$ behaves correctly with respect to
them. Notice that neither $\SessionTypeS_1$ nor $\SessionTypeS_2$ is
recursive, and yet $\ProcessQ$ yields the same kind of leak that we
have observed in the reduction of $\ProcessP$.

What do $\SessionTypeT_2$ and $\SessionTypeS_2$ have in common that
$\SessionTypeT_2'$ and $\SessionTypeS_1$ do not and that makes them
dangerous?
First of all, both $\SessionTypeT_2$ and $\SessionTypeS_2$ begin with
an input action so they denote endpoints in a \emph{receive state},
and only endpoints in a receive state can have a non-empty queue.
Second, the type of the arguments in $\SessionTypeT_2$ and
$\SessionTypeS_2$ \emph{may} denote other endpoints with a non-empty
queue: in $\SessionTypeT_2$ this is evident as the type of the
argument is $\SessionTypeT_2$ itself; in $\SessionTypeS_2$ the type of
the argument is the existentially quantified type variable $\tvar$,
which can be instantiated with \emph{any} endpoint type and, in
particular, with an endpoint type beginning with an input action. If
we think of the chain of pointers originating from the queue of an
endpoint, we see that both $\SessionTypeT_2$ and $\SessionTypeS_2$
allow for chains of arbitrary length and the leak originates when this
chain becomes in fact infinite, meaning that a cycle has formed in the
heap.
Our idea to avoid these cycles uses the fact that it is possible to
compute, for each endpoint type, a value in the set $\natset \cup \{
\infty \}$, that we call \emph{weight}, representing the upper bound
of the length of any chain of pointers originating from the queue of
the endpoints it denotes. A weight equal to $\infty$ means that there
is no such upper bound.
Then, the idea is to restrict the type system so that:
\begin{quote}
  \vspace{0.5ex}
  \framebox{Only endpoints having a finite-weight type can be sent as
    messages.}
  \vspace{0.5ex}
\end{quote}

A major issue in defining the weight of types is how to deal with type
variables. If type variables can be instantiated with arbitrary
endpoint types, hence with endpoint types having arbitrary weight, the
weight of type variables cannot be estimated to be finite. At the same
time, assigning an infinite weight to \emph{every} type variable can
be overly restrictive. To see why, consider the following fragment of
the $\mathrm{MAP}$ process defined in Example~\ref{ex:map_process}:
\[
[{\cdots}].
\receive\ESource{\mathtt{Data}}{\VarX : \qlin~\tvarA}.
\send\EMapper{\mathtt{Arg}}\VarX.
[{\cdots}]
\]

The process performs an output operation
$\send\EMapper{\mathtt{Arg}}{\VarX}$ which, according to our idea,
would be allowed only if the type of argument $\VarX$ had a finite
weight. It turns out that $\VarX$ has type $\qlin~\tvarA$ and is bound
by the preceding input action $\receive\ESource{\mathtt{Data}}{\VarX
  : \qlin~\tvarA}$.
If we estimate the weight to $\tvarA$ to be infinite, a simple process
like $\mathrm{MAP}$ would be rejected by our type system.
By looking at the process more carefully one realizes that, since
$\Var$ has been received from a message, its actual type \emph{must
  be} finite-weight, for otherwise the sender (the process using
$\ESource$'s peer endpoint) would have been rejected by the type
system.
In general, since type variables denote values that can only be passed
around and these must have a finite-weight type, it makes sense to
impose a further restriction:
\begin{center}
  \vspace{0.5ex}
  \framebox{Only finite-weight endpoint types can instantiate type
    variables.}
  \vspace{0.5ex}
\end{center}
Then, in computing the weight of a type, we should treat its free and
bound type variables differently: free type variables are placeholders
for a finite-weight endpoint type and are given a finite weight; bound
type variables are yet to be instantiated with some unknown endpoint
type of arbitrary weight and therefore their weight cannot be
estimated to be finite. We will thus define the weight
$\xweight\BoundContext\Type$ of a type $\Type$ with respect to a set
$\BoundContext$ of free type variables:

\begin{definition}[type weight]
\label{def:type_weight}
We say that $\wrel$ is a \emph{coinductive weight bound} if
$(\BoundContext, \SessionType, n) \in \wrel$ implies either:
\begin{iteMize}{$\bullet$}
\item $\SessionType = \SessionEnd$, or

\item $\SessionType = \tvar \in \BoundContext$, or

\item $\SessionType =
  \InternalChoice{\tmsg{\Tag_i}{\tvar_i}{\TypeT_i}.\SessionType_i}_{i\in
    I}$, or

\item $\SessionType =
  \ExternalChoice{\tmsg{\Tag_i}{\tvar_i}{\Qualifier_i~\SessionTypeS_i}.\SessionTypeT_i}_{i\in
    I}$ and $n > 0$ and $\tvar_i \not\in \BoundContext$ and
  $(\BoundContext, \SessionTypeS_i, n-1) \in \wrel$ and
  $(\BoundContext, \SessionTypeT_i, n) \in \wrel$ for every $i\in I$.\smallskip
\end{iteMize}

\noindent We write $\BoundContext \vdash \SessionType \wbound n$ if
$(\BoundContext, \SessionType, n) \in {\wrel}$ for some coinductive
weight bound $\wrel$.
The \emph{weight} of an endpoint type $\SessionType$ with respect to
$\BoundContext$, denoted by $\xweight\BoundContext\SessionType$, is
defined by $\xweight\BoundContext\SessionType = \min\{n\in\natset \mid
\BoundContext \vdash \SessionType \wbound n\}$ where we let
$\min\emptyset = \infty$. We simply write $\weight\SessionType$ in
place of $\xweight\EmptyBoundContext\SessionType$ and we extend
weights to types so that $\weight{\Qualifier~\SessionType} =
\weight\SessionType$.
When comparing weights we extend the usual total orders $<$ and $\leq$
over natural numbers so that $n < \infty$ for every $n\in\natset$ and
$\infty \leq \infty$.
\end{definition}

The weight of $\Type$ is defined as the least of its weight bounds, or
$\infty$ if there is no such weight bound.
A few weights are straightforward to compute, for example we have
$\weight\SessionEnd =
\weight{\InternalChoice{\tmsg{\Tag_i}{\tvar_i}{\TypeT_i}.\SessionType_i}_{i\in
    I}} = 0$.
Indeed, the queues of endpoints with type $\SessionEnd$ and those in a
send state are empty and therefore the chains of pointers originating
from them has zero length.
A type variable $\tvar$ can have a finite or infinite weight depending
on whether it occurs free or bound. So we have
$\xweight{\{\tvar\}}\tvar = 0$ and $\weight\tvar = \infty$.
Note that $\xweight{\{\tvar\}}\tvar = 0$ although $\tvar$ may be
actually instantiated with a type that has a strictly positive, but
finite weight.
Endpoint types in a receive state have a strictly positive
weight. For instance we have
$\weight{{?}\Tag(\SessionEnd).\SessionEnd} = 1$ and
$\weight{{?}\Tag({?}\Tag(\SessionEnd).\SessionEnd).\SessionEnd} = 2$.
If we go back to the examples of endpoint types that we used to
motivate this discussion, we have $\weight{\SessionTypeT_2'} =
\weight{\SessionTypeS_1} = 0$ and $\weight{\SessionTypeT_2} =
\weight{\SessionTypeS_2} = \infty$, from which we deduce that
endpoints with type $\SessionTypeT_2'$ or $\SessionTypeS_1$ are safe
to be sent as messages, while endpoints with type $\SessionTypeT_2$ or
$\SessionTypeS_2$ are not.

Before we move on to illustrating the type system, we must discuss one
last issue that has to do with subtyping. Any type system with
subtyping normally allows to use a value having type $\TypeT$ where a
value having type $\TypeS$ with $\TypeT \subt \TypeS$ is expected.
For example, in the $\mathrm{MAP}$ process we have silently made the
assumption that the value $\VarX$ received with the
$\mathtt{Data}$-tagged message had \emph{exactly} the (finite-weight)
type with which $\tvarA$ has been instantiated while in fact $\VarX$
might have a \emph{smaller} type. Therefore, the restrictions we have
designed work provided that, if $\TypeT \subt \TypeS$ and $\TypeS$ is
finite-weight, then $\TypeT$ is finite-weight as well.
This is indeed the case, and in fact we can express an even stronger
correspondence between weights and subtyping:

\begin{proposition}
\label{prop:subt_weight}
$\TypeT \subt \TypeS$ implies $\xweight\BoundContext\TypeT \leq
\xweight\BoundContext\TypeS$.
\end{proposition}
\begin{proof}
  It is easy to show that ${\wrel} = \{ (\BoundContext, \SessionTypeT,
  n) \mid \exists\SessionTypeS: \BoundContext \vdash \SessionTypeT
  \subt \SessionTypeS \land \wbb\BoundContext\SessionTypeS{n} \} \cup
  \{ (\BoundContext, \TypeT, n) \mid \exists \TypeS: \BoundContext
  \vdash \TypeT \subt \TypeS \land \wbb\BoundContext\TypeS{n} \}$ is a
  coinductive weight bound.
\end{proof}
  
\subsection{Typing the Heap}

The heap plays a primary role because inter-process communication
utterly relies on heap-allocated structures; also, most properties of
well-behaved processes are direct consequences of related properties
of the heap.
Therefore, just as we will check well typedness of a process
$\Process$ with respect to a type environment that associates the
pointers occurring in $\Process$ with the corresponding types, we will
also need to check that the heap is consistent with respect to the
same environment. This leads to a notion of well-typed heap that we
develop in this section.
The mere fact that we have this notion does not mean that we need to
type-check the heap at runtime, because well-typed processes will only
create well-typed heaps and the empty heap will be trivially well
typed.
We shall express well-typedness of a heap $\Memory$ with respect to a
pair $\Context_0;\Context$ of type environments where $\Context$
contains the type of unrestricted pointers and the type of the
\emph{roots} of $\Memory$ (the pointers that are not referenced by any
other structure allocated on the heap), while $\Context_0$ contains
the type of the pointers to allocated structures that are reachable
from the roots of $\Memory$.

Among the properties that a well-typed heap must enjoy is the
complementarity between the endpoint types associated with peer
endpoints. This notion of complementarity does not coincide with
duality because of the communication model that we have adopted, which
is asynchronous: since messages can accumulate in the queue of an
endpoint before they are received, the types of peer endpoints can be
misaligned. The two peers are guaranteed to have dual types only when
both their queues are empty. In general, we need to compute the actual
endpoint type of an endpoint by taking into account the messages in
its queue.
To this end we introduce a $\tail(\cdot,\cdot)$ function for endpoint
types such that
\[
\tail(\SessionTypeT, \xmsg\Tag\SessionTypeS\TypeS) = \SessionTypeT'
\]
indicates that a message with tag $\Tag$, type argument
$\SessionTypeS$, and argument of type $\TypeS$ can be received from an
endpoint with type $\SessionTypeT$ which can be used according to type
$\SessionTypeT'$ thereafter. The function is defined by the rule:
\[
\inferrule{
  k\in I
  \\
  \TypeS \subt \TypeT_k\subst{\SessionTypeS}{\tvar_k}
}{
  \tail(\ExternalChoice{\tmsg{\Tag_i}{\tvar_i}{\Type_i}.\SessionTypeT_i}_{i\in I},
        \xmsg{\Tag_k}\SessionTypeS\TypeS)
  =
  \SessionTypeT_k\subst{\SessionTypeS}{\tvar_k}
}
\]

Note that $\tail(\SessionTypeT, \xmsg\Tag\SessionTypeS\TypeS)$ is
undefined when $\SessionTypeT = \SessionEnd$ or $\SessionTypeT$ is an
internal choice. This is consistent with the observation that it is
not possible to receive messages from endpoints having these types.
We extend $\tail(\cdot,\cdot)$ to possibly empty sequences of message
specifications thus:
\[
\begin{array}{r@{~}c@{~}l}
  \tail(\SessionTypeT, \varepsilon) & = & \SessionTypeT
  \\
  \tail(\SessionTypeT, \tmsg{\Tag_1}{\SessionTypeS_1}{\TypeS_1}\cdots\tmsg{\Tag_n}{\SessionTypeS_n}{\TypeS_n})
  & = &
  \tail(\tail(\SessionTypeT,
  \tmsg{\Tag_1}{\SessionTypeS_1}{\TypeS_1}),
  \tmsg{\Tag_2}{\SessionTypeS_2}{\TypeS_2}\cdots\tmsg{\Tag_n}{\SessionTypeS_n}{\TypeS_n})
\end{array}
\]

We now have all the notions to express the well-typedness of a heap
$\Memory$ with respect to a pair $\Context_0; \Context$ of type
environments.
A \emph{type environment} is a finite map $\Context = \{ \Name_i :
\Qualifier_i~\SessionType_i \}_{i\in I}$ from names to types. We adopt
the following notation regarding type environments:
\begin{iteMize}{$\bullet$}
\item We write $\dom(\Context)$ for the domain of $\Context$, namely
  the set $\{ \Name_i \mid i \in I \}$;

\item we write $\Context,\Context'$ for the union of $\Context$ and
  $\Context'$ when $\dom(\Context) \cap \dom(\Context') = \emptyset$;

\item we write $\Qualifier(\Context)$ if $\Qualifier = \Qualifier_i$
  for every $i\in I$ and we say that $\Context$ is \emph{linear} if
  $\qlin(\Context)$ and \emph{unrestricted} if $\qun(\Context)$;

\item we define the $\Qualifier$-\emph{restriction} of $\Context$ as
  $\Context|_\Qualifier = \{ \Name_i : \Qualifier~\SessionType_i \mid
  i\in I \land \Qualifier_i = \Qualifier \}$;

\item finally, we write $\Context \vdash \Name : \Type$ if
  $\Context(\Name) = \Type$.
\end{iteMize}

\begin{definition}[well-typed heap]
\label{def:wth}
Let $\qlin(\Context_0)$ and $\dom(\Context_0) \cap \dom(\Context) =
\emptyset$ where every endpoint type in $\Context_0, \Context$ is well
formed. We write $\Context_0;\Context \vdash \Memory$ if all of the
following conditions hold:
\begin{enumerate}[(1)]
\item For every $\PointerA \mapsto [\PointerB, \Queue] \in \Memory$ we
  have $\PointerB \mapsto [\PointerA, \Queue'] \in \Memory$ and either
  $\PointerA = \PointerB$ or $\Queue = \EmptyQueue$ or $\Queue' =
  \EmptyQueue$.

\item For every
  $\PointerA \mapsto [\PointerB, \EmptyQueue] \in
  \Memory$
  and
  $\PointerB \mapsto [\PointerA,
  \xmsg{\Tag_1}{\SessionTypeS_1}{\Value_1}::\cdots::\xmsg{\Tag_n}{\SessionTypeS_n}{\Value_n}]
  \in \Memory$
with $\PointerA \ne \PointerB$ we have
\[
\co\SessionTypeT
=
\tail(\SessionTypeS,
      \tmsg{\Tag_1}{\SessionTypeS_1}{\TypeS_1}\cdots\tmsg{\Tag_n}{\SessionTypeS_n}{\TypeS_n})
\]
where $\Context_0,\Context \vdash \PointerA : \qlin~\SessionTypeT$ and
$\Context_0,\Context \vdash \PointerB : \qlin~\SessionTypeS$ and
$\Context_0,\Context \vdash \Value_i : \TypeS_i$ and
$\max\{\weight{\SessionTypeS_i}, \weight{\TypeS_i}\} < \infty$ for
$1\le i\le n$.

\item For every $\Pointer \mapsto [\Pointer,
  \xmsg{\Tag_1}{\SessionTypeS_1}{\Value_1}::\cdots::\xmsg{\Tag_n}{\SessionTypeS_n}{\Value_n}] \in \Memory$
  we have
\[
\co\SessionTypeT
=
\tail(\SessionTypeS,
      \tmsg{\Tag_1}{\SessionTypeS_1}{\TypeS_1}\cdots\tmsg{\Tag_n}{\SessionTypeS_n}{\TypeS_n})
\]
where $\Context_0,\Context \vdash \shared\Pointer :
\qun~\SessionTypeT$ and $\Context_0,\Context \vdash \Pointer :
\qlin~\SessionTypeS$ and $\Context_0,\Context \vdash \Value_i :
\TypeS_i$ and $\max\{\weight{\SessionTypeS_i}, \weight{\TypeS_i}\} <
\infty$ for $1\le i\le n$.
  
\item $\dom(\Memory) = \dom(\Context_0,\Context|_\qlin) =
  \reachable{\dom(\Context)}\Memory$;

\item $\reachable{\{\PointerA\}}\Memory \cap
  \reachable{\{\PointerB\}}\Memory = \emptyset$ for every
  $\PointerA,\PointerB \in \dom(\Context)$ with
  $\PointerA\ne\PointerB$.
\end{enumerate}
\end{definition}

Condition~(1) requires that in a well-typed heap every endpoint comes
along with its peer and that at least one of the queues of peer
endpoints be empty. This invariant is ensured by duality, since a
well-typed process cannot send messages on an endpoint until it has
read all the pending messages from the corresponding queue.
Condition~(2) requires that the endpoint types of peer endpoints are
dual. More precisely, for every endpoint $\PointerA$ with an empty
queue, the dual $\co\SessionTypeT$ of its type coincides with the
residual $\tail(\SessionTypeS,
\xmsg{\Tag_1}{\SessionTypeS_1}{\TypeS_1}\cdots\xmsg{\Tag_n}{\SessionTypeS_n}{\TypeS_n})$
of the peer's type $\SessionTypeS$. Additionally, every
$\SessionTypeS_i$ and $\TypeS_i$ has finite weight.
Condition~(3) is similar to condition~(2), but deals with unrestricted
endpoints. The only difference is that $\PointerA$ has no peer
endpoint, and the (unrestricted) dual endpoint type is associated
instead with $\shared\PointerA$.
Condition~(4) states that the type environment $\Context_0,\Context$
must specify a type for all of the allocated objects in the heap and,
in addition, every object (located at) $\PointerA$ in the heap must be
reachable from a root $\PointerB \in \dom(\Context)$.
Finally, condition~(5) requires the uniqueness of the root for every
allocated object.
Overall, since the roots will be distributed linearly to the processes
of the system, conditions~(4) and~(5) guarantee the ownership
invariant, namely that every allocated object belongs to one and only
one process.

\subsection{Typing Processes}

First of all we define an operation on type environments to add new
associations:
\[
\begin{array}{rcl}
  \Context + \Name : \Type =
  \begin{cases}
    \Context
    & \text{if $\Context \vdash \Name : \Type$ and $\qun(\Type)$} \\
    \Context, \Name : \Type
    & \text{if $\Name \not\in \dom(\Context)$} \\
    \text{undefined} & \text{otherwise}
  \end{cases}
\end{array}
\]

In plain words, an association $\Name : \Type$ where $\Type$ is linear
can be added to $\Context$ only if $\Name$ does not already occur in
$\Context$. An association $\Name : \Type$ where $\Type$ is
unrestricted can be added to $\Context$ in two cases: either $\Name$
does not occur in $\Context$, in which case the association is simply
added, or the \emph{same} association already occurs in $\Context$, in
which case the operation has no effect on the environment. In all the
other cases the result is undefined.
We generalize $+$ to pairs of arbitrary environments $\Context +
\Context'$ in the natural way.

\begin{table}
\caption{\label{tab:typing_processes}\strut Typing rules for processes.}
\framebox[\textwidth]{
\begin{math}
\displaystyle
\begin{array}{@{}c@{}}
  \inferrule[\rulename{T-Idle}]{
    \qun(\Context)
  }{
    \RecContext; \BoundContext; \Context \vdash \idle
  }
  \qquad
  \inferrule[\rulename{T-Close}]{
    \qun(\Context)
  }{
    \RecContext; \BoundContext;
    \Context, \Name : \qlin~\SessionEnd \vdash \closeChannel(\Name)
  }
  \\\\
  \inferrule[\rulename{T-Open Linear Channel}]{
    \BoundContext \wfdash \SessionTypeT
    \\
    \RecContext;
    \BoundContext;
    \Context, \PointerA : \qlin~\SessionTypeT, \PointerB : \qlin~\co\SessionTypeT
    \vdash
    \Process
  }{
    \RecContext; \BoundContext; \Context
    \vdash
    \openChannel(\PointerA:\SessionTypeT, \PointerB:\co\SessionTypeT).
    \Process
  }
  \qquad
  \inferrule[\rulename{T-Open Unrestricted Channel}]{
    \BoundContext \wfdash \SessionTypeT
    \\
    \RecContext;
    \BoundContext;
    \Context,
    \Pointer : \qlin~\SessionTypeT,
    \shared\Pointer : \qun~\co\SessionTypeT
    \vdash
    \Process
  }{
    \RecContext; \BoundContext; \Context
    \vdash
    \openChannel(\Pointer: \SessionTypeT).\Process
  }
  \\\\
  \inferrule[\rulename{T-Send}]{
    \BoundContext \wfdash \SessionTypeS
    \\
    k\in I
    \\
    \TypeS \subt \TypeT_k\subst{\SessionTypeS}{\tvar_k}
    \\
    \max\{\xweight{\BoundContext}{\SessionTypeS}, \xweight{\BoundContext}{\TypeS}\} < \infty
    \\
    \RecContext; \BoundContext;
    \Context, \NameA : \Qualifier~\SessionTypeT_k\subst{\SessionTypeS}{\tvar_k}
    \vdash
    \Process
  }{
    \RecContext; \BoundContext;
    (\Context, \NameA : \Qualifier~\InternalChoice{\tmsg{\Tag_i}{\tvar_i}{\Type_i}.\SessionTypeT_i}_{i\in I}) + \NameB : \TypeS
    \vdash
    \xsend\NameA{\Tag_k}{\SessionTypeS}{\NameB}.\Process
  }
  \\\\
  \inferrule[\rulename{T-Receive}]{
    \BoundContext \wfdash \TypeT_i~{}^{(i\in I)}
    \\
    \TypeS_i \subt \TypeT_i~{}^{(i\in I)}
    \\
    \RecContext;
    \BoundContext, \tvar_i;
    \Context, \NameA : \qlin~\SessionType_i, \Var_i : \TypeT_i
    \vdash
    \Process_i
    ~{}^{(i\in I)}
  }{
    \textstyle
    \RecContext;
    \BoundContext;
    \Context,
    \NameA : \qlin~\ExternalChoice{\tmsg{\Tag_i}{\tvar_i}{\TypeS_i}.\SessionTypeT_i}_{i\in I}
    \vdash
    \sum_{i\in I \cup J} \xreceive\NameA{\Tag_i}{\tvar_i}{\Var_i:\TypeT_i}.\Process_i
  }
  \\\\
  \inferrule[\rulename{T-Choice}]{
    \RecContext; \BoundContext; \Context \vdash \ProcessP
    \\
    \RecContext; \BoundContext; \Context \vdash \ProcessQ
  }{
    \RecContext; \BoundContext; \Context \vdash \ProcessP \choice \ProcessQ
  }
  \qquad
  \inferrule[\rulename{T-Par}]{
    \RecContext; \BoundContext; \Context_1 \vdash \ProcessP
    \\
    \RecContext; \BoundContext; \Context_2 \vdash \ProcessQ
  }{
    \RecContext; \BoundContext; \Context_1 + \Context_2 \vdash \ProcessP \parop \ProcessQ
  }
  \\\\
  \inferrule[\rulename{T-Rec}]{
    \RecContext, \{ \RecVar \mapsto (\BoundContext; \Context) \};
    \BoundContext;
    \Context \vdash \Process
    \\
    \dom(\Context|_\qlin) \subseteq \fn(\Process)
  }{
    \RecContext;
    \BoundContext;
    \Context
    \vdash \rec\RecVar.\Process
  }
  \qquad
  \inferrule[\rulename{T-Var}]{
    \qun(\Context')
  }{
    \RecContext, \{ \RecVar \mapsto (\BoundContext; \Context) \};
    \BoundContext, \BoundContext';
    \Context, \Context'
    \vdash
    \RecVar
  }
\end{array}
\end{math}
}
\end{table}

The typing rules for processes are inductively defined in
Table~\ref{tab:typing_processes}. Judgments have the form
$\RecContext; \BoundContext; \Context \vdash \Process$ and state that
process $\Process$ is well typed under the specified environments.
The additional environment $\RecContext$ is a map from process
variables to pairs $(\BoundContext; \Context)$ and is used for typing
recursive processes.
We describe the typing rules in the following paragraphs:
\begin{iteMize}{$\bullet$}
\item Rule~\rulename{T-Idle} states that the idle process is well
  typed in every unrestricted type environment. Since we impose a
  correspondence between the free names of a process and the roots of
  the heap, this rule states that the terminated process has no leaks.

\item Rule~\rulename{T-Close} states that a process
  $\closeChannel(\Name)$ is well typed provided that $\Name$
  corresponds to an endpoint with type $\SessionEnd$, on which no
  further interaction is possible. Also, the remaining type
  environment must be unrestricted.

\item Rule~\rulename{T-Open Linear Channel} deals with the creation of
  a new linear channel, which is visible in the continuation process
  as two peer endpoints typed by dual endpoint types.
  The premise $\BoundContext \wfdash \SessionTypeT$ requires
  $\SessionTypeT$ to be well formed with respect to the type variables
  in $\BoundContext$.
  In addition, the rule implicitly requires that no type variable, not
  even those in $\BoundContext$, can occur at the top level in
  $\SessionTypeT$, for otherwise its dual $\co\SessionTypeT$ would be
  undefined.

\item Rule~\rulename{T-Open Unrestricted Channel} deals with the
  creation of a new unrestricted channel, which is accessible in the
  continuation process by means of two names: $\Pointer$ is the linear
  pointer used for receiving messages while $\shared\Pointer$ is the
  unrestricted pointer used for sending messages. Note that
  $\co\SessionTypeT$ is qualified by `$\qun$', therefore it must be
  $\co\SessionTypeT =
  \InternalChoice{\tmsg{\Tag_i}{\tvar_i}{\Type_i}.\co\SessionTypeT}_{i\in
    I}$ and $\SessionTypeT =
  \ExternalChoice{\tmsg{\Tag_i}{\tvar_i}{\Type_i}.\SessionTypeT}_{i\in
    I}$.

\item Rule~\rulename{T-Send} states that a process
  $\xsend\NameA\Tag\SessionTypeS\NameB.\Process$ is well typed if
  $\NameA$ (which can be either linear or unrestricted according to
  $\Qualifier$) is associated with an endpoint type $\SessionTypeT$
  that permits the output of $\Tag$-tagged messages (second
  premise). The endpoint type $\SessionTypeS$ instantiates the type
  argument of the message, while the type of the argument $\NameB$
  must be a subtype of the expected type in the endpoint type where
  $\tvar$ has been instantiated with $\SessionTypeS$ (third premise).
  Both $\SessionTypeS$ and $\TypeS$ must be finite-weight (fourth
  premise).  Since the peer of $\NameA$ must be able to accept a
  message with an argument of type $\TypeS$, its weight will be
  strictly larger than that of $\TypeS$. This is to make sure that the
  the output operation does not create any cycle in the heap. Observe
  that the weights of $\SessionTypeS$ and $\TypeS$ are computed with
  respect to the environment $\BoundContext$, containing all the free
  type variables that can possibly occur in $\SessionTypeS$ and
  $\TypeS$.
  Finally, the continuation $\Process$ must be well typed in a
  suitable type environment where the endpoint $\NameA$ is typed
  according to a properly instantiated continuation of $\SessionTypeT$
  (fifth premise).
  Beware of the use of $+$ in the type environments of the rule: if
  $\TypeS$ is linear, then $\NameB$ is no longer accessible in the
  continuation $\Process$; if $\TypeS$ is unrestricted, then $\NameB$
  may or may not be available in $\Process$ depending on whether
  $\Process$ uses $\NameB$ again or not.
  Note also that every endpoint type occurring in the process is
  verified to be well formed with respect to $\BoundContext$ (first
  premise).


\item Rule~\rulename{T-Receive} deals with inputs: a process waiting
  for a message from an endpoint $\NameA : \Qualifier~\SessionTypeT$
  is well typed if it can deal with at least all of the message tags
  in the topmost inputs of $\SessionTypeT$. The continuation processes
  may use the endpoint $\NameA$ according to the endpoint type
  $\SessionType_i$ and can access the message argument $\Var_i$. The
  context $\BoundContext$ is enriched with the type variable $\tvar_i$
  denoting the fact that $\Process_i$ does not know the exact type
  with which $\tvar_i$ has been instantiated.
  Like for the previous typing rule, there is an explicit premise
  demanding well-formedness of the types occurring in the process.

\item Rules~\rulename{T-Choice} and~\rulename{T-Par} are standard. In
  the latter, the type environment is split into two environments to
  type the processes being composed. According to the definition of
  $+$, $\Context_1$ and $\Context_2$ can only share associations with
  unrestricted types and, if they do, the associations in $\Context_1$
  and in $\Context_2$ for the same name must be equal.

\item Rule~\rulename{T-Rec} is a nearly standard rule for recursive
  processes, except for the premise $\dom(\Context|_\qlin) \subseteq
  \fn(\Process)$ that enforces a weak form of contractivity in
  processes. It states that $\rec\RecVar.\Process$ is well typed under
  $\Context$ only if $\Process$ actually uses the linear names in
  $\dom(\Context)$. Normally, divergent processes such as
  $\rec\RecVar.\RecVar$ are well typed in every type environment.  If
  this were the case, however, the process $\openChannel(\PointerA :
  \SessionTypeT, \PointerB : \co\SessionTypeT).\rec\RecVar.\RecVar$,
  which leaks $\PointerA$ and $\PointerB$, would be well typed.

\item We conclude with the familiar rule~\rulename{T-Var} that deals
  with recursion variables. The rule takes into account the
  possibility that new type variables and (unrestricted) associations
  have accumulated in $\BoundContext$ and $\Context$ since the binding
  of $\RecVar$.
\end{iteMize}

\noindent Systems $\system\Memory\Process$ are well typed if so are their
components:

\begin{definition}[well-typed system]
  We write $\Context_0;\Context \vdash (\Memory;\Process)$ if
  $\Context_0;\Context \vdash \Memory$ and $\Context \vdash \Process$.
\end{definition}

Let us present the two main results about our framework:
well-typedness is preserved by reduction, and well-typed processes are
well behaved.  Subject reduction takes into account the possibility
that types in the environment may change as the process reduces, which
is common in behavioral type theories.

\begin{theorem}[subject reduction]
\label{thm:sr}
Let $\Context_0; \Context \vdash (\Memory; \Process)$ and $(\Memory;
\Process) \red{} (\Memory'; \Process')$. Then $\Context_0'; \Context'
\vdash (\Memory'; \Process')$ for some $\Context_0'$ and $\Context'$.
\end{theorem}

\begin{theorem}[safety]
\label{thm:safety}
Let $\vdash \Process$. Then $\Process$ is well behaved.
\end{theorem}

\subsection{Examples}

We conclude this section with a few extended examples: the first one
is meant to show a typing derivation; the second one presents a
scenario in which it would be natural to send around endpoints with
infinite weight, and shows a safe workaround to circumvent the
finite-weight restriction; the last example demonstrates the
expressiveness of our calculus in modeling some advanced features of
\Sing{}, namely the ability to safely share linear pointers between
several processes.

\begin{example}[forwarder]
\label{ex:forwarder}
\newcommand{\homo}{1}
\newcommand{\hetero}{2}
We illustrate a type derivation for a simple forwarder process that
receives two endpoints with dual types and forwards the stream of
$\Tag$-tagged messages coming from the first endpoint to the second
one. We have at least two ways to implement the forwarder, depending
on whether the stream is homogeneous (all the $\Tag$-tagged messages
carry an argument of the same type) or heterogeneous (different
$\Tag$-tagged messages may carry arguments of possibly different
types). Considering the latter possibility we have:
\[
\begin{array}{@{}r@{~}c@{~}l@{}}
    \mathrm{FWD}(\PointerA) & = &
    \receive\Pointer{\mathtt{Src}}{\VarX : \qlin~\SessionType}.
    \receive\Pointer{\mathtt{Dest}}{\VarY :
      \qlin~\co{\SessionType}}.
    \\
    & & \quad
    (\closeChannel(\Pointer) \parop
    \rec\RecVar.
    \xreceive\VarX\Tag\tvar{\VarZ:\qlin~\tvar}.
    \xsend\VarY\Tag\tvar\VarZ.\RecVar)
\end{array}
\]
where
\[
  \SessionType =
  \trec\tvarB.{?}\xmsg\Tag\tvarA{\qlin~\tvarA}.\tvarB
  \,.
\]
\newcommand{\ETC}{[\,{\cdots}\,]}%
Below we show the derivation proving that $\mathrm{FWD}$ is
well typed. To keep the derivation's size manageable, we elide some
subprocesses with $\ETC$ and we define $\Context = \VarX :
\qlin~\SessionType, \VarY : \qlin~\co{\SessionType}$.
\begin{prooftree}
  \AxiomC{}
  \RightLabel{\rulename{T-Close}}
  \UnaryInfC{$
    \Pointer : \qlin~\SessionEnd
    \vdash
    \closeChannel(\Pointer)
  $}
  \AxiomC{}
  \RightLabel{\rulename{T-Var}}
  \UnaryInfC{$
    \{ \RecVar \mapsto (\EmptyBoundContext; \Context) \};
    \tvar;
    \Context
    \vdash    
    \RecVar
  $}
  \RightLabel{\rulename{T-Send}}
  \UnaryInfC{$
    \{ \RecVar \mapsto (\EmptyBoundContext; \Context) \};
    \tvar;
    \Context,
    \VarZ : \qlin~\tvar
    \vdash    
    \xsend\VarY\Tag\tvar\VarZ.\RecVar
  $}
  \RightLabel{\rulename{T-Receive}}
  \UnaryInfC{$
    \{ \RecVar \mapsto (\EmptyBoundContext; \Context) \};
    \EmptyBoundContext;
    \Context
    \vdash    
    \xreceive\VarX\Tag\tvar{\VarZ:\qlin~\tvar}.\ETC
  $}
  \RightLabel{\rulename{T-Rec}}
  \UnaryInfC{$
    \VarX : \qlin~\SessionType,
    \VarY : \qlin~\co{\SessionType}
    \vdash
    \rec\RecVar.\ETC
  $}
  \RightLabel{\rulename{T-Par}}
  \BinaryInfC{$
    \Pointer : \qlin~\SessionEnd,
    \VarX : \qlin~\SessionType,
    \VarY : \qlin~\co{\SessionType}
    \vdash
    \closeChannel(\Pointer) \parop
    \rec\RecVar.\ETC
  $}
  \RightLabel{\rulename{T-Receive}}
  \UnaryInfC{$
    \Pointer : \qlin~{?}\msg{\mathtt{Dest}}{\qlin~\co{\SessionType}}.
               \SessionEnd,
    \VarX : \qlin~\SessionType
    \vdash
    \receive\Pointer{\mathtt{Dest}}{\VarY : \qlin~\co{\SessionType}}.
    \ETC
  $}
  \RightLabel{\rulename{T-Receive}}
  \UnaryInfC{$
    \Pointer : \qlin~{?}\msg{\mathtt{Src}}{\qlin~\SessionType}.
               {?}\msg{\mathtt{Dest}}{\qlin~\co{\SessionType}}.
               \SessionEnd
               \vdash \mathrm{FWD}(\Pointer)
  $}
\end{prooftree}\smallskip

\noindent Observe that, by the time rule \rulename{T-Var} is applied for the
process variable $\RecVar$, a type variable $\tvar$ has accumulated
into the bound type variables context which was empty when $\RecVar$
was introduced in \rulename{T-Rec}. Therefore, it is essential for
rule \rulename{T-Var} to discharge extra type variables in the bound
type variable context for declaring this process well typed.
%
\eoe
\end{example}

\begin{example}[linear lists]
\label{ex:lists}
\newcommand{\List}[1]{\mathtt{List}(#1)}
\newcommand{\TagPrompt}{\mathtt{Prompt}}
\newcommand{\TagNil}{\mathtt{Nil}}
\newcommand{\TagCons}{\mathtt{Cons}}
In most of the examples we have presented so far the type of services
begins with an output action, suggesting that it is the consumers of
these services that play the first move and \emph{invoke} them by
sending a message. There are cases, in particular with the modeling of
datatypes, where it is more natural to adopt the dual point of view,
in which the reception of a message indicates the consumption of the
data type.
In this example we represent a linear list as an endpoint from which
one of two kinds of messages can be received: a $\TagNil$-tagged
message indicates that the list is empty; a $\TagCons{}$-tagged
message indicates that the list has at least one element, and the
parameters of the message are the head of the list and its tail, which
is itself a list. Reading a message from the endpoint corresponds to
deconstructing the list and the tag-based dispatching of messages
implements pattern matching.
Along these lines, the type of lists with elements of type $\tvar$
would be encoded as the endpoint type
\[
\List\tvarA =
\trec\tvarB.
({?}\msg\TagNil{}.\SessionEnd
+
{?}\msg\TagCons{\qlin~\tvarA, \qlin~\tvarB}.\SessionEnd)
\]

Note that, just as this type denotes lists of arbitrary length, the
encoding of lists in terms of messages within endpoints may yield
chains of pointers of arbitrary length because of the recursion of
$\tvarB$ through an input prefix. As a consequence we have
$\xweight{\{\tvarA\}}{\List\tvarA} = \infty$, meaning that our type
system would reject any output operation sending a list over an
endpoint. Incidentally, since a non-empty list is encoded as a
$\TagCons$-tagged message containing another list, the finite-weight
restriction on the type of message arguments would in fact prevent the
construction of any non-trivial list, rendering the type $\List\tvarA$
useless.

It is possible to fix this by requiring the consumers of the list to
signal the imminent deconstruction via a ``prompt'' message. This
corresponds to defining
\[
\List\tvarA =
\trec\tvarB.{!}\msg\TagPrompt{}.
({?}\msg\TagNil{}.\SessionEnd
+
{?}\msg\TagCons{\qlin~\tvarA, \qlin~\tvarB}.\SessionEnd)
\]

The insertion of an output action between the binding of $\tvarB$ and
its occurrence among the arguments of $\TagCons$ nullifies the weight
of $\List\tvarA$, that is $\xweight{\{\tvarA\}}{\List\tvarA} = 0$.
\REVISION{To see why this is sufficient for preventing the creation of
  cycles in the heap consider a process
\[
   \send\PointerB\TagCons{\VarX, \PointerA}.\Process
\]
where we assume that $\PointerA : \List\tvar$ and $\PointerB :
{!}\msg\TagNil{}.\SessionEnd \oplus {!}\msg\TagCons{\qlin~\tvar,
  \List\tvar})$. The intention here is to yield a leak like the one
generated by the process in~\eqref{eq:micidiale}. Note however that
the peer endpoint of $\PointerB$ must have already been used for
sending the $\TagPrompt$-tagged message, while $\PointerA$ has type
$\List\tvar$ and therefore no $\TagPrompt$-tagged message has been
sent on $\PointerA$ yet. We conclude that $\PointerA$ cannot be
$\PointerB$'s peer.}

As an example of list-manipulating function we can now define the
polymorphic consing service on channel $\PointerC$, that creates a
list from a head and a tail, thus:
\[
\mathrm{CONS}(\PointerC) =
\begin{array}[t]{@{}l@{}}
  \rec\RecVar.
  \xreceive\PointerC{\mathtt{Invoke}}\tvar{\VarX : \qlin~\SessionType}. \\
  \receive\VarX{\mathtt{Arg}}{\VarY : \qlin~\tvar}.
  \receive\VarX{\mathtt{Arg}}{\VarZ : \qlin~\List\tvar}. \\
  \openChannel(\PointerA : \List\tvar, \PointerB : \co{\List\tvar}). \\
  (\receive\PointerB\TagPrompt{}.
   \send\PointerB\TagCons{\VarY, \VarZ}.
   \closeChannel(\PointerB)
  \parop
  \send\VarX{\mathtt{res}}\PointerA.\RecVar)
\end{array}
\]
where
\[
  \SessionType =
  {?}\msg{\mathtt{Arg}}{\qlin~\tvar}.
  {?}\msg{\mathtt{Arg}}{\qlin~\List\tvar}.
  {!}\msg{\mathtt{Res}}{\qlin~\List\tvar}.
  \SessionEnd
\]

The interested reader can verify that
\[
  \PointerC : \trec\tvarB.{?}\xmsg{\mathtt{Invoke}}\tvarA{\qlin~\SessionType}.\tvarB
  \vdash
  \mathrm{CONS}(\PointerC)
\]
is derivable.
\eoe
\end{example}

\begin{example}
\label{ex:tref}
\newcommand{\acquire}{\mathit{acquire}}
\newcommand{\release}{\mathit{release}}
\newcommand{\buffer}{\mathit{buffer}}
\newcommand{\MKCELL}{\mathrm{MKCELL}}
\newcommand{\FULL}{\mathrm{FULL}}
\newcommand{\EMPTY}{\mathrm{EMPTY}}
\newcommand{\invokeT}{\mathtt{Invoke}}
\newcommand{\acquireT}{\mathtt{Acquire}}
\newcommand{\releaseT}{\mathtt{Release}}
\newcommand{\argT}{\mathtt{Arg}}
\newcommand{\okT}{\mathtt{Ok}}
\newcommand{\resT}{\mathtt{Res}}
\newcommand{\inT}{\mathtt{In}}
\newcommand{\SessionCell}{\SessionTypeT_{\mathtt{TCell}}(\tvar)}
\newcommand{\SessionBuffer}{\SessionTypeT_{\mathtt{Buffer}}(\tvar)}
\newcommand{\SessionAcquire}{\SessionTypeT_{\mathtt{Acquire}}(\tvar)}
\newcommand{\SessionRelease}{\SessionTypeT_{\mathtt{Release}}(\tvar)}

Development of the Singularity OS prototype has suggested that there
are many scenarios in which the ownership invariant, requiring that a
given object -- an endpoint -- can be owned exclusively by one sole
process at any given time, easily leads to convoluted code. For this
reason, \Sing{} provides a \Code{TCell<$\tvar$>} class that permits
the unrestricted sharing of exchange heap pointers at the expense of
some runtime checks. In practice, an instance of \Code{TCell<$\tvar$>}
acts like a 1-place buffer for a linear pointer of type $\tvar$ and
can be shared non-linearly among different processes. A process
willing to use the pointer must explicitly \emph{acquire} it, while a
process that has finished using the pointer must \emph{release}
it. The internal implementation of \Code{TCell<$\tvar$>} makes sure
that, once the pointer has been acquired, all subsequent acquisition
requests will be blocked until a release is performed. The interface
of \Code{TCell<$\tvar$>} is as follows:

\begin{SingSharp}
  class TCell<$\tvar$> {
    TCell([Claims] $\tvar$ in ExHeap);
    $\tvar$ in ExHeap Acquire();
    void Release([Claims] $\tvar$ in ExHeap);
  }
\end{SingSharp}

\begin{table}
\caption{\label{tab:supercell}\strut Modeling of a shared mutable cell.}
\framebox[\textwidth]{
\begin{math}
\displaystyle
\begin{array}{@{}r@{~}c@{~}l@{}}
  \MKCELL(\PointerA) & = &
  \xreceive\PointerA\invokeT\tvar{\VarX : \qlin~{!}\msg\resT{\qun~\co{\SessionCell}}.\SessionEnd}. \\
  & &
  \openChannel(\PointerC : \SessionCell).
  \openChannel(\buffer : \SessionBuffer). \\
  & &
  \openChannel(\acquire : \SessionAcquire).
  \openChannel(\release : \SessionRelease).
  \\
  & & \send\VarX\resT{\shared\PointerC}.(\EMPTY(\tvar,\PointerC) \parop 
       \closeChannel(\VarX) \parop \MKCELL(\PointerA))
  \\
  \EMPTY(\tvar,\PointerC) & = &
  \receive\PointerC\invokeT{\VarX : \qlin~\SessionTypeT}. \\
  & &
  \begin{array}[t]{@{}l@{~}l@{}}
    ( & \receive\Var\acquireT{}.
      \send{\shared\acquire}\inT\VarX.
      \EMPTY(\tvar,\PointerC)
    \\
    + & \receive\VarX\releaseT{}.
        \send\VarX\okT{}.
        \receive\VarX\argT{\VarY : \qlin~\tvar}. \\
      & \iteb{\mathtt{empty}(\acquire)}{
          (\send{\shared\buffer}\inT\VarY \parop \closeChannel(\VarX) \parop \FULL(\tvar,\PointerC))
        }{
          \begin{array}[t]{@{}l@{}}
            \receive\acquire\inT{\VarZ : \qlin~{!}\msg\resT{\qlin~\tvar}.\SessionEnd}.
            \send\VarZ\resT{\VarY}. \\
            (\closeChannel(\VarX) \parop \closeChannel(\VarZ) \parop \EMPTY(\tvar,\PointerC)))
          \end{array}
        }
  \end{array}
  \\
  \FULL(\tvar,\PointerC) & = &
  \receive\PointerC\invokeT{\VarX : \qlin~\SessionTypeT}. \\
  & &
  \begin{array}[t]{@{}l@{~}l@{}}
    ( & \receive\VarX\acquireT{}.
        \receive\buffer\inT{\VarY : \qlin~\tvar}.
        \send\VarX\resT\VarY. \\
      & \iteb{\mathtt{empty}(\release)}{
          (\closeChannel(\VarX) \parop \EMPTY(\tvar,\PointerC))
        }{
          \begin{array}[t]{@{}l@{}}
            \receive\release\inT{\VarZ : \qlin~{!}\msg\okT{}.{?}\msg\resT{\qlin~\tvar}.\SessionEnd}. \\
            \send\VarZ\okT{}.
            \receive\VarZ\argT{\VarY : \qlin~\tvar}.
            \send{\shared\buffer}\inT\VarY.
            \\
            (\closeChannel(\VarX) \parop \closeChannel(\VarZ) \parop \FULL(\tvar,\PointerC)))
          \end{array}
        }
    \\
    + & \receive\VarX\releaseT{}.
        \send{\shared\release}\inT\VarX.
        \FULL(\tvar,\PointerC))
  \end{array}
  \\\\
  \SessionTypeT(\tvar) & = &
  {?}\msg\acquireT{}.{!}\msg\resT{\qlin~\tvar}.\SessionEnd
  +
  {?}\msg\releaseT{}.{!}\msg\okT{}.{?}\msg\argT{\qlin~\tvar}.\SessionEnd
  \\
  \SessionCell & = & \trec\tvarB.{?}\msg\invokeT{\qlin~\SessionTypeT(\tvar)}.\tvarB
  \\
  \SessionBuffer & = & \trec\tvarB.{?}\msg\inT{\qlin~\tvar}.\tvarB
  \\
  \SessionAcquire & = & \trec\tvarB.{?}\msg\inT{\qlin~{!}\msg\resT{\qlin~\tvar}.\SessionEnd}.\tvarB
  \\
  \SessionRelease & = & \trec\tvarB.{?}\msg\inT{\qlin~{!}\msg\okT{}.{?}\msg\resT{\qlin~\tvar}.\SessionEnd}.\tvarB
\end{array}
\end{math}
}
\end{table}

\noindent Table~\ref{tab:supercell} presents an implementation of the \Sing{}
class \Code{TCell<$\tvar$>} in our process calculus. For readability,
we have defined the $\MKCELL(\PointerA)$ process in terms of
(mutually) recursive equations that can be folded into a proper term
as by~\cite{Courcelle83}. Below we describe the process from a bird's
eye point of view and expect the reader to fill in the missing details.

$\MKCELL(\PointerA)$ waits for invocations on endpoint
$\PointerA$. Each invocation creates a new cell represented as a
linear endpoint $\PointerC$ (retained by the implementation) and an
unrestricted pointer $\shared\PointerC$ (that can be shared by the
users of the cell). The cell consists of three unrestricted endpoints:
$\buffer$ is the actual buffer that contains the pointer to be shared,
while $\acquire$ and $\release$ are used to enqueue pending requests
for acquisition and release of the cell content. The implementation
ensures that $\buffer$ always contains at most one message (of type
$\tvar$), that $\acquire$ can have pending requests only when
$\buffer$ is empty, and that $\release$ can have pending requests only
when $\buffer$ is full.
Users of the cell send invocation requests on endpoint $\PointerC$.
When the cell is empty, any acquisition request is enqueued into
$\acquire$ while a release request checks whether there are pending
acquisition requests by means of the $\mathtt{empty}(\acquire)$
primitive: if there is no pending request, the released pointer
$\VarY$ is stored within $\buffer$ and the cell becomes full; if there
are pending requests, the first one ($\VarZ$) is dequeued and served,
and the cell stays empty.
When the cell is full, any release request is enqueued into $\release$
while the first acquisition request is served immediately. Then, the
cell may become empty or stay full depending on whether there are
pending release requests.

One aspect of this particular implementation which is highlighted by
the endpoint type $\SessionType(\tvar)$ is the handling of multiple
release requests. In principle, it could be reasonable for the
$\releaseT$-tagged message to carry an argument of type $\tvar$, the
pointer being released. However, if this were the case a process
releasing a pointer would \emph{immediately} transfer the ownership of
the pointer to the cell, even in case the cell is in a full
state. This is because communication is asynchronous and send
operations are non-blocking, so the message with the pointer would be
enqueued into $\release$, which is permanently owned by the cell,
regardless of whether the cell is already full. In our modeling, the
$\releaseT$-tagged message carries no argument, its only purpose being
to signal the \emph{intention} for a process to release a pointer. If
the cell is empty, then the cell answers the requester with an
$\okT$-tagged message, and only at that point the pointer (and its
ownership) is transferred from the requester to the cell with an
$\argT$-tagged message. If however the cell is full when the release
request is made, the $\okT$-tagged message is deferred and the
requester remains the formal owner of the pointer being released until
the cell becomes empty again.
\eoe
\end{example}


\section{Algorithms}
\label{sec:algorithms}

In this section we define algorithms for deciding subtyping and for
computing the weight of (endpoint) types. We also argue how the typing
rules in Table~\ref{tab:typing_processes} can be easily turned into a
type checking algorithm using a technique explained elsewhere.

\subsection{Subtyping}
\label{sec:algorithmic_subtyping}

The algorithm for deciding the subtyping relation $\SessionTypeT \subt
\SessionTypeS$ is more easily formulated if we make a few assumptions
on the variables occurring in $\SessionTypeT$ and $\SessionTypeS$. The
reason is that $\subt$ (Definition~\ref{def:subt}) implicitly uses
alpha renaming in order to match the bound type variables occurring in
one endpoint type with the bound type variables occurring in the other
endpoint type. However, termination of the subtyping algorithm can be
guaranteed only if we perform these renamings in a rather controlled
way, and the assumptions we are going to make are aimed at this.

\begin{definition}[independent endpoint types]
  We say that $\SessionTypeT$ and $\SessionTypeS$ are
  \emph{independent} if:
\begin{enumerate}[(1)]
\item $\ftv(\SessionTypeT) \cap \btv(\SessionTypeT) = \emptyset$;

\item $\ftv(\SessionTypeS) \cap \btv(\SessionTypeS) = \emptyset$;

\item no type variable in $\SessionTypeT$ or in $\SessionTypeS$ is
  bound more than once;

\item $\btv(\SessionTypeT) \cap \btv(\SessionTypeS) = \emptyset$.
\end{enumerate}
\end{definition}

Informally, conditions~(1--3) state that $\SessionTypeT$ and
$\SessionTypeS$ obey the so-called Barendregt convention for type
variables, by stating that free and bound type variables are disjoint
and that every type variable is bound at most once.
Condition~(4) makes sure that there is no shared bound type variable
between $\SessionTypeT$ and $\SessionTypeS$.

We will restrict the subtyping algorithm to independent endpoint
types. This is bearable as every pair of endpoint types can be easily
rewritten into an equivalent pair of independent endpoint types:

\begin{proposition}
  For every $\SessionTypeT$ and $\SessionTypeS$ there exist
  independent $\SessionTypeT'$ and $\SessionTypeS'$ such that
  $\SessionTypeT = \SessionTypeT'$ and $\SessionTypeS =
  \SessionTypeS'$.
\end{proposition}
\begin{proof}[Proof sketch]
  A structural induction on $\SessionTypeT$ followed by a structural
  induction on $\SessionTypeS$, in both cases renaming bound variables
  with fresh ones.
\end{proof}

\begin{table}
\caption{\label{tab:algorithmic_subtyping}\strut Algorithmic subtyping rules.}
\framebox[\textwidth]{
\begin{math}
\displaystyle
\begin{array}{@{}c@{}}
  \inferrule[\rulename{S-Var}]{}{
    \srel \vdash_\varmap \tvar \asubt \tvar
  }
  \qquad
  \inferrule[\rulename{S-End}]{}{
    \srel \vdash_\varmap \SessionEnd \asubt \SessionEnd
  }
  \\\\
  \inferrule[\rulename{S-Axiom}]{
    (\SessionTypeT, \SessionTypeS) \in \srel
  }{
    \srel \vdash_\varmap \SessionTypeT \asubt \SessionTypeS
  }
  \qquad
  \inferrule[\rulename{S-Rec Left}]{
    \srel \cup \{( \trec\tvar.\SessionTypeT, \SessionTypeS )\}
    \vdash_\varmap \SessionTypeT\subst{\trec\tvar.\SessionTypeT}{\tvar} \asubt \SessionTypeS
  }{
    \srel \vdash_\varmap \trec\tvar.\SessionTypeT \asubt \SessionTypeS
  }
  \\\\
  \inferrule[\rulename{S-Rec Right}]{
    \srel \cup \{( \SessionTypeT, \trec\tvar.\SessionTypeS )\}
    \vdash_\varmap \SessionTypeT \asubt \SessionTypeS\subst{\trec\tvar.\SessionTypeS}{\tvar}
  }{
    \srel \vdash_\varmap \SessionTypeT \asubt \trec\tvar.\SessionTypeS
  }
  \qquad
  \inferrule[\rulename{S-Type}]{
    \Qualifier \leq \Qualifier'
    \\
    \srel \vdash_\varmap \SessionTypeT \asubt \SessionTypeS
  }{
    \srel \vdash_\varmap
    \Qualifier~\SessionTypeT
    \asubt
    \Qualifier'~\SessionTypeS
  }
  \\\\
  \inferrule[\rulename{S-Input}]{
    \srel' = \srel \cup \{ (\SessionTypeT, \SessionTypeS) \}
    \\
    I \subseteq J
    \\
    \srel' \vdash_\varmap
    \TypeT_i\subst{\varmap(\tvarA_i, \tvarB_i)}{\tvarA_i}
    \asubt
    \TypeS_i\subst{\varmap(\tvarA_i, \tvarB_i)}{\tvarB_i}
    ~{}^{(i\in I)}
    \\
    \srel' \vdash_\varmap
    \SessionTypeT_i\subst{\varmap(\tvarA_i, \tvarB_i)}{\tvarA_i}
    \asubt
    \SessionTypeS_i\subst{\varmap(\tvarA_i, \tvarB_i)}{\tvarB_i}
    ~{}^{(i\in I)}
  }{
    \srel \vdash_\varmap
    \SessionTypeT \seq \ExternalChoice{\xmsg{\Tag_i}{\tvarA_i}{\TypeT_i}.\SessionTypeT_i}_{i\in I}
    \asubt
    \ExternalChoice{\xmsg{\Tag_j}{\tvarB_j}{\TypeS_j}.\SessionTypeS_j}_{j\in J} \seq \SessionTypeS
  }
  \\\\
  \inferrule[\rulename{S-Output}]{
    \srel' = \srel \cup \{ (\SessionTypeT, \SessionTypeS) \}
    \\
    J \subseteq I
    \\
    \srel' \vdash_\varmap
    \TypeS_j\subst{\varmap(\tvarA_j, \tvarB_j)}{\tvarB_j}
    \asubt
    \TypeT_j\subst{\varmap(\tvarA_j, \tvarB_j)}{\tvarA_j}
    ~{}^{(j\in J)}
    \\
    \srel' \vdash_\varmap
    \SessionTypeT_j\subst{\varmap(\tvarA_j, \tvarB_j)}{\tvarA_j}
    \asubt
    \SessionTypeS_j\subst{\varmap(\tvarA_j, \tvarB_j)}{\tvarB_j}
    ~{}^{(j\in J)}
  }{
    \srel \vdash_\varmap
    \SessionTypeT \seq \InternalChoice{\xmsg{\Tag_i}{\tvarA_i}{\TypeT_i}.\SessionTypeT_i}_{i\in I}
    \asubt
    \InternalChoice{\xmsg{\Tag_j}{\tvarB_j}{\TypeS_j}.\SessionTypeS_j}_{j\in J} \seq \SessionTypeS
  }
\end{array}
\end{math}
}
\end{table}

The subtyping algorithm is defined using the rules in
Table~\ref{tab:algorithmic_subtyping}, thus:

\begin{definition}[subtyping algorithm]
\label{def:asubt}
Let $\SessionTypeT$ and $\SessionTypeS$ be independent endpoint types.
Let $\varmap$ be a map from unordered pairs of type variables to type
variables such that $\varmap(\tvarA, \tvarB) \not\in
\ftv(\SessionTypeT) \cup \btv(\SessionTypeT) \cup \ftv(\SessionTypeS)
\cup \btv(\SessionTypeS)$ for every $\tvarA \in \btv(\SessionTypeT)$
and $\tvarB \in \btv(\SessionTypeS)$.
We write $\vdash_\varmap \SessionTypeT \asubt \SessionTypeS$ if and
only if $\emptyset \vdash_\varmap \SessionTypeT \asubt \SessionTypeS$
is derivable with the axioms and rules in
Table~\ref{tab:algorithmic_subtyping}, where we give
rule~\rulename{S-Axiom} the highest priority, followed by
rule~\rulename{S-Rec Left}, followed by rule~\rulename{S-Rec Right},
followed by all the remaining rules which are
syntax-directed.\footnote{The relative priority of
  rules~\rulename{S-Rec Left} and~\rulename{S-Rec Right} is irrelevant
  since they are confluent.}
\end{definition}

The algorithm derives judgments of the form $\srel \vdash_\varmap
\SessionTypeT \asubt \SessionTypeS$, where $\srel$ is a memoization
context that records pairs of endpoint types that are assumed to be
related by subtyping. The map $\varmap$ is used for unifying
consistently the bound type variables of the endpoint types being
related. The same (unordered) pair of bound type variables $(\tvarA,
\tvarB)$ is always unified to the same fresh type variable
$\varmap(\tvarA, \tvarB)$, which is essential for guaranteeing that
the memoization context $\srel$ does not grow unwieldy. The fact that
we work with unordered pairs simply means that $\varmap(\tvarA,
\tvarB) = \varmap(\tvarB, \tvarA)$ for every $\tvarA \in
\btv(\SessionTypeT)$ and $\tvarB \in \btv(\SessionTypeS)$.
The axioms and rules in Table~\ref{tab:algorithmic_subtyping} are
mostly unremarkable, since they closely mimic the coinductive
definition of subtyping (Definition~\ref{def:subt}), therefore we only
comment on the peculiar features of this deduction system:
Axiom~\rulename{S-Axiom} allows one to immediately deduce $\srel
\vdash_\varmap \SessionTypeT \asubt \SessionTypeS$ whenever the pair
$(\SessionTypeT, \SessionTypeS)$ occurs in $\srel$. This prevents the
algorithm to loop forever when comparing recursive endpoint types.
A pair $(\SessionTypeT, \SessionTypeS)$ is added to $\srel$ whenever a
constructor is crossed, which happens in rules~\rulename{S-Rec Left},
\rulename{S-Rec Right}, \rulename{S-Input}, and~\rulename{S-Output}.
Rules~\rulename{S-Rec Left} and~\rulename{S-Rec Right} unfold
recursive endpoint types in order to expose their outermost proper
constructor (an internal/external choice or
$\SessionEnd$). Contractivity of endpoint types guarantees that a
finite number of applications of these rules is always enough to
achieve this exposure.  In Definition~\ref{def:subt} recursive
endpoint types are not treated explicitly since equality `$=$' is
defined modulo folding/unfolding of recursions.
Rules~\rulename{S-Input} and~\rulename{S-Output} deal with inputs and
outputs. Note that the pairs of endpoint types being compared in the
conclusions of the rules have distinct sets $\{ \tvarA_i \}_{i\in I}$
and $\{ \tvarB_j \}_{j\in J}$ of bound type variables that are unified
in the premises by means of the map $\varmap$.

The following result establishes the correctness and completeness of
the subtyping algorithm with respect to $\subt$ for independent
endpoint types.

\begin{theorem}[correctness and completeness]
\label{thm:subtyping_algorithm}
Let $\SessionTypeT_0$ and $\SessionTypeS_0$ be independent endpoint
types and $\varmap$ be a map as by Definition~\ref{def:asubt}. Then
$\vdash_\varmap \SessionTypeT_0 \asubt \SessionTypeS_0$ if and only if
$\SessionTypeT_0 \subt \SessionTypeS_0$.
\end{theorem}

\begin{example}
\newcommand{\TagA}{\mathtt{a}}
\newcommand{\TagB}{\mathtt{b}}
\newcommand{\TagC}{\mathtt{c}}
Consider the endpoint types
\[
\begin{array}{rcl}
  \SessionTypeT & \seq &
  \trec\tvarA.
  {!}\xmsg{\TagA}{\tvarA_1}{{!}\xmsg{\TagB}{\tvarA_2}{\tvarA_2}.\tvarA_1}.\tvarA
  \\
  \SessionTypeS & \seq &
  \trec\tvarB.
  {!}\xmsg{\TagA}{\tvarB_1}{
    {!}\xmsg{\TagB}{\tvarB_2}{\tvarB_2}.\tvarB_1
    \oplus
    {!}\xmsg{\TagC}{\tvarB_3}{\tvarB_3}.\SessionEnd
  }.\tvarB
\end{array}
\]
and observe that they are independent.
The following derivation, together with
Theorem~\ref{thm:subtyping_algorithm}, shows that $\SessionTypeT \subt
\SessionTypeS$:
\begin{prooftree}
  \AxiomC{}
  \RightLabel{\rulename{S-Var}}
  \UnaryInfC{$
    \srel_3 \vdash_\varmap
    \tvarC_2 \asubt \tvarC_2
  $}
  \AxiomC{}
  \RightLabel{\rulename{S-Var}}
  \UnaryInfC{$
    \srel_3 \vdash_\varmap
    \tvarC_1 \asubt \tvarC_1
  $}
  \RightLabel{\rulename{S-Output}}
  \BinaryInfC{$
    \srel_2 \vdash_\varmap
    {!}\xmsg{\TagB}{\tvarB_2}{\tvarB_2}.\tvarC_1 \oplus
    {!}\xmsg{\TagC}{\tvarB_3}{\tvarB_3}.\SessionEnd
    \asubt
    {!}\xmsg{\TagB}{\tvarA_2}{\tvarA_2}.\tvarC_1
  $}
  \AxiomC{}
  \RightLabel{\rulename{S-Axiom}}
  \UnaryInfC{$
    \srel_2 \vdash_\varmap
    \SessionTypeT \asubt \SessionTypeS
  $}
  \RightLabel{\rulename{S-Output}}
  \BinaryInfC{$
    \srel_1 \vdash_\varmap
    {!}\xmsg{\TagA}{\tvarA_1}{{!}\xmsg{\TagB}{\tvarA_2}{\tvarA_2}.\tvarA_1}.\SessionTypeT
    \asubt
    {!}\xmsg{\TagA}{\tvarB_1}{
    {!}\xmsg{\TagB}{\tvarB_2}{\tvarB_2}.\tvarB_1 \oplus
    {!}\xmsg{\TagC}{\tvarB_3}{\tvarB_3}.\SessionEnd}.\SessionTypeS
  $}
  \RightLabel{\rulename{S-Rec Right}}
  \UnaryInfC{$
    \{(\SessionTypeT, \SessionTypeS)\} \vdash_\varmap
    {!}\xmsg{\TagA}{\tvarA_1}{{!}\xmsg{\TagB}{\tvarA_2}{\tvarA_2}.\tvarA_1}.\SessionTypeT
    \asubt
    \SessionTypeS
  $}
  \RightLabel{\rulename{S-Rec Left}}
  \UnaryInfC{$
    \emptyset \vdash_\varmap
    \SessionTypeT \asubt \SessionTypeS
  $}
\end{prooftree}
where we have used the abbreviations:
\begin{iteMize}{$\bullet$}
\item $\tvarC_1 = \varmap(\tvarA_1, \tvarB_1)$ and $\tvarC_2 =
  \varmap(\tvarA_2, \tvarB_2)$;

\item $\srel_1 = \{ (\SessionTypeT, \SessionTypeS),
  ({!}\xmsg{\TagA}{\tvarA_1}{{!}\xmsg{\TagB}{\tvarA_2}{\tvarA_2}.\tvarA_1}.\SessionTypeT,
  \SessionTypeS) \}$;

\item $\srel_2 = \srel_1 \cup \{ (
    {!}\xmsg{\TagA}{\tvarA_1}{{!}\xmsg{\TagB}{\tvarA_2}{\tvarA_2}.\tvarA_1}.\SessionTypeT,
    {!}\xmsg{\TagA}{\tvarB_1}{
    {!}\xmsg{\TagB}{\tvarB_2}{\tvarB_2}.\tvarB_1 \oplus
    {!}\xmsg{\TagC}{\tvarB_3}{\tvarB_3}.\SessionEnd}.\SessionTypeS
) \}$;

\item $\srel_3 = \srel_2 \cup \{
  ({!}\xmsg{\TagB}{\tvarB_2}{\tvarB_2}.\tvarC_1 \oplus
  {!}\xmsg{\TagC}{\tvarB_3}{\tvarB_3}.\SessionEnd,
  {!}\xmsg{\TagB}{\tvarA_2}{\tvarA_2}.\tvarC_1) \}$.
\eoe
\end{iteMize}
\end{example}

\subsection{Type Weight}
\label{sec:algorithmic_weight}

We now address the computation of the weight of an (endpoint) type,
which is the least of its weight bounds or $\infty$ if it has no
weight bound. Unlike the definition of weight bound
(Definition~\ref{def:type_weight}), the algorithm avoids unfoldings of
recursive endpoint types in order to terminate. This imposes a
refinement in the strategy we use for weighing type variables. Recall
that, according to Definition~\ref{def:type_weight}, when determining
$\xweight{\BoundContext_0}\SessionType$ type variables are weighed
either $0$ or $\infty$ according to whether they occur in the context
$\BoundContext_0$ or in $\btv(\SessionType)$ when they are bound in an
input or output prefix. If we avoid unfoldings of recursions, we must
also deal with type variables that are bound by recursive terms
$\trec\tvar.\SessionType$. The idea is that these variables must be
weighed differently, depending on whether they occur \emph{within an
  input prefix} of $\SessionType$ or not. For this reason, we use
another context $\BoundContext$ that contains the subset of type
variables bound by a recursive term and that can be weighed $0$.

\newcommand{\zvars}{\BoundContext_0}
\newcommand{\ivars}{\BoundContext_\infty}

Ultimately, we define a function $\aweight(\BoundContext_0,
\BoundContext, \SessionType)$ by induction on the structure of
$\SessionType$, thus:
\[
\begin{array}{@{}r@{~}c@{~}ll@{}}
  \aweight(\zvars, \BoundContext, \SessionEnd)
  & = &
  0
  \\
  \aweight(\zvars, \BoundContext, \tvar)
  & = &
  \begin{cases}
    0 & \text{if $\tvar \in \zvars \cup \BoundContext$} \\
    \infty & \text{otherwise}
  \end{cases}
  \\
  \aweight(\zvars, \BoundContext, \trec\tvar.\SessionType)
  & = &
  \aweight(\zvars, \BoundContext \cup \{ \tvar \}, \SessionType)
  \\
  \aweight(\zvars, \BoundContext, \InternalChoice{\xmsg{\Tag_i}{\tvar_i}{\Type_i}.\SessionType_i}_{i\in I})
  & = &
  0
  \\
  \aweight(\zvars, \BoundContext, \ExternalChoice{\xmsg{\Tag_i}{\tvar_i}{\Type_i}.\SessionType_i}_{i\in I})
  & = &
  \max
  \{ 1 + \aweight(\zvars, \emptyset, \Type_i),
     \aweight(\zvars, \BoundContext \setminus \{ \tvar_i \},
     \SessionType_i)
  \}_{i\in I}
  \\
  \aweight(\zvars, \BoundContext, \Qualifier~\SessionType)
  & = &
  \aweight(\zvars, \BoundContext, \SessionType)
\end{array}
\]

The first and fourth equations give a null weight to $\SessionEnd$ and
endpoint types in a send state, as expected.
The third equation weighs a recursive term $\trec\tvar.\SessionType$
by weighing the body $\SessionType$ and recording the fact that
$\tvar$ can be given a null weight, as long as $\tvar$ does not occur
in a prefix of $\SessionType$.
The second equation weighs a type variable $\tvar$: if $\tvar$ occurs
in $\BoundContext_0 \cup \BoundContext$, then it means that either
$\tvar$ occurs free in the original endpoint type being weighed and
therefore must be given a null weight, or $\tvar$ is bound in a
recursive term $\trec\tvar.\SessionTypeS$ but it does not occur within
an input prefix of $\SessionTypeS$; if $\tvar$ does not occur in
$\BoundContext_0 \cup \BoundContext$, then it means that either
$\tvar$ was bound in an prefix of an endpoint type in send/receive
state, or it was bound in a recursive term $\trec\tvar.\SessionTypeS$
and it occurs within an input prefix of $\SessionTypeS$.
The fifth equation determines the weight of an endpoint type in
receive state. The rule essentially mimics the corresponding condition
of Definition~\ref{def:type_weight}, but notice that when weighing the
types $\Type_i$ in the prefixes the context $\BoundContext$ is
emptied, since if any of the type variables in it is encountered, then
it must be given an infinite weight.
The last equation simply determines the weight of a qualified endpoint
type to be the weight of the endpoint type itself.

We work out a few simple examples to help clarifying the algorithm:
\begin{iteMize}{$\bullet$}
\item $\aweight(\emptyset, \emptyset,
  {?}\xmsg\Tag\tvar\SessionEnd.\SessionEnd) = \max \{ 1 +
  \aweight(\emptyset, \emptyset, \SessionEnd), \aweight(\emptyset,
  \emptyset, \SessionEnd) \} = 1$;

\item $\aweight(\emptyset, \emptyset,
  {?}\xmsg\Tag\tvar\tvar.\SessionEnd) = \max \{ 1 +
  \aweight(\emptyset, \emptyset, \tvar), \aweight(\emptyset,
  \emptyset, \SessionEnd) \} = \infty$;

\item $\aweight(\emptyset, \emptyset,
  \trec\tvar.{?}\msg\Tag\tvar.\SessionEnd) = \aweight(\emptyset, \{
  \tvar \}, {?}\msg\Tag\tvar.\SessionEnd) = \max \{ 1 +
  \aweight(\emptyset, \emptyset, \tvar), \aweight(\emptyset, \{ \tvar
  \}, \SessionEnd) \} = \infty$;

\item $\aweight(\emptyset, \emptyset,
  \trec\tvar.{?}\msg\Tag\SessionEnd.\tvar) = \aweight(\emptyset, \{
  \tvar \}, {?}\msg\Tag\SessionEnd.\tvar) = \max \{ 1 +
  \aweight(\emptyset, \emptyset, \SessionEnd), \aweight(\emptyset, \{
  \tvar \}, \tvar)\} = 1$.
\end{iteMize}

In the last example, note that the type variable $\tvar$ that
virtually represents the recursive term $\trec\tvar.\SessionType$ is
weighed $0$ even though the whole term turns out to have weight
$1$. The idea is that the proper weight of the whole term will be
computed anyway according to the structure of the term in which
$\tvar$ occurs, and therefore we can safely approximate the weight of
$\tvar$ to $0$.
This property of the algorithm, which is also one of the key
ingredients for proving its correctness, can be formalized as the fact
that the weight of a recursive term and of its unfolding are the same:

\begin{proposition}
\label{prop:unfold_weight}
$\aweight(\BoundContext_0, \emptyset, \trec\tvarA.\SessionType) =
\aweight(\BoundContext_0, \emptyset,
\SessionType\subst{\trec\tvarA.\SessionType}{\tvarA})$.
\end{proposition}

\newcommand{\substitution}{\sigma}

We conclude with the formal statement saying that the algorithm for
computing weights is correct. Its termination is guaranteed as it
works by structural induction over finite terms.

\begin{theorem}
\label{thm:weight_algorithm}
$\xweight\BoundContext\SessionType = \aweight(\BoundContext,
\emptyset, \SessionType)$.
\end{theorem}

\subsection{Type Checking}

In Sections~\ref{sec:algorithmic_subtyping}
and~\ref{sec:algorithmic_weight} we have already presented algorithms
for deciding whether two (endpoint) types are related by subtyping and
for computing the weight of (endpoint) types.
Therefore, there is just one aspect left that makes the type checking
rules in Table~\ref{tab:typing_processes} non-algorithmic, which is
the decomposition of the type environment $\Context$ into $\Context_1
+ \Context_2$ when attempting to derive the judgment $\RecContext;
\BoundContext; \Context \vdash \ProcessP \parop \ProcessQ$ by means of
rule~\rulename{T-Par}. The idea is to look at the free names of
$\ProcessP$ and $\ProcessQ$ that have linear types in $\Context$ and
to split $\Context$ in such a way that $\dom(\Context_1|_\qlin)
\subseteq \fn(\ProcessP)$ and $\dom(\Context_2|_\qlin) \subseteq
\fn(\ProcessQ)$ and $\dom(\Context_1|_\qun) = \dom(\Context_2|_\qun) =
\dom(\Context|_\qun)$. Clearly, if $\ProcessP$ and $\ProcessQ$ share a
free name that has a linear type in $\Context$ there is no way to
derive the judgment $\RecContext; \BoundContext; \Context \vdash
\ProcessP \parop \ProcessQ$.
We omit a formal definition of this splitting since it can be worked
out precisely as explained in~\cite{GayHole05}.


\section{Related work}
\label{sec:related}

\subsection*{Singularity OS}
Copyless message passing is one of the key features adopted by the
Singularity OS~\cite{SingularityOverview05} to compensate the overhead
of communication-based interactions between isolated
processes. Communication safety is enforced by checking processes
against \emph{channel contracts} that are \emph{deterministic},
\emph{autonomous}, and
\emph{synchronizing}~\cite{StengelBultan09,VillardLozesCalcagno09}.
A contract is deterministic if there cannot be two transitions that
differ only for the target state, autonomous if every two transitions
departing from the same state are either two sends or two receives,
and synchronizing if every loop that goes through a final state has at
least one input and one output action.
As argued in~\cite{Fahndrich06}, session types can model channel
contracts quite well because they always correspond by construction to
contracts that are deterministic and autonomous. Session types like
those adopted in this work have just one final state $\SessionEnd$ and
therefore are trivially synchronizing, but this implies that we are
unable to model contracts where a final state has outgoing
transitions. This is not an intrinsic limit of session types (it is
possible to extend session types with more general ``final states'' as
shown in~\cite{CastagnaDezaniGiachinoPadovani09}) and plausibly this
restriction is quite natural in practice (for example, all the channel
contracts in the source code of Singularity OS have final states
without outgoing transitions).

Interestingly, already in~\cite{Fahndrich06} it was observed that
special attention must be deserved to the type of endpoints that are
sent as messages to avoid inconsistencies. In Singularity OS,
endpoints (as well as any other memory block) allocated in the
exchange heap are explicitly tagged with the identifier of their owner
process, and when a block changes owner (because its pointer is sent
in a message) it is the sender's responsibility to update the tag with
the identifier of the receiver process. If this update is not
performed atomically (and it cannot be, for efficiency reasons) the
following can happen: a process sends a message $m$ on an endpoint
$\PointerA$ whose peer $\PointerB$ is owned by some process
$\ProcessP_1$; the sender therefore tags $m$ with $\ProcessP_1$;
simultaneously, $\ProcessP_1$ sends $\PointerB$ away to some other
process $\ProcessP_2$; message $m$ is now formally owned by
$\ProcessP_1$, while in fact it is enqueued in an endpoint that is
owned by $\ProcessP_2$. The authors of~\cite{Fahndrich06} argue that
this inconsistency is avoided if only endpoint in a ``send state''
(those whose type begins with an internal choice) can be sent as
messages. The reason is that, if $\PointerB$ is in a ``send state'',
then $\PointerA$, which must have a dual type, is in a ``receive
state'', and therefore it is not possible to send message $m$ on
it.
\REVISION{
  In this respect, our work shows that the ``send state'' restriction
  has deeper motivations that go beyond the implementation details of
  ownership transfer, it gives formal evidence that the restriction
  devised in~\cite{Fahndrich06} is indeed safe, because endpoints in a
  ``send state'' always have a null weight, and it shows how to handle
  a more expressive type system with polymorphic endpoint types.
}

\subsection*{Early Type-Theoretic Formalizations of Singularity OS}
This work improves previous formalizations of Singularity OS presented
in~\cite{BonoMessaPadovani11,BonoPadovani11}.
The main differences regard polymorphic and unrestricted endpoint
types and the modeling of \Sing{}'s \Code{expose}.

Polymorphic endpoint types increase the flexibility of the type system
and are one of the features of Singularity OS, in the form of
polymorphic contracts, documented in the design note dedicated to
channels~\cite{SDN5}. The most interesting aspect of polymorphic
endpoint types is their interaction with the ownership invariant (see
the example~\eqref{eq:micidiale2}) and with the computation of type
weights.
Polymorphism was not considered in~\cite{BonoMessaPadovani11}, and
in~\cite{BonoPadovani11} we have introduced a \emph{bounded} form of
polymorphism, along the lines of~\cite{Gay08}, but we did not impose
any constraint on the instantiation of type variables without bound
which were all estimated to have infinite weight. This proved to be
quite restrictive (a simple forwarder process like the one in
Example~\ref{ex:forwarder} would be ill-typed).  The crucial
observation of the present type system is that type variables denote
``abstract'' values that can only be passed around. So, just as values
that are passed around must have a finite-weight type, it makes sense
to impose the same restriction when instantiating type variables.
For the sake of simplicity, in the present work we have dropped type
bounds for type variables. This allowed us to define the subtyping
algorithm as a relatively simple extension of the standard subtyping
algorithm for session types~\cite{GayHole05}. It should be possible to
work out a subtyping algorithm for bounded, polymorphic, recursive
endpoint types, possibly adapting related algorithms defined for
functional types~\cite{Jeffrey01,ColazzoGhelli05}, although the
details might be quite involved.

In~\cite{BonoMessaPadovani11,BonoPadovani11} only linear endpoint
types were considered. However, as pointed out by some referees, a
purely linear type system is quite selective on the sort of constructs
that can be effectively modeled with the calculus. For this reason, in
the present version we have introduced unrestricted endpoint types in
addition to linear ones, with the understanding that other kinds of
unrestricted data types (such as the primitive types of boolean or
integer values) can be accommodated just as easily. We have shown that
unrestricted endpoint types can be used for representing the type of
non-linear resources such as permanent services and functions and we
have also been able to implement the \Code{TCell} type constructor of
\Sing{} (Example~\ref{ex:tref}).
Interestingly, the introduction of unrestricted endpoint types
required very little change to the process language (only a different
$\openChannel$ primitive) and no change at all to the heap model.

The remaining major difference between~\cite{BonoMessaPadovani11} and
this work is the lack of any \Code{expose} primitive in the process
calculus, which is used in the \Sing{} compiler to keep track of
memory ownership.
To illustrate the construct, consider the code fragment
\begin{SingSharp}
  expose (a) {
    b.Arg(*a);
    *a = new[ExHeap] T();
  }
\end{SingSharp}
which dereferences a cell \Code{a} and sends its content on endpoint
\Code{b}. After the \Code{b.Arg(*a)} operation the process no longer
owns \Code{*a} but it still owns \Code{a}. Therefore, the ownership
invariant could be easily violated if the process were allowed to
access \Code{*a} again. To prevent this, the \Sing{} compiler allows
(linear) pointer dereferentiation only within \Code{expose}
blocks. The \Code{expose (a)} block temporarily transfers the
ownership of \Code{*a} from \Code{a} to the process exposing \Code{a}
and is well-typed if the process still owns \Code{*a} at the end of
block. In this example, the only way to regain ownership of \Code{*a}
is to assign it with the pointer to another object that the process
owns.
In~\cite{BonoMessaPadovani11} we showed that all we need to capture
the static semantics of \Code{expose} blocks is to distinguish cells
with type $\Ref\Type$ (whose content, of type $\Type$, is owned by the
cell) from cells with type $\Ref\Open$ (whose content is owned
directly by the process). At the beginning of the expose block, the
type of \Code{a} turns from $\Ref\Type$ to $\Ref\Open$; within the
block it is possible to (linearly) use \Code{*a}; at the end of the
block, \Code{*a} is assigned with the pointer to a newly allocated
object that the process owns, thus turning \Code{a}'s type from
$\Ref\Open$ back to some $\Ref\TypeS$.
In other words, cell types (and other object types) are simple
behavioral types that can be easily modeled in terms of polymorphic
endpoint types. In~\cite{BonoPadovani11} we have shown that the
endpoint type
\[
  \mathtt{CellT} =
  \trec\tvarA.
  ({!}\tmsg{\mathtt{Set}}{\tvarB}{\qlin~\tvarB}.{?}\msg{\mathtt{Get}}{\qlin~\tvarB}.\tvarA
  \oplus
  {!}\msg{\mathtt{Free}}{}.\SessionEnd)
\]
corresponds to the open cell type $\Ref\Open$ that allows for setting
a cell with a value of arbitrary type and for freeing the cell. Once
the cell has been set, its type turns to some
\[
  {?}\msg{\mathtt{Get}}\Type.\mathtt{CellT}
\]
corresponding to the cell type $\Ref\Type$ that only allows for
retrieving its content. The cell itself can be easily modeled as a
process that behaves according to $\co{\mathtt{CellT}}$, as shown
in~\cite{BonoPadovani11}.

As a final note, in~\cite{BonoMessaPadovani11} we have shown how to
accommodate the possibility of closing endpoints ``in advance'' (when
their type is different from $\SessionEnd$), since this feature is
available in \Sing{}. Overall, it seems like the issues it poses
exclusively concern the implementation details rather than the
peculiar characteristics of the formal model. Consequently, we have
decided to drop this feature in the present paper.

\subsection*{Type Weight}
Other works~\cite{Fahndrich06,GayVasconcelos10} introduce apparently
similar, finite-size restrictions on session types. In these cases,
the size estimates the maximum number of enqueued messages in an
endpoint and it is used for efficient, static allocation of endpoints
with finite-size type. Our weights are unrelated to the size of queues
and concern the length of chains of pointers involving queues. For
example, in~\cite{GayVasconcelos10} the session type $\SessionTypeT =
\trec\tvar.{?}\msg\Tag{\qlin~\tvar}.\SessionEnd$ has size 1 (there can
be at most one message of type $\qlin~\SessionTypeT$ in the queue of
an endpoint with type $\SessionTypeT$) and the session type
$\SessionTypeS = \trec\tvar.{?}\msg\Tag{\qlin~\SessionEnd}.\tvar$ has
size $\infty$ (there can be any number of messages, each of type
$\qlin~\SessionEnd$, in the queue of an endpoint with type
$\SessionTypeS$). In our theory we have just the opposite, that is
$\weight\SessionTypeT = \infty$ and $\weight\SessionTypeS = 1$.
Despite these differences, the workaround we have used to bound the
weight of endpoint types (Example~\ref{ex:lists}) can also be used to
bound the size of session types as well, as pointed out
in~\cite{GayVasconcelos10}.

\subsection*{Logic-Based Analysis}
A radically different approach for the static analysis of Singularity
processes is given
by~\cite{VillardLozesCalcagno09,VillardLozesCalcagno10}, where the
authors develop a proof system based on a variant of \emph{separation
  logic}~\cite{OHearnReynoldsYang01}.  The proof system permits the
derivation of Hoare triples of the form $\{A\}~P~\{B\}$ where $P$ is a
program and $A$ and $B$ are logical formulas describing the state of
the heap before and after the execution of $P$. A judgment $\{
\mathsf{emp} \}~P~\{ \mathsf{emp} \}$ indicates that if $P$ is
executed in the empty heap (the pre-condition $\mathsf{emp}$), then it
leaks no memory (the post-condition $\mathsf{emp}$).  However, leaks
in~\cite{VillardLozesCalcagno09} manifest themselves only when both
endpoints of any channel have been closed. In particular, it is
possible to prove that the function~\Code{foo} in
Section~\ref{sec:singularity} is safe, although it may indeed leak
some memory.
This problem has been subsequently recognized and solved
in~\cite{Villard11}. Roughly, the solution consists in forbidding the
output of a message unless it is possible to prove (in the logic) that
the queue that is going to host the message is reachable from the
content of the message itself. In principle this condition is optimal,
in the sense that it should permit every safe output. However, it
relies on the knowledge of the identity of endpoints, that is a very
precise information that is not always available. For this reason,
\cite{Villard11} also proposes an approximation of this condition,
consisting in tagging endpoints of a channel with distinct
\emph{roles} (basically, what are called \emph{importing} and
\emph{exporting} views in Singularity). Then, an endpoint can be
safely sent as a message only if its role matches the one of the
endpoint on which it is sent. This solution is incomparable to the one
we advocate -- restricting the output to endpoints with finite-weight
type -- suggesting that it may be possible to work out a combination
of the two.
In any case, neither~\cite{VillardLozesCalcagno09}
nor~\cite{Villard11} take into account polymorphism.

\subsection*{Global Progress}
There exist a few works on session
types~\cite{BCDDDY08,CastagnaDezaniGiachinoPadovani09} that guarantee
a global progress property for well-typed systems where the basic idea
is to impose an order on channels to prevent circular dependencies
that could lead to a deadlock.
Not surprisingly, the critical processes such as~\eqref{eq:micidiale}
that we rule out thanks to the finite-weight restriction on the type
of messages are ill typed in these works.
It turns out that a faithful encoding of~\eqref{eq:micidiale} into the
models proposed in these works is impossible, because the
$\openChannel({\cdot},{\cdot})$ primitive we adopt (and that mimics
the corresponding primitive operation in Singularity OS) creates
\emph{both} endpoints of a channel within the same process, while the
session initiation primitives
in~\cite{BCDDDY08,CastagnaDezaniGiachinoPadovani09} associate the
fresh endpoints of a newly opened session to different processes
running in parallel. This invariant -- that the same process cannot
own more than one endpoint of the same channel -- is preserved in
well-typed processes because of a severe restriction: whenever an
endpoint $\PointerC$ is received, the continuation process cannot use
any endpoint other than $\PointerC$ and the one from which $\PointerC$
was received.


\section{Conclusions}
\label{sec:conclusions}

We have defined the static analysis for a calculus where processes
communicate through the exchange of \emph{pointers}. Verified
processes are guaranteed to be free from memory faults, they do not
leak memory, and do not fail on input actions.  Our type system has
been inspired by session type theories. The basic idea of session
types, and of behavioral types in general, is that operating on a
(linearly used) value may change its type, and thus the capabilities
of that value thereafter. Endpoint types express the capabilities of
endpoints, in terms of the type of messages that can be sent or
received and in which order.
We have shown that, in the copyless message passing paradigm,
linearity alone is not enough for preventing memory leaks, but also
that endpoint types convey enough information -- their \emph{weight}
-- to devise a manageable type system that detects potentially
dangerous processes: it is enough to restrict send operations so that
only endpoint with a finite-weight type can be sent as messages and
only finite-weight endpoint types can instantiate type variables. This
restriction can be circumvented in a fairly easy and general way at
the cost of a few extra communications, still preserving all the nice
properties of the type system (Example~\ref{ex:lists}).

We claim that our calculus provides a fairly comprehensive
formalization of the peculiar features of \Sing{}, among which are the
explicit memory management of the exchange heap, the controlled
ownership of memory allocated on the exchange heap, and channel
contracts. We have also shown how to accommodate some advanced
features of the \Sing{} type system, namely (the lack of)
\Code{[Claims]} annotations, the \Code{TCell} type constructor that
allows for the sharing of linear pointers, and polymorphic channel
contracts. In prior work~\cite{BonoPadovani11} we had already shown
how polymorphic endpoint types permit the encoding of \Code{expose}
blocks for accessing linear pointers stored within other objects
allocated on the exchange heap.
Interestingly, previous studies on Singularity channel
contracts~\cite{Fahndrich06} had already introduced a restriction on
send operations so that only endpoints in a \emph{send-state}, those
whose type begins with an internal choice, can be safely sent as
messages.  There the restriction was motivated by the implementation
of ownership transfer in Singularity, where it is the sender's
responsibility to explicitly tag sent messages with their new
owner. We have shown that there are more reasons for being careful
about which endpoints can be sent as messages and that the send-state
restriction is a sound approximation of our finite-weight restriction,
because endpoints in a send-state always have a null weight.

On a more technical side, we have also developed a decidable theory of
polymorphic, recursive behavioral types. Our theory is incomparable
with that developed in~\cite{Gay08}: we handle recursive behavioral
types, whereas \cite{Gay08} only considers finite ones; polymorphism
in~\cite{Gay08} is bounded, while it is unrestricted in our case.
The subtyping relation that takes into account both recursive
behaviors and bounds is in fact quite straightforward to define
(see~\cite{BonoPadovani11}), but its decision algorithm appears to be
quite challenging. As observed in~\cite{Gay08}, bounded polymorphic
session types share many properties with the type language in system
\Fsub~\cite{CardelliMartiniMitchellScedrov94}, and subtyping
algorithms for extensions of \Fsub{} with recursive types are well
known for their complexity~\cite{Jeffrey01,ColazzoGhelli05}. We leave
the decision algorithm for subtyping of behavioral types with
recursion and bounded polymorphism as future work.


\paragraph{Acknowledgments.} We are grateful to Lorenzo Bettini for
discussions on the notion of memory leak, to Nobuko Yoshida for
comments on an early version of this paper, and to the anonymous
referees for the detailed and useful reviews.

\bibliographystyle{plain}
\bibliography{main}

\appendix

\section{Supplement to Section~\ref{sec:types}}
\label{sec:extra_types}

\begin{proposition}[Proposition~\ref{prop:wf_types}]
  The following properties hold:
\begin{enumerate}[\em(1)]
\item $\co{\co{\SessionType}} = \SessionType$.

\item $\emptyset \wfdash \SessionType$ implies that $\SessionType
  \dual \co\SessionType$ and $\emptyset \wfdash \co\SessionType$.

\item $\BoundContext; \{ \tvar \} \wfdash \SessionTypeT$ and
  $\BoundContext \wfdash \SessionTypeS$ imply $\BoundContext \wfdash
  \SessionTypeT\subst\SessionTypeS\tvar$.

\item $\emptyset; \{\tvar\} \wfdash \SessionTypeT$ and $\emptyset
  \wfdash \SessionTypeS$ imply
  $\co{\SessionTypeT\subst{\SessionTypeS}{\tvar}} =
  \co{\SessionTypeT}\subst{\SessionTypeS}{\tvar}$.
\end{enumerate}
\end{proposition}
\begin{proof}[Proof sketch]
  \REVISION{Item~(1) is proved by induction on $\SessionTypeT$. The
    only interesting case is when $\SessionTypeT \seq
    \trec\tvar.\SessionTypeS$. Then, by definition of dual, we have
    $\co\SessionTypeT \seq
    \trec\tvar.\co{\SessionTypeS\psubst{\SessionTypeT}{\tvar}}$ and
    now:
\[
\begin{array}{r@{~}c@{~}l@{\qquad}l}
  \co{\co\SessionTypeT} & \seq &
  \trec\tvar.\co{\co{\SessionTypeS\psubst{\SessionTypeT}{\tvar}}\psubst{\co\SessionTypeT}{\tvar}}
  & \text{(by definition of dual)}
  \\
  & = &
  \trec\tvar.\co{\co{\SessionTypeS\psubst{\SessionTypeT}{\tvar}}} 
  & \text{(by definition of inner substitution)}
  \\
  & = &
  \trec\tvar.(\co{\co\SessionTypeS}\psubst{\SessionTypeT}{\tvar})
  & \text{(because inner substitution and dual commute)}
  \\
  & = &
  \trec\tvar.(\SessionTypeS\psubst{\SessionTypeT}{\tvar})
  & \text{(by induction hypothesis)}
  \\
  & = &
  \trec\tvar.\SessionTypeS \seq \SessionTypeT
  & \text{(by folding the recursion)}
\end{array}
\]

Item~(2) relies on the fact that duality and unfolding commute. Indeed
we have~\hypo* $\co{\trec\tvar.\SessionTypeT} \seq
\trec\tvar.\co{\SessionTypeT\psubst{\trec\tvar.\SessionTypeT}{\tvar}}$
and now:
\[
\begin{array}{@{}r@{~}c@{~}l@{\quad}l@{}}
  \co{\SessionTypeT\subst{\trec\tvar.\SessionTypeT}{\tvar}} & = &
  \co{\SessionTypeT\psubst{\trec\tvar.\SessionTypeT}{\tvar}\subst{\trec\tvar.\SessionTypeT}{\tvar}}
  & \text{(def. of inner substitution)}
  \\
  & = & 
  \co{\SessionTypeT\psubst{\trec\tvar.\SessionTypeT}{\tvar}}\subst{\co{\trec\tvar.\SessionTypeT}}{\tvar}
  & \text{(by def. of dual)}
  \\
  & = &
  \co{\SessionTypeT\psubst{\trec\tvar.\SessionTypeT}{\tvar}}\subst{
    \trec\tvar.\co{\SessionTypeT\psubst{\trec\tvar.\SessionTypeT}{\tvar}}
  }{
    \tvar
  }
  & \text{(by \hypo*)}
  \\
  & = &
  \trec\tvar.\co{\SessionTypeT\psubst{\trec\tvar.\SessionTypeT}{\tvar}}
  & \text{(by folding the recursion)}
  \\
  & \seq &
  \co{\trec\tvar.\SessionTypeT}
  & \text{(by \hypo*)}
\end{array}
\]
}

In proving items~(2--4) it is also needed the fact that $\BoundContext
\wfdash \SessionTypeT$ implies $\ftv(\SessionTypeT) \subseteq
\BoundContext$ and these free type variables can only occur within
prefixes of $\SessionTypeT$.
\REVISION{We let the reader fill in the remaining details.}
\end{proof}

\begin{proposition}[Proposition~\ref{prop:subt_dual}]
  Let $\emptyset \wfdash \SessionTypeT$ and $\emptyset \wfdash
  \SessionTypeS$. Then $\SessionTypeT \subt \SessionTypeS$ if and only
  if $\co\SessionTypeS \subt \co\SessionTypeT$.
\end{proposition}
\begin{proof}
  It is enough to show that
  \[ {\srel} \eqdef {\subt} \cup \{ (\co\SessionTypeS,
  \co\SessionTypeT) \mid \SessionTypeT \subt \SessionTypeS \land
  \emptyset; \BoundContext \wfdash \SessionTypeT \land \emptyset;
  \BoundContext \wfdash \SessionTypeS \}
\]
is a coinductive subtyping.
Suppose $(\co\SessionTypeS, \co\SessionTypeT) \in {\srel}$ where
\hypo1 $\SessionTypeT \subt \SessionTypeS$ and \hypo2 $\emptyset;
\BoundContext \wfdash \SessionTypeT$ and $\emptyset; \BoundContext
\vdash \SessionTypeS$. We reason by cases on the shape of
$\SessionTypeT$ and $\SessionTypeS$:
\begin{iteMize}{$\bullet$}
\item ($\SessionTypeT = \SessionTypeS = \SessionEnd$) We conclude
  immediately since $\co\SessionEnd = \SessionEnd$.

\item ($\SessionTypeT = \SessionTypeS = \tvar$) This case is
  impossible because of the hypothesis \hypo2.

\item ($\SessionTypeT =
  \ExternalChoice{\xmsg{\Tag_i}{\tvar_i}{\TypeT_i}.\SessionTypeT_i}_{i\in
    I}$ and $\SessionTypeS =
  \ExternalChoice{\xmsg{\Tag_j}{\tvar_j}{\TypeS_j}.\SessionTypeS_j}_{j\in
    J}$)
  Then $\co\SessionTypeT =
  \InternalChoice{\xmsg{\Tag_i}{\tvar_i}{\TypeT_i}.\co{\SessionTypeT_i}}_{i\in
    I}$ and $\SessionTypeS =
  \InternalChoice{\xmsg{\Tag_i}{\tvar_i}{\TypeS_i}.\co{\SessionTypeS_i}}_{i\in
    J}$.
  From \hypo1 we deduce $I\subseteq J$ and $\TypeT_i \subt \TypeS_j$
  and $\SessionTypeT_i \subt \SessionTypeS_i$ for every $i\in I$.
  From \hypo2 we deduce $\emptyset; \BoundContext, \tvar_i \wfdash
  \SessionTypeT_i$ and $\emptyset; \BoundContext, \tvar_i \wfdash
  \SessionTypeS_i$.
  By definition of $\srel$ we conclude $(\co{\SessionTypeS_i},
  \co{\SessionTypeT_i}) \in \srel$ for every $i\in I$.

\item ($\SessionTypeT =
  \InternalChoice{\xmsg{\Tag_i}{\tvar_i}{\TypeT_i}.\SessionTypeT_i}_{i\in
    I}$ and $\SessionTypeS =
  \InternalChoice{\xmsg{\Tag_j}{\tvar_j}{\TypeS_j}.\SessionTypeS_j}_{j\in
    J}$) Dual of the previous case.
\qedhere
\end{iteMize}
\end{proof}


\section{Supplement to Section~\ref{sec:type_system}}
\label{sec:extra_type_system}

Before addressing subject reduction and soundness we prove a series of
auxiliary results. The first one states an expected property of
endpoint types, namely that the weight
$\xweight{\{\tvar\}}\SessionType$ where we take the free occurrences
of $\tvar$ to have null weight remains finite if we replace the same
occurrences of $\tvar$ with an arbitrary, but finite-weight endpoint
type $\SessionTypeS$ (recall that
$\SessionTypeT\subst\SessionTypeS\tvar$ is a capture-avoiding
substitution).

\begin{proposition}
\label{prop:weight_subst}
Let $\max\{\xweight{\{\tvar\}}\SessionTypeT, \weight{\SessionTypeS}\}
< \infty$.  Then $\weight{\SessionTypeT\subst\SessionTypeS\tvar} <
\infty$.
\end{proposition}
\begin{proof}
We show that $\{ \tvar \} \vdash \SessionTypeT \wbound m$ and
  $\SessionTypeS \wbound n$ imply
  $\SessionTypeT\subst\SessionTypeS\tvar \wbound m + n$.
It is enough to show that
\[
  {\wrel} \eqdef \{ (\EmptyBoundContext, \SessionTypeT'\subst{\SessionTypeS}{\tvar}, m + n)
  \mid
  \exists m\in\natset : \{ \tvar \} \vdash \SessionTypeT' \wbound m
  \}
\]
is a coinductive weight bound.
Observe that $\SessionTypeT'' \wbound n$ implies $(\EmptyBoundContext,
\SessionTypeT'', n) \in {\wrel}$.
Let $(\EmptyBoundContext, \SessionTypeT'', k) \in {\wrel}$. Then there
exist $\SessionTypeT'$ and $m$ such that $\SessionTypeT'' =
\SessionTypeT'\subst\SessionTypeS\tvar$ and $k = m + n$ and \hypo* $\{
\tvar \} \vdash \SessionTypeT' \wbound m$.
We reason by cases on $\SessionTypeT'$ assuming, without loss of
generality, that $(\{ \tvar \} \cup \ftv(\SessionTypeS)) \cap
\btv(\SessionTypeT') = \emptyset$:
\begin{iteMize}{$\bullet$}
\item ($\SessionTypeT' = \SessionEnd$) Trivial.

\item ($\SessionTypeT' = \tvar$) Then
  $\SessionTypeT'\subst\SessionTypeS\tvar = \SessionTypeS$ and from
  the hypothesis $\SessionTypeS \wbound n$ we conclude $\SessionTypeS
  \wbound m + n$.

\item ($\SessionTypeT' = \tvarB \ne \tvarA$) This case is impossible
  for it contradicts \hypo*.

\item ($\SessionTypeT' =
  \InternalChoice{\xmsg{\Tag_i}{\tvar_i}{\Type_i}.\SessionTypeT_i}_{i\in
    I}$) Trivial.

\item ($\SessionTypeT' =
  \ExternalChoice{\xmsg{\Tag_i}{\tvar_i}{\Qualifier_i~\SessionTypeS_i}.\SessionTypeT_i}_{i\in
    I}$)
  From \hypo* we deduce $m > 0$ and $\{ \tvar \} \vdash
  \SessionTypeS_i \wbound m - 1$ and $\{ \tvar \} \vdash
  \SessionTypeT_i \wbound m$ for every $i\in I$.
  By definition of $\wrel$ we conclude $(\EmptyBoundContext,
  \SessionTypeS_i\subst\SessionTypeS\tvar, (m - 1) + n) \in {\wrel}$
  and $(\EmptyBoundContext, \SessionTypeT_i\subst\SessionTypeS\tvar, m
  + n) \in {\wrel}$ for every $i\in I$.
\qedhere
\end{iteMize}
\end{proof}

\noindent Type variable instantiation does not affect the subtyping relation:

\begin{proposition}
\label{prop:subt_subst}
The following properties hold:
\begin{enumerate}[\em(1)]
\item $\SessionTypeT_1 \subt \SessionTypeT_2$ implies
  $\SessionTypeT_1\subst\SessionTypeS\tvar \subt
  \SessionTypeT_2\subst\SessionTypeS\tvar$;

\item $\TypeT_1 \subt \TypeT_2$ implies
  $\TypeT_1\subst\SessionTypeS\tvar \subt
  \TypeT_2\subst\SessionTypeS\tvar$.
\end{enumerate}
\end{proposition}
\begin{proof}[Proof sketch]
  Follows from the fact that a free type variable $\tvar$ can only be
  related to itself. The details are left as an technical exercise.
\end{proof}

We now turn to a series of standard auxiliary results of type
preservation under structural congruence and various forms of
substitutions.

\begin{lemma}
  \label{lem:cong}
  Let $\Context \vdash \ProcessP$ and $\ProcessP \equiv
  \ProcessQ$. Then $\Context \vdash \ProcessQ$.
\end{lemma}
\begin{proof}
  By case analysis on the derivation of $\ProcessP \equiv \ProcessQ$.
\end{proof}

\begin{lemma}[type substitution]
\label{lem:type_subst}
If $\RecContext; \BoundContext, \tvar; \Context \vdash \Process$ and
$\emptyset \wfdash \SessionTypeS$ and $\weight\SessionTypeS < \infty$,
then $\RecContext; \BoundContext; \Context\subst{\SessionTypeS}{\tvar}
\vdash \Process\subst\SessionTypeS\tvar$.
\end{lemma}
\begin{proof}[Proof sketch]
  Straightforward induction on the derivation of $\RecContext;
  \BoundContext, \tvar; \Context \vdash \Process$, using
  Propositions~\ref{prop:weight_subst} and~\ref{prop:subt_subst}
  wherever necessary.
\end{proof}

\begin{lemma}[value substitution]
\label{lem:value_subst}
If $\RecContext; \BoundContext; \Context, \Var : \TypeT \vdash
\Process$ and $\Context + \Value : \TypeS$ is defined and well formed
and $\TypeS \subt \TypeT$, then $\RecContext; \BoundContext; \Context
+ \Value : \TypeS \vdash \Process\subst{\Value}{\Var}$.
\end{lemma}
\begin{proof}
  By induction on the derivation of $\RecContext; \BoundContext;
  \Context, \Var : \TypeT \vdash \Process$ and by cases on the last
  rule applied.
  We only show the proof of the \rulename{T-Send} case, the others
  being simpler or trivial. In the \rulename{T-Send} case we have:
\begin{iteMize}{$\bullet$}
\item $\Process =
  \xsend{\NameA}{\Tag}{\SessionTypeS}{\NameB}.\Process'$;

\item $\Context, \Var : \TypeT = (\Context'', \NameA :
  \Qualifier~\InternalChoice{\tmsg{\Tag_i}{\tvar_i}{\Type_i}.\SessionType_i}_{i\in
    I}) + \NameB : \TypeS''$;

\item $\RecContext; \BoundContext; \Context'', \NameA :
  \Qualifier~\SessionTypeT_k\subst{\SessionTypeS}{\tvar_k} \vdash
  \Process'$.
\end{iteMize}

We can assume $\Var \in \dom(\Context'') \cup \{ \NameA \}$ for
otherwise $\Var \not\in \fn(\Process')$ and there is nothing left to
prove.
Let $\Context'', \NameA :
\Qualifier~\SessionTypeT_k\subst{\SessionTypeS}{\tvar_k} = \Context',
\Var : \TypeT'$ for some $\Context'$ and $\TypeT'$.
In order to apply the induction hypothesis and deduce $\RecContext;
\BoundContext; \Context' + \Value : \TypeS' \vdash
\Process'\subst\Value\Var$, we must find $\TypeS'$ such that \hypo{a}
$\TypeS' \subt \TypeT'$ and \hypo{b} $\Context' + \Value : \TypeS'$ is
defined and well formed.
Observe that the type of $\Var$, $\TypeT'$, may change from the
conclusion to the premise of the rule if $\Var = \NameA$.
We distinguish the following sub-cases:
\begin{iteMize}{$\bullet$}
\item ($\Value \ne \NameA, \NameA \ne \Var$)
  For \hypo{a}, we deduce $\TypeT' = \TypeT$ and we conclude by taking
  $\TypeS' = \TypeS$.
  For \hypo{b}, then either $\Value \not\in \dom(\Context'')$ or
  $\qun(\Context''(\Value))$. In both cases we conclude that
  $\Context' + \Value : \TypeS'$ is defined and well formed.

\item ($\Value \ne \NameA, \NameA = \Var$)
  For \hypo{a}, we deduce $\TypeT =
  \Qualifier~\InternalChoice{\tmsg{\Tag_i}{\tvar_i}{\TypeT_i}.\SessionTypeT_i}_{i\in
    I}$.
  From $\TypeS \subt \TypeT$, we deduce $\TypeS =
  \Qualifier'~\InternalChoice{\tmsg{\Tag_i}{\tvar_i}{\TypeS_i}.\SessionTypeS_i}_{i\in
    I \cup J}$ and $\Qualifier' \leq \Qualifier$ and $\SessionTypeS_i
  \subt \SessionTypeT_i$ for $i\in I$.
  By Proposition~\ref{prop:subt_subst}(1) we obtain
  $\SessionTypeS_k\subst\SessionTypeS{\tvar_k} \subt
  \SessionTypeT_k\subst\SessionTypeS{\tvar_k}$ and we conclude by
  taking $\TypeS' = \SessionTypeS_k\subst\SessionTypeS{\tvar_k}$.
  For \hypo{b} we can reason as for the previous case.

\item ($\Value = \NameA$)
  Since $\NameA \in \dom(\Context)$, then $\TypeS = \Context(\NameA)$
  and $\Qualifier = \qun$.
  Since the $\qun$ qualifier can only be applied to invariant types,
  it must be the case that
  $\SessionTypeT_k\subst\SessionTypeS{\tvar_k} = \TypeS$.
  We conclude \hypo{a} by taking $\TypeS' = \TypeS$ and \hypo{b}
  follows immediately.
\qedhere
\end{iteMize}
\end{proof}

\begin{lemma}[weakening]
\label{lem:weakening}
If $\RecContext; \BoundContext; \Context \vdash \ProcessP$ and
$\qun(\Context')$, then $\RecContext, \RecContext'; \BoundContext,
\BoundContext'; \Context, \Context' \vdash \ProcessP$.
\end{lemma}
\begin{proof}
  Straightforward induction on the derivation of $\RecContext;
  \BoundContext; \Context \vdash \ProcessP$.
\end{proof}

\begin{lemma}[process substitution]
\label{lem:process_subst}
Let \hypo1 $\RecContext, \{ \RecVar \mapsto (\BoundContext; \Context)
\}; \BoundContext; \Context \vdash \ProcessQ$.  Then \hypo{2}
$\RecContext, \{ \RecVar \mapsto (\BoundContext; \Context) \},
\RecContext'; \BoundContext'; \Context' \vdash \ProcessP$ implies
$\RecContext, \RecContext'; \BoundContext'; \Context' \vdash
\ProcessP\subst{\rec\RecVar.\ProcessQ}{\RecVar}$.
\end{lemma}
\begin{proof}
\newcommand{\SUBST}{\subst{\rec\RecVar.\ProcessQ}{\RecVar}}
By induction on $\ProcessP$. Whenever we encounter some bound
name/type variable/process variable in $\ProcessP$ we assume, without
loss of generality, that it does not occur free in $\ProcessQ$:
\begin{iteMize}{$\bullet$}
\item ($\ProcessP = \idle$)
  Then $\ProcessP\SUBST = \idle$.
  From \hypo{2} and \rulename{T-Idle} we deduce $\qun(\Context')$. We
  conclude with an application of \rulename{T-Idle}.

\item ($\ProcessP = \RecVar$)
  Then $\ProcessP\SUBST =
  \rec\RecVar.\ProcessQ$.
  From \hypo{2} and \rulename{T-Var} we deduce:
  \STEP{$\BoundContext' = \BoundContext, \BoundContext''$;}
  \STEP{$\Context' = \Context, \Context''$;}
  \STEP{$\qun(\Context'')$.}

  \noindent From \hypo{1} and Lemma~\ref{lem:weakening} we obtain $\RecContext,
  \{ \RecVar \mapsto (\BoundContext; \Context) \}, \RecContext';
  \BoundContext'; \Context' \vdash \ProcessQ$.
  We conclude with an application of \rulename{T-Rec}.

\item ($\ProcessP = \RecVarY \ne \RecVar$)
  Then $\ProcessP\SUBST = \RecVarY$ and we conclude immediately from
  \rulename{T-Var}.

\item ($\ProcessP = \closeChannel(\Name)$)
  Then $\ProcessP\SUBST = \closeChannel(\Name)$ and we conclude
  immediately from \rulename{T-Close}.

\item ($\ProcessP = \ProcessP_1 \choice \ProcessP_2$)
  Then $\ProcessP\SUBST = \ProcessP_1\SUBST \choice
  \ProcessP_2\SUBST$.  From \hypo{2} and \rulename{T-Choice} we
  deduce:
  \STEP{$\RecContext, \{ \RecVar \mapsto (\BoundContext; \Context) \},
    \RecContext'; \BoundContext'; \Context' \vdash \ProcessP_i$ for
    $i=1,2$.}

  By induction hypothesis we obtain: \STEP{$\RecContext, \RecContext';
    \BoundContext'; \Context' \vdash \ProcessP_i\SUBST$ for $i=1,2$.}

  We conclude with an application of \rulename{T-Choice}.

\item ($\ProcessP = \ProcessP_1 \parop \ProcessP_2$)
  Then $\ProcessP\SUBST = \ProcessP_1\SUBST \parop
  \ProcessP_2\SUBST$. From \hypo{2} and \rulename{T-Par} we deduce
  \STEP{$\Context' = \Context_1 + \Context_2$ and}
  \STEP{$\RecContext, \{ \RecVar \mapsto (\BoundContext; \Context) \},
    \RecContext'; \BoundContext'; \Context_i \vdash \ProcessP_i$ for
    $i=1,2$.}

  By induction hypothesis: \STEP{$\RecContext, \RecContext';
    \BoundContext'; \Context_i \vdash \ProcessP_i\SUBST$.}

  We conclude with an application of \rulename{T-Par}.

\item ($\ProcessP = \openChannel(\PointerA : \SessionTypeT, \PointerB
  : \SessionTypeS).\ProcessP'$)
  Then $\ProcessP\SUBST = \openChannel(\PointerA : \SessionTypeT,
  \PointerB : \SessionTypeS).(\ProcessP'\SUBST)$.
  From \hypo{2} and \rulename{T-Open Linear Channel} we deduce:
  \STEP{$\BoundContext' \wfdash \SessionTypeT$;}
  \STEP{$\RecContext, \{ \RecVar \mapsto (\BoundContext; \Context) \},
    \RecContext'; \BoundContext'; \Context', \PointerA :
    \qlin~\SessionTypeT, \PointerB : \qlin~\SessionTypeS \vdash
    \ProcessP'$;}
  \STEP{$\SessionTypeS = \co\SessionTypeT$.}

  By induction hypothesis: \STEP{$\RecContext, \RecContext';
    \BoundContext'; \Context', \PointerA : \qlin~\SessionTypeT,
    \PointerB : \qlin~\SessionTypeS \vdash \ProcessP'\SUBST$.}

  We conclude with an application of \rulename{T-Open Linear Channel}.

\item ($\ProcessP = \openChannel(\Pointer : \SessionType).\ProcessP'$)
  Similar to the previous case.

\item ($\ProcessP = \sum_{i\in I}
  \xreceive\NameA{\Tag_i}{\tvar_i}{\Var_i:\TypeS_i}.\Process_i$) 
  Then \[\ProcessP\SUBST = \sum_{i\in I}
  \xreceive\NameA{\Tag_i}{\tvar_i}{\Var_i:\Type_i}.(\Process_i\SUBST).\]
  From \hypo{2} and \rulename{T-Receive} we deduce:
  \STEP{$\Context' = \Context'', \NameA : \qlin~\ExternalChoice{\tmsg{\Tag_i}{\tvar_i}{\TypeS_i}.\SessionTypeT_i}_{i\in J}$;}
  \STEP{$J \subseteq I$;}
  \STEP{$\TypeS_i \subt \TypeT_i$ for every $i \in J$;}
  \STEP{$\RecContext, \{ \RecVar \mapsto (\BoundContext; \Context) \},
    \RecContext'; \BoundContext', \tvar_i; \Context', \NameA :
    \qlin~\SessionType_i, \Var_i : \Type_i \vdash \Process_i$ for
    every $i\in J$.}

  By induction hypothesis: \STEP{$\RecContext, \RecContext';
    \BoundContext', \tvar_i; \Context', \NameA : \qlin~\SessionType_i,
    \Var_i : \Type_i \vdash \Process_i\SUBST$ for $i\in J$.}

  We conclude with an application of \rulename{T-Receive}.

\item ($\ProcessP =
  \xsend\NameA{\Tag}{\SessionTypeS}{\NameB}.\ProcessP'$)
  Then $\ProcessP\SUBST =
  \xsend\NameA{\Tag}{\SessionTypeS}{\NameB}.(\ProcessP'\SUBST)$.
%
  From \hypo2 and \rulename{T-Send} we deduce:
  \STEP{$\Context' = (\Context'', \NameA : \Qualifier~\InternalChoice{\tmsg{\Tag_i}{\tvar_i}{\TypeT_i}.\SessionTypeT_i}_{i\in I}) +
    \NameB : \TypeS$;}
  \STEP{$\BoundContext' \wfdash \SessionTypeS$;}
  \STEP{$\Tag = \Tag_k$ for some $k\in I$;}
  \STEP{$\TypeS \subt \Type_k\subst{\SessionTypeS}{\tvar_k}$;}
  \STEP{$\max\{\xweight{\BoundContext}{\SessionTypeS},
    \xweight{\BoundContext}{\TypeS}\} < \infty$;}
  \STEP{$\RecContext, \{ \RecVar \mapsto (\BoundContext; \Context) \},
    \RecContext'; \BoundContext'; \Context'', \NameA :
    \Qualifier~\SessionTypeT_k\subst{\SessionTypeS}{\tvar_k} \vdash
    \ProcessP'$.}

  By induction hypothesis: \STEP{$\RecContext, \RecContext';
    \BoundContext'; \Context'', \NameA :
    \qlin~\SessionTypeT'\subst{\SessionTypeS}{\tvar} \vdash
    \ProcessP'\SUBST$.}

  We conclude with an application of \rulename{T-Send}.

\item ($\ProcessP = \rec\RecVarY.\ProcessP'$)
  Then $\ProcessP\SUBST = \rec\RecVarY.(\ProcessP'\SUBST)$.
  From \hypo{2} and \rulename{T-Rec} we deduce:
  \STEP{$\RecContext, \{ \RecVar \mapsto (\BoundContext; \Context),
    \RecVarY \mapsto (\BoundContext'; \Context') \}, \RecContext';
    \BoundContext'; \Context' \vdash \ProcessP'$;}
  \STEP[t2]{$\dom(\Context'|_\qlin) \subseteq \fn(\ProcessP')$.}

  By induction hypothesis:
  \STEP{$\RecContext, \{ \RecVarY \mapsto (\BoundContext'; \Context')
    \}, \RecContext'; \BoundContext'; \Context' \vdash
    \ProcessP'\SUBST$;}

  From \hypo{t2} and by definition of process substitution:
  \STEP{$\dom(\Context'|_\qlin) \subseteq \fn(\ProcessP') \subseteq
    \fn(\ProcessP'\SUBST)$.}

  We conclude with an application of \rulename{T-Rec}.
  \qedhere
\end{iteMize}
\end{proof}

\noindent The following lemma serves as a slight generalization of subject
reduction (Theorem~\ref{thm:sr}). Note that the last condition
$\Context|_\qun \subseteq \Context'|_\qun$ implies that unrestricted
values can only accumulate (they are never removed from the type
environment) and furthermore their type does not change over time.

\begin{lemma}
\label{lem:sr}
Let \hypo1 $\Context_0;\Context_R,\Context \vdash \Memory$ where
$\qlin(\Context_R)$ and \hypo2 $\Context \vdash \Process$ and
$\system\Memory\Process \red{} \system{\Memory'}{\Process'}$.
Then $\Context_0';\Context_R,\Context' \vdash \Memory'$ and $\Context'
\vdash \Process'$ for some $\Context_0'$ and $\Context'$ such that
$\Context|_\qun \subseteq \Context'|_\qun$.
\end{lemma}
\begin{proof}
  By induction on the derivation of $\system\Memory\Process \red{}
  \system{\Memory'}{\Process'}$ and by cases on the last rule applied.
\begin{iteMize}{$\bullet$}
\item\rulename{R-Open Linear Channel}
  In this case:
\begin{iteMize}{$-$}
\item $\Process = \openChannel(\PointerA : \SessionTypeT, \PointerB :
  \SessionTypeS).\ProcessP'$;

\item $\Memory' = \Memory, \PointerA \mapsto [\PointerB, \EmptyQueue],
  \PointerB \mapsto [\PointerA, \EmptyQueue]$.
\end{iteMize}
From the hypothesis~\hypo2 and rule~\rulename{T-Open Linear Channel}
we obtain:
\begin{iteMize}{$-$}
\item $\emptyset \wfdash \SessionTypeT$;

\item $\SessionTypeS = \co\SessionTypeT$;

\item $\Context, \PointerA : \qlin~\SessionTypeT, \PointerB :
  \qlin~\co\SessionTypeT \vdash \ProcessP'$.
\end{iteMize}
From Proposition~\ref{prop:wf_types}(1) we deduce:
\begin{iteMize}{$-$}
\item $\emptyset \wfdash \SessionTypeS$.
\end{iteMize}
We conclude by taking $\Context_0' = \Context_0$ and $\Context' =
\Context, \PointerA : \qlin~\SessionTypeT, \PointerB :
\qlin~\co\SessionTypeT$. The proof that $\Context_0'; \Context_R,
\Context' \vdash \Memory'$ is trivial and $\Context|_\qun =
\Context'|_\qun$.

\item\rulename{R-Open Unrestricted Channel} Similar to the previous
  case, except that a fresh unrestricted pointer is added to
  $\Context'$.

\item\rulename{R-Choice Left/Right} Trivial.

\item\rulename{R-Send Linear}
  In this case:
\begin{iteMize}{$-$}
\item $\Process = \xsend\PointerA\Tag\SessionTypeS\Value.\ProcessP'$;

\item $\Memory = \Memory'', \PointerA \mapsto [\PointerB, \Queue],
  \PointerB \mapsto [\PointerA, \Queue']$;

\item $\Memory' = \Memory'', \PointerA \mapsto [\PointerB, \Queue],
  \PointerB \mapsto [\PointerA,
  \Queue'::\xmsg\Tag\SessionTypeS\Value]$.
\end{iteMize}
From the hypothesis \hypo2 and rule~\rulename{T-Send} we obtain:
\begin{iteMize}{$-$}
\item \hypo{t1} $\Context = (\Context'', \PointerA :
  \qlin~\InternalChoice{\tmsg{\Tag_i}{\tvar_i}{\TypeT_i}.\SessionTypeT_i}_{i\in
    I}) + \Value : \TypeS$;

\item $\emptyset \wfdash \SessionTypeS$;

\item $\Tag = \Tag_k$ for some $k\in I$;

\item $\TypeS \subt \Type_k\subst{\SessionTypeS}{\tvar_k}$;

\item $\weight{\SessionTypeS} < \infty$ and $\weight{\TypeS} <
  \infty$;

\item $\Context'', \PointerA :
  \SessionTypeT_k\subst{\SessionTypeS}{\tvar_k} \vdash \ProcessP'$.
\end{iteMize}
Let $\Context_0' = \Context_0 + (\Value : \TypeS)|_\qlin$ and
$\Context' = (\Context'', \PointerA :
\qlin~\SessionTypeT_k\subst{\SessionTypeS}{\tvar_k}) + (\Value :
\TypeS)|_\qun$.
Since $\Context|_\qun = \Context'|_\qun$ we only have to show that
$\Context_0';\Context_R,\Context' \vdash \Memory'$.

We prove the items of Definition~\ref{def:wth} in order.
\begin{enumerate}[(1)]
\item We only need to show that $\Queue$ is empty. Suppose by
  contradiction that this is not the case. Then the endpoint type
  associated with $\PointerA$ before the reduction occurs must begin
  with an external choice, which contradicts \hypo{t1}.

\item Let $\Queue' =
  \xmsg{\Tag_1}{\SessionTypeT_1}{\Value_1}::\cdots::\xmsg{\Tag_n}{\SessionTypeT_n}{\Value_n}$.
  From hypothesis~\hypo1 and \hypo{t1} we deduce
  $\Context_0,\Context_R,\Context \vdash \PointerB :
  \qlin~\SessionTypeT_\PointerB$ and $\Context_0,\Context_R,\Context
  \vdash \Value_i : \TypeS_i$ where
\[
\ExternalChoice{\tmsg{\Tag_i}{\tvar_i}{\TypeT_i}.\co{\SessionTypeT_i}}_{i\in
  I}
=
\co{\InternalChoice{\tmsg{\Tag_i}{\tvar_i}{\TypeT_i}.\SessionTypeT_i}_{i\in
    I}} 
=
\tail(\SessionTypeT_\PointerB,
\xmsg{\Tag_1}{\SessionTypeT_1}{\TypeS_1}\cdots\xmsg{\Tag_n}{\SessionTypeT_n}{\TypeS_n})
\]
and by Proposition~\ref{prop:wf_types}(3) we conclude
\[
\co{\SessionTypeT_k\subst{\SessionTypeS}{\tvar_k}}
=
\co{\SessionTypeT_k}\subst{\SessionTypeS}{\tvar_k}
=
\tail(\SessionTypeT_\PointerB,
\xmsg{\Tag_1}{\SessionTypeT_1}{\TypeS_1}\cdots\xmsg{\Tag_n}{\SessionTypeT_n}{\TypeS_n}\xmsg\Tag\SessionTypeS{\TypeS})
\,.
\]

\item Immediate from hypothesis~\hypo1.

\item From hypothesis~\hypo1 we have $\dom(\Memory) =
  \dom(\Context_0,\Context_R,\Context|_\qlin)$ and for every
  $\PointerA' \in \dom(\Memory)$ there exists $\PointerB' \in
  \dom(\Context_R,\Context)$ such that $\PointerA' \reach[\Memory]
  \PointerB'$.
  Clearly $\dom(\Memory') = \dom(\Context_0',\Context_R,\Context'|_\qlin)$
  since $\dom(\Memory') = \dom(\Memory)$ and $\dom(\Context_0') \cup
  \dom(\Context') = \dom(\Context_0) \cup \dom(\Context)$.
  Let $\PointerB \reach[\Memory] \PointerB_0$ and $\Context_R,\Context
  \vdash \PointerB_0 : \SessionTypeT_0$ and assume $\Value \in
  \PointerSet$. We have $\Value \xreach[\Memory'] \PointerB
  \reach[\Memory'] \PointerB_0$, namely $\Value \reach[\Memory']
  \PointerB_0$.
  Now
\[
\weight{\TypeS}
<
\weight{\tail(\SessionTypeT_\PointerB,
  \xmsg{\Tag_1}{\SessionTypeT_1}{\TypeS_1}\cdots\xmsg{\Tag_n}{\SessionTypeT_n}{\TypeS_n})}
\leq
\weight{\SessionTypeT_\PointerB}
\leq
\weight{\SessionTypeT_0}
\]
therefore $\Value \ne \PointerB_0$. We conclude $\PointerB_0 \in
\dom(\Context_R, \Context')$.\Luca{Uso il fatto che puntatori con tipi
  che hanno pesi diversi devono essere diversi, e il fatto che la tail
  diminuisce il peso.}

\item Immediate from hypothesis~\hypo1.
\end{enumerate}

\item\rulename{R-Send Unrestricted}
  In this case:
\begin{iteMize}{$-$}
\item $\Process =
  \xsend{\co\PointerA}\Tag\SessionTypeS\Value.\ProcessP'$;

\item $\Memory = \Memory'', \PointerA \mapsto [\PointerA, \Queue]$;

\item $\Memory' = \Memory'', \PointerA \mapsto [\PointerA,
  \Queue::\xmsg\Tag\SessionTypeS\Value]$.
\end{iteMize}
From the hypothesis \hypo2 and rule~\rulename{T-Send} we obtain:
\begin{iteMize}{$-$}
\item \hypo{t1} $\Context = (\Context'', \shared\Pointer :
  \qun~\SessionTypeT) + \Value : \TypeS$ where $\SessionTypeT =
  \InternalChoice{\tmsg{\Tag_i}{\tvar_i}{\TypeT_i}.\SessionTypeT}_{i\in
    I}$;

\item $\emptyset \wfdash \SessionTypeS$;

\item $\Tag = \Tag_k$ for some $k \in I$;

\item $\TypeS \subt \TypeT_k\subst{\SessionTypeS}{\tvar_k}$;

\item $\weight{\SessionTypeS} < \infty$ and $\weight\TypeS < \infty$;

\item $\Context'', \PointerA : \qun~\SessionTypeT \vdash \ProcessQ$.
\end{iteMize}
Let $\Context_0' = \Context_0 + (\Value : \TypeS)|_\qlin$ and
$\Context' = (\Context'', \shared\Pointer : \qun~\SessionTypeT) +
(\Value : \TypeS)|_\qun$.
Since $\Context|_\qun = \Context'|_\qun$ we only have to show that
$\Context_0';\Context_R,\Context' \vdash \Memory'$.

We prove the items of Definition~\ref{def:wth} in order.
\begin{enumerate}[(1)]
\item Trivial since no queue of linear endpoint was affected by the
  reduction.

\item Ditto.

\item Let $\Queue =
  \xmsg{\Tag_1}{\SessionTypeT_1}{\Value_1}::\cdots::\xmsg{\Tag_n}{\SessionTypeT_n}{\Value_n}$
  and $\Context_0,\Context_R,\Context \vdash \Pointer :
  \qlin~\SessionTypeT_\Pointer$ and $\Context_0,\Context_R,\Context
  \vdash \Value_i : \TypeS_i$.
  We deduce
\[
\ExternalChoice{\tmsg{\Tag_i}{\tvar_i}{\TypeT_i}.\co{\SessionTypeT}}_{i\in I}
=
\co{\InternalChoice{\tmsg{\Tag_i}{\tvar_i}{\TypeT_i}.\SessionTypeT}_{i\in I}}
=
\tail(\SessionTypeT_\Pointer,
\xmsg{\Tag_1}{\SessionTypeT_1}{\TypeS_1}\cdots\xmsg{\Tag_n}{\SessionTypeT_n}{\TypeS_n})
\]
and we conclude
\[
{\co\SessionTypeT}
=
\tail(\SessionTypeT_\Pointer,
\xmsg{\Tag_1}{\SessionTypeT_1}{\TypeS_1}\cdots\xmsg{\Tag_n}{\SessionTypeT_n}{\TypeS_n}\xmsg\Tag\SessionTypeS{\TypeS})
\,.
\]

\item Analogous to the case~\rulename{R-Send Linear} with
  $\SessionType_\PointerA$ in place of $\SessionType_\PointerB$.

\item Immediate from hypothesis~\hypo1.
\end{enumerate}

\item\rulename{R-Receive} In this case:
\begin{iteMize}{$-$}
\item $\Process = \sum_{i\in I}
  \xreceive{\PointerA}{\Tag_i}{\tvar_i}{\Var_i:\Type_i}.\Process_i$;

\item $\Memory = \Memory'', \PointerA \mapsto [\PointerB,
  \xmsg\Tag\SessionTypeS\Value :: \Queue]$ where $\Queue =
  \xmsg{\Tag_1}{\SessionTypeS_1}{\Value_1} :: \cdots ::
  \xmsg{\Tag_n}{\SessionTypeS_n}{\Value_n}$;

\item $\Tag = \Tag_k$ for some $k\in I$;

\item $\Process' =
  \Process_k\subst{\SessionTypeS}{\tvar_k}\subst{\Value}{\Var_k}$;

\item $\Memory' = \Memory'', \PointerA \mapsto [\PointerB, \Queue]$.
\end{iteMize}
From the hypothesis~\hypo2 and rule~\rulename{T-Receive} we obtain:
\begin{iteMize}{$-$}
\item $\Context = \Context'', \PointerA :
  \qlin~\ExternalChoice{\tmsg{\Tag_i}{\tvar_i}{\TypeS_i}.\SessionTypeT_i}_{i\in
    J}$ with $J \subseteq I$;

\item $\TypeS_k \subt \TypeT_k$;

\item\hypo{t3} $\tvar_k; \Context'', \PointerA :
  \qlin~\SessionTypeT_k, \Var_k : \TypeT_k \vdash \Process_k$
\end{iteMize}
Let $\Context_0,\Context_R,\Context \vdash \Value : \TypeS$. From
hypothesis~\hypo1 and Proposition~\ref{prop:subt_subst} we obtain:
\begin{iteMize}{$-$}
\item\hypo{c1} $\emptyset \wfdash \SessionTypeS$ and
  $\weight\SessionTypeS < \infty$;

\item\hypo{c2} $\TypeS \subt \TypeS_k\subst{\SessionTypeS}{\tvar_k}
  \subt \TypeT_k\subst{\SessionTypeS}{\tvar_k}$.
\end{iteMize}
From hypothesis~\hypo1 we also deduce that:
\begin{iteMize}{$-$}
\item \hypo{f1} if $\qun(\TypeS)$, then $\Value \in \dom(\Context)$
  and $\Context \vdash \Value : \TypeS$, because all the unrestricted
  values are in $\Context$;

\item \hypo{f2} if $\qlin(\TypeS)$, then $\Value \not\in
  \dom(\Context)$, because $\Value \xreach[\Memory] \PointerA$ and
  therefore it must be $\Value \in \dom(\Context_0)$ (process
  isolation prevents $\PointerA$ from being reachable from any pointer
  in $\dom(\Context_R,\Context)$ and different from $\PointerA$).
\end{iteMize}
From \hypo{t3}, \hypo{c1}, and Lemma~\ref{lem:type_subst} we have:
\begin{iteMize}{$-$}
\item\hypo{t3'} $\Context''\subst\SessionTypeS{\tvar_k}, \PointerA :
  \Qualifier~\SessionTypeT_k\subst{\SessionTypeS}{\tvar_k}, \Var_k :
  \TypeT_k\subst{\SessionTypeS}{\tvar_k} \vdash
  \Process_k\subst\SessionTypeS{\tvar_k}$.
\end{iteMize}
From \hypo{f1} and \hypo{f2} we deduce that $\Context_0 = \Context_0',
(\Value : \TypeS)|_\qlin$ for some $\Context_0'$. Take $\Context' =
(\Context'', \PointerA :
\Qualifier~\SessionTypeT_k\subst{\SessionTypeS}{\tvar_k}) + \Value :
\TypeS$ and observe that $\Context'$ is well defined by~\hypo{f1}
and~\hypo{f2} and also $\Context|_\qun \subseteq \Context'|_\qun$ by
construction of $\Context'$.  From \hypo{t3'}, \hypo{c2}, and
Lemma~\ref{lem:value_subst} we conclude:
\begin{iteMize}{$-$}
\item $\Context' \vdash
  \Process_k\subst{\SessionTypeS}{\tvar_k}\subst{\Value}{\Var_k}$
\end{iteMize}
We have to show
$\Context_0', \Context_R, \Context' \vdash \Memory'$ and we prove the
items of Definition~\ref{def:wth} in order.
\begin{enumerate}[(1)]
\item If $\PointerA = \PointerB$ there is nothing to prove. Suppose
  $\PointerA \ne \PointerB$. Since the queue associated with
  $\PointerA$ is not empty in $\Memory$, the queue associated with its
  peer endpoint $\PointerB$ must be empty. The reduction does not
  change the queue associated with $\PointerB$, therefore
  condition~(1) of Definition~\ref{def:wth} is satisfied.

\item Suppose $\PointerA \ne \PointerB$ for otherwise there is nothing
  to prove. From hypothesis~\hypo1 we deduce
  $\Context_0,\Context_R,\Context \vdash \PointerB :
  \qlin~\SessionTypeT_\PointerB$ and
\[
\begin{array}{rcl}
\co{\SessionTypeT_\PointerB} & = &
  \tail(\ExternalChoice{\tmsg{\Tag_i}{\tvar_i}{\TypeS_i}.\SessionTypeT_i}_{i\in J},
  \xmsg{\Tag}{\SessionTypeS}{\TypeS}\xmsg{\Tag_1}{\SessionTypeS_1}{\TypeS'_1}\cdots\xmsg{\Tag_n}{\SessionTypeS_n}{\TypeS'_n})
  \\
  & = & 
  \tail(\SessionTypeT_k\subst{\SessionTypeS}{\tvar_k},
  \xmsg{\Tag_1}{\SessionTypeS_1}{\TypeS'_1}\cdots\xmsg{\Tag_n}{\SessionTypeS_n}{\TypeS'_n})
\end{array}
\]
where $\Context_0,\Context_R,\Context \vdash \Value_i : \TypeS'_i$ for
$1\le i\le n$.

\item Similar to the previous item, where $\PointerA =
  \PointerB$.\Luca{\`E vero?}

\item Straightforward by definition of $\Context_0'$ and $\Context'$.

\item Immediate from hypothesis~\hypo1.
\end{enumerate}

\item\rulename{R-Par} In this case:
\begin{iteMize}{$-$}
\item $\Process = \Process_1 \parop \Process_2$;

\item $\system\Memory{\Process_1} \red{} \system{\Memory'}{\Process_1'}$;

\item $\Process' = \Process_1' \parop \Process_2$.
\end{iteMize}
From the hypothesis~\hypo2 and rule~\rulename{T-Par} we obtain:
\begin{iteMize}{$-$}
\item $\Context = \Context_1 + \Context_2$;

\item $\Context_i \vdash \Process_i$ for $i\in\{1,2\}$.
\end{iteMize}
In particular, from Lemma~\ref{lem:weakening} we have:
\begin{iteMize}{$-$}
\item $(\Context_0; \Context_R, \Context_2|_\qlin, \Context_1) +
  \Context_2|_\qun \vdash \Memory$;

\item $\Context_1 + \Context_2|_\qun \vdash \Process_1$.
\end{iteMize}
By induction hypothesis we deduce that there exist $\Context_0'$ and
$\Context_1'$ such that:
\begin{iteMize}{$-$}
\item $(\Context_1 + \Context_2|_\qun)|_\qun = (\Context_1 +
  \Context_2)|_\qun \subseteq \Context_1'|_\qun$;

\item $\Context_0'; \Context_R, \Context_2|_\qlin, \Context_1' \vdash
  \Memory'$;

\item $\Context_1' \vdash \Process_1'$.
\end{iteMize}
Now $\Context_2|_\qlin,\Context_1' = \Context_2|_\qlin, (\Context_1' +
\Context_2|_\qun) = \Context_1' + \Context_2$. Therefore, from
rule~\rulename{T-Par} we obtain $\Context_1' + \Context_2 \vdash
\Process'$.
We conclude by taking $\Context' = \Context_1' + \Context_2$.

\item\rulename{R-Rec} In this case:
  \begin{iteMize}{$-$}
  \item $\ProcessP = \rec\RecVar.\ProcessQ$;
  \item $\ProcessP' = \ProcessQ\subst\ProcessP\RecVar$;
  \item $\Memory' = \Memory$.
  \end{iteMize}
  From the hypothesis \hypo2 and rule~\rulename{T-Rec} we obtain:
  \begin{iteMize}{$-$}
  \item \hypo{t3} $\{ \RecVar \mapsto (\EmptyBoundContext; \Context)
    \}; \EmptyBoundContext; \Context \vdash \ProcessQ$;
  \item $\dom(\Context|_\qlin) \subseteq \fn(\ProcessQ)$.
  \end{iteMize}
  From \hypo{t3} and Lemma~\ref{lem:process_subst} we obtain:
  \begin{iteMize}{$-$}
  \item $\Context \vdash \ProcessP'$.
  \end{iteMize}
%
%
  We conclude by taking $\Context_0' = \Context_0$ and $\Context' =
  \Context$.

\item\rulename{R-Struct} Follows from Lemma~\ref{lem:cong} and
  induction.
  \qedhere
\end{iteMize}
\end{proof}

\noindent We conclude with the proofs of subject reduction and soundness.

\begin{theorem}[Theorem~\ref{thm:sr}]
  Let $\Context_0; \Context \vdash (\Memory; \Process)$ and $(\Memory;
  \Process) \red{} (\Memory'; \Process')$. Then $\Context_0';
  \Context' \vdash (\Memory'; \Process')$ for some $\Context_0'$ and
  $\Context'$.
\end{theorem}
\begin{proof}
  Follows from Lemma~\ref{lem:sr} by taking $\Context_R = \emptyset$.
\end{proof}

\begin{proposition}
\label{prop:dom_context}
Let $\Context \vdash \Process$. Then $\fn(\Process) \subseteq
\dom(\Context)$ and $\dom(\Context|_\qlin) \subseteq \fn(\Process)$.
\end{proposition}
\begin{proof}
  From the hypothesis $\Context \vdash \Process$ we deduce that
  $\Process$ is closed with respect to process variables. The results
  follows by a straightforward induction on the derivation of
  $\Context \vdash \Process$ and by cases on the last rule applied,
  where the case for \rulename{T-Var} is impossible by hypothesis and
  \rulename{T-Rec} is a base case when proving the second inclusion.
\end{proof}

\begin{theorem}[Theorem~\ref{thm:safety}]
Let $\vdash \Process$. Then $\Process$ is well behaved.
\end{theorem}
\begin{proof}
  Consider a derivation $\system\EmptyMemory\ProcessP \wred{}
  \system\Memory{\ProcessQ}$.
  From Theorem~\ref{thm:sr} we deduce \hypo* $\Context_0;\Context
  \vdash \system\Memory{\ProcessQ}$ for some $\Context_0$ and
  $\Context$.
  We prove conditions~\hypo{1--3} of Definition~\ref{def:wb} in order:
\begin{enumerate}[(1)]
\item Using Proposition~\ref{prop:dom_context},
  Definition~\ref{def:reachable}, and Definition~\ref{def:wth} we have
  $\reachable{\fn(\ProcessQ)}\Memory \subseteq
  \reachable{\dom(\Context)}\Memory = \dom(\Memory)$ and
  $\dom(\Memory) = \reachable{\dom(\Context)}\Memory =
  \reachable{\dom(\Context|_\qlin)}\Memory \subseteq
  \reachable{\fn(\ProcessQ)}\Memory$.

\item Suppose $\ProcessQ \equiv \Process_1 \parop \Process_2$.
  By Lemma~\ref{lem:cong} we deduce $\Context \vdash \Process_1 \parop
  \Process_2$, namely there exist $\Context_1$ and $\Context_2$ such
  that $\Context = \Context_1 + \Context_2$ and $\Context_i \vdash
  \Process_i$.
  From the definition of $\Context_1 + \Context_2$ we deduce
  $\dom(\Context_1|_\qlin) \cap \dom(\Context_2|_\qlin) = \emptyset$.
  From Proposition~\ref{prop:dom_context} we have $\fn(\Process_i)
  \subseteq \dom(\Context_i)$ for $i\in\{1,2\}$.
  From \hypo* we conclude $\reachable{\fn(\Process_1)}\Memory \cap
  \reachable{\fn(\Process_2)}\Memory \subseteq
  \reachable{\dom(\Context_1)}\Memory \cap
  \reachable{\dom(\Context_2)}\Memory =
  \reachable{\dom(\Context_1|_\qlin)}\Memory \cap
  \reachable{\dom(\Context_2|_\qlin)}\Memory = \emptyset$.

\item
  Suppose $\ProcessQ \equiv \ProcessP' \parop \ProcessQ'$ where
  $\ProcessP'$ has no unguarded parallel composition and
  $\system\Memory{\ProcessQ} \nred{}$.
  Then $\Process'$ contains no unfolded recursion, choice,
  $\openChannel$, output prefix that is not guarded by an input
  prefix, for all these processes reduce. In the case of output
  prefixes, one uses 
  $\Context_0; \Context \vdash
  \Memory$ to deduce that either \rulename{R-Send Linear} or
  \rulename{R-Send Unrestricted} can be applied.
  Suppose $\Process' \ne \idle$. \REVISION{Then either $\Process' =
  \closeChannel(\Pointer)$ or $\Process' = \sum_{i\in I}
  \xreceive\Pointer{\Tag_i}{\tvar_i}{\Var_i:\Type_i}.\Process_i$.
  Suppose by contradiction that the queue associated with $\Pointer$
  is not empty, namely that $\Pointer \mapsto [\PointerB,
  \xmsg\Tag\SessionTypeS\Value :: \Queue] \in \Memory$.
  From the hypothesis $\Context \vdash \Memory$ we deduce that the
  endpoint type associated with $\Pointer$ cannot be $\SessionEnd$,
  and therefore $\Process' \ne \closeChannel(\Pointer)$.}
  From the hypothesis $\Context \vdash \ProcessQ$ and
  rule~\rulename{T-Receive} we deduce $\Context \vdash \Pointer :
  \qlin~\ExternalChoice{\xmsg{\Tag_i}{\tvar_i}{\TypeS_i}.\SessionType_i}_{i\in
    J}$ and $J \subseteq I$.
  From the hypothesis $\Context \vdash \Memory$ we deduce $\Tag =
  \Tag_k$ for some $k\in J$, namely $\system\Memory{\ProcessP'}
  \red{}$, which is absurd.
  \qedhere
\end{enumerate}
\end{proof}


\section{Supplement to Section~\ref{sec:algorithms}}
\label{sec:extra_algorithms}

\subsection{Subtyping}

In order to prove correctness and completeness of the subtyping
algorithm (Definition~\ref{def:asubt}) with respect to the subtyping
relation (Definition~\ref{def:subt}) we need a few more concepts. The
first one is that of \emph{trees} of an endpoint type $\SessionType$,
which is the set of all subtrees of $\SessionType$ where recursive
terms have been infinitely unfolded. We build this set inductively, as
follows:

\begin{definition}[endpoint type trees]
We write $\trees(\SessionType)$ for the least set such that:
\begin{iteMize}{$\bullet$}
\item $\SessionType \in \trees(\SessionType)$;

\item $\trec\tvar.\SessionTypeS \in \trees(\SessionType)$ implies
  $\SessionTypeS\subst{\trec\tvar.\SessionTypeS}{\tvar} \in
  \trees(\SessionType)$;

\item
  $\{\dagger\xmsg{\Tag_i}{\tvar_i}{\Qualifier_i~\SessionTypeS_i}.\SessionTypeT_i\}_{i\in
    I} \in \trees(\SessionType)$ where $\dagger \in \{ {?}, {!} \}$
  implies $\SessionTypeS_i \in \trees(\SessionType)$ and
  $\SessionTypeT_i \in \trees(\SessionType)$ for every $i\in I$.
\end{iteMize}
\end{definition}

\noindent Observe that $\trees(\SessionType)$ is \emph{finite} for every
$\SessionType$, because the infinite unfolding of an endpoint type is
a regular tree~\cite{Courcelle83}.
Also, every free type variable in $\SessionTypeS \in
\trees(\SessionTypeT)$ is either free in $\SessionTypeT$ or it is
bound by a prefix of $\SessionTypeT$. In particular, it cannot be
bound by a recursion.

The next concept we need is that of \emph{instance} of an endpoint
type subtree. The idea is to generate the set of all instances of the
(trees of the) endpoint types that the subtyping algorithm visits, and
to make sure that this set is finite. Looking at the rules in
Table~\ref{tab:algorithmic_subtyping} we see that only type variables
that are bound in a prefix $\tmsg\Tag\tvar\Type$ are ever
instantiated. Also, each variable $\tvarA$ in one of the endpoint
types can be instantiated with $\varmap(\tvarA,\tvarB)$ where $\tvarB$
is some type variable of the other endpoint type. These considerations
lead to the following definition of endpoint type instances:

\begin{definition}[endpoint type instances]
Let $\varmap$ be a map as by Definition~\ref{def:asubt}.  We define
$\instances(\varmap, \SessionTypeT, \SessionTypeS)$ as the smallest
set such that:
\begin{iteMize}{$\bullet$}
\item if $\SessionTypeT' \in \trees(\SessionTypeT)$ and $\{ \tvarA_1,
  \dots, \tvarA_n \} = \ftv(\SessionTypeT') \cap \btv(\SessionTypeT)$
  and $\{ \tvarB_1, \dots, \tvarB_n \} \subseteq \btv(\SessionTypeS)$,
  then $\SessionTypeT'\subst{\varmap(\tvarA_1,
    \tvarB_1)}{\tvarA_1}\cdots\subst{\varmap(\tvarA_n,
    \tvarB_n)}{\tvarA_n} \in \instances(\varmap, \SessionTypeT,
  \SessionTypeS)$;

\item if $\SessionTypeS' \in \trees(\SessionTypeS)$ and $\{ \tvarB_1,
  \dots, \tvarB_n \} = \ftv(\SessionTypeS') \cap \btv(\SessionTypeS)$
  and $\{ \tvarA_1, \dots, \tvarA_n \} \subseteq \btv(\SessionTypeT)$,
  then $\SessionTypeS'\subst{\varmap(\tvarB_1,
    \tvarA_1)}{\tvarB_1}\cdots\subst{\varmap(\tvarB_n,
    \tvarA_n)}{\tvarB_n} \in \instances(\varmap, \SessionTypeT,
  \SessionTypeS)$.
\end{iteMize}
\end{definition}

\noindent Observe that $\SessionTypeT, \SessionTypeS \in \instances(\varmap,
\SessionTypeT, \SessionTypeS)$ and that $\instances(\varmap,
\SessionTypeT, \SessionTypeS)$ is finite, since it contains finitely
many instantiations of finitely many subtrees of $\SessionTypeT$ and
$\SessionTypeS$.

\begin{proposition}
\label{prop:instances}
Every endpoint type occurring in the derivation of $\emptyset
\vdash_\varmap \SessionTypeT \asubt \SessionTypeS$ is in
$\instances(\SessionTypeT, \SessionTypeS)$.
\end{proposition}
\begin{proof}[Proof sketch]
  A simple induction on the derivation of $\emptyset \vdash_\varmap
  \SessionTypeT \asubt \SessionTypeS$.
\end{proof}

\begin{lemma}
\label{lem:cut}
Let $\emptyset \vdash_\varmap \SessionTypeT \asubt \SessionTypeS$ and
$\{ (\SessionTypeT, \SessionTypeS) \} \vdash_\varmap \SessionTypeT'
\asubt \SessionTypeS'$. Then $\emptyset \vdash_\varmap \SessionTypeT'
\asubt \SessionTypeS'$.
\end{lemma}
\begin{proof}[Proof sketch]
  A simple induction on the proof of $\{ (\SessionTypeT,
  \SessionTypeS) \} \vdash_\varmap \SessionTypeT' \asubt
  \SessionTypeS'$ where every application of rule~\rulename{S-Axiom}
  for the pair $(\SessionTypeT, \SessionTypeS)$ is replaced by a copy
  of the proof of $\emptyset \vdash_\varmap \SessionTypeT \asubt
  \SessionTypeS$.
\end{proof}

\begin{theorem}[Theorem~\ref{thm:subtyping_algorithm}]
  Let $\SessionTypeT_0$ and $\SessionTypeS_0$ be independent endpoint
  types and $\varmap$ be a map as by Definition~\ref{def:asubt}. Then
  $\vdash_\varmap \SessionTypeT_0 \asubt \SessionTypeS_0$ if and only
  if $\SessionTypeT_0 \subt \SessionTypeS_0$.
\end{theorem}
\begin{proof}
($\Rightarrow$)
It is enough to show that
\[
\begin{array}{r@{~}c@{~}l}
  \srel
  & \eqdef &
  \{ (\SessionTypeT, \SessionTypeS) \mid \emptyset \vdash_\varmap \SessionTypeT \asubt \SessionTypeS \land \SessionTypeT, \SessionTypeS \in \instances(\varmap, \SessionTypeT_0, \SessionTypeS_0) \}
  \\
  & & \quad{} \cup
  \{ (\Qualifier~\SessionTypeT, \Qualifier'~\SessionTypeS) \mid
  \Qualifier \leq \Qualifier'
  \land
  \SessionTypeT, \SessionTypeS \in \instances(\varmap, \SessionTypeT_0, \SessionTypeS_0)
  \land
  \emptyset \vdash_\varmap \SessionTypeT \asubt \SessionTypeS \}
\end{array}
\]
is a coinductive subtyping.
Let $(\Qualifier~\SessionTypeT, \Qualifier'~\SessionTypeS) \in \srel$.
Then $\Qualifier \leq \Qualifier'$ and $\emptyset \vdash_\varmap
\SessionTypeT \asubt \SessionTypeS$. By definition of $\srel$ we
conclude $(\SessionTypeT, \SessionTypeS) \in \srel$.

Let $(\SessionTypeT, \SessionTypeS) \in \srel$. Then \hypo{J}
$\emptyset \vdash_\varmap \SessionTypeT \asubt \SessionTypeS$. We
reason by induction on the number of topmost applications of
rules~\rulename{S-Rec Left} and~\rulename{S-Rec Right} (which must be
finite because of contractivity of endpoint types) and by cases on the
first (bottom-up) rule different from~\rulename{S-Rec Left}
and~\rulename{S-Rec Right} applied for deriving \hypo{J}, observing
that is cannot be \rulename{S-Axiom} for the context is initially
empty and rules~\rulename{S-Rec Left} and~\rulename{S-Rec Right} only
add pairs of endpoint types where at least one of them begins with a
recursion:
\begin{iteMize}{$\bullet$}
\item \rulename{S-Rec Left}
  Then $\SessionTypeT \seq \trec\tvar.\SessionTypeT'$ and $\{(
  \SessionTypeT, \SessionTypeS )\} \vdash_\varmap
  \SessionTypeT'\subst{\SessionTypeT}{\tvar} \asubt \SessionTypeS$.
  From \hypo{J} and Lemma~\ref{lem:cut} we derive $\emptyset
  \vdash_\varmap \SessionTypeT'\subst{\SessionTypeT}{\tvar} \asubt
  \SessionTypeS$.
  By induction hypothesis we derive
  $\SessionTypeT'\subst{\SessionTypeT}{\tvar} \subt \SessionTypeS$ and
  we conclude by observing that $\SessionTypeT =
  \SessionTypeT'\subst{\SessionTypeT}{\tvar}$.

\item \rulename{S-Rec Right} Symmetric of the previous case.

\item \rulename{S-Var} Then $\SessionTypeT \seq \SessionTypeS \seq
  \tvar$ and there is nothing left to prove.

\item \rulename{S-End} Then $\SessionTypeT \seq \SessionTypeS \seq
  \SessionEnd$ and there is nothing left to prove.

\item \rulename{S-Input}
  Then $\SessionTypeT \seq
  \ExternalChoice{\xmsg{\Tag_i}{\tvarA_i}{\TypeT_i}.\SessionTypeT_i}_{i\in
    I}$ and $\SessionTypeS \seq
  \ExternalChoice{\xmsg{\Tag_j}{\tvarB_j}{\TypeS_j}.\SessionTypeS_j}_{j\in
    I \cup J}$.
  Let $\tvarC_i = \varmap(\tvarA_i, \tvarB_i)$ for $i\in I$.
  From \rulename{S-Input} we deduce:
\begin{iteMize}{$-$}
\item $\{ (\SessionTypeT, \SessionTypeS) \} \vdash_\varmap
  \TypeT_i\subst{\tvarC_i}{\tvarA_i} \asubt
  \TypeS_i\subst{\tvarC_i}{\tvarB_i}$ for every
  $i\in I$;

\item $\{ (\SessionTypeT, \SessionTypeS) \} \vdash_\varmap
  \SessionTypeT_i\subst{\tvarC_i}{\tvarA_i} \asubt
  \SessionTypeS_i\subst{\tvarC_i}{\tvarB_i}$ for
  every $i\in I$.
\end{iteMize}
From Lemma~\ref{lem:cut} we derive:
\begin{iteMize}{$-$}
\item $\emptyset \vdash_\varmap \TypeT_i\subst{\tvarC_i}{\tvarA_i}
  \asubt \TypeS_i\subst{\tvarC_i}{\tvarB_i}$ for every $i\in I$;

\item $\emptyset \vdash_\varmap
  \SessionTypeT_i\subst{\tvarC_i}{\tvarA_i} \asubt
  \SessionTypeS_i\subst{\tvarC_i}{\tvarB_i}$ for
  every $i\in I$.
\end{iteMize}
By definition of $\srel$ we know that $\SessionTypeT, \SessionTypeS
\in \instances(\varmap, \SessionTypeT_0, \SessionTypeS_0)$. For every
$i\in I$, we can deduce that $\tvarC_i \not\in \ftv(\TypeT_i) \cup
\ftv(\SessionTypeT_i) \cup \ftv(\TypeS_i) \cup \ftv(\SessionTypeS_i)$,
because $\tvarC_i$ can only substitute the free occurrences of
$\tvarA_i$ and of $\tvarB_i$ and:
\begin{iteMize}{$-$}
\item $\tvarA_i$ is bound in the $i$-th branch of $\SessionTypeT$ and
  does not occur in $\SessionTypeS$;

\item $\tvarB_i$ is bound in the $i$-th branch of $\SessionTypeS$ and
  does not occur in $\SessionTypeT$.
\end{iteMize}
Therefore, by alpha conversion we obtain:
\begin{iteMize}{$-$}
\item $\SessionTypeT =
  \ExternalChoice{\xmsg{\Tag_i}{\tvarC_i}{\TypeT_i\subst{\tvarC_i}{\tvarA_i}}.\SessionTypeT_i\subst{\tvarC_i}{\tvarA_i}}_{i\in
    I}$;

\item $\SessionTypeS =
  \ExternalChoice{\xmsg{\Tag_j}{\tvarC_i}{\TypeS_i\subst{\tvarC_i}{\tvarB_i}}.\SessionTypeS_i\subst{\tvarC_i}{\tvarB_i}}_{i\in
    I} +
  \ExternalChoice{\xmsg{\Tag_j}{\tvarB_j}{\TypeS_j}.\SessionTypeS_j}_{j\in
    J\setminus I}$.
\end{iteMize}
We conclude $(\TypeS_i\subst{\tvarC_i}{\tvarB_i},
\TypeT_i\subst{\tvarC_i}{\tvarB_i}) \in \srel$ and
$(\SessionTypeS_i\subst{\tvarC_i}{\tvarB_i},
\SessionTypeT_i\subst{\tvarC_i}{\tvarB_i}) \in \srel$ by definition of
$\srel$.

\item \rulename{S-Output} Analogous to the previous case.
\end{iteMize}

($\Leftarrow$)
We prove that $\SessionTypeT, \SessionTypeS \in \instances(\varmap,
\SessionTypeT_0, \SessionTypeS_0)$ and $\SessionTypeT \subt
\SessionTypeS$ imply $\srel \vdash_\varmap \SessionTypeT \asubt
\SessionTypeS$ by induction on $\instances(\varmap, \SessionTypeT,
\SessionTypeS) \setminus \srel$.
In the base case we have $(\SessionTypeT, \SessionTypeS) \in \srel$
and we conclude with an application of \rulename{T-Axiom}.  For the
inductive case we reason by case analysis on the structure of
$\SessionTypeT$ and $\SessionTypeS$, knowing that $\SessionTypeT \subt
\SessionTypeS$:
\begin{iteMize}{$\bullet$}
\item ($\SessionTypeT \seq \SessionTypeS \seq \tvar$) We conclude with
  an application of \rulename{S-Var}.

\item ($\SessionTypeT \seq \SessionTypeS \seq \SessionEnd$) We
  conclude with an application of \rulename{S-End}.

\item ($\SessionTypeT \seq \trec\tvar.\SessionTypeT'$)
  Since $\SessionTypeT = \SessionTypeT'\subst{\SessionTypeT}{\tvar}$
  we have $\SessionTypeT'\subst{\SessionTypeT}{\tvar} \subt
  \SessionTypeS$.
  By induction hypothesis we know that $\srel \cup \{(\SessionTypeT,
  \SessionTypeS)\} \vdash_\varmap
  \SessionTypeT'\subst{\SessionTypeT}{\tvar} \asubt \SessionTypeS$ is
  derivable.
  We conclude with an application of \rulename{S-Rec Left}.

\item ($\SessionTypeS \seq \trec\tvar.\SessionTypeS'$) Symmetric of
  the previous case.

\item ($\SessionTypeT \seq
  \ExternalChoice{\xmsg{\Tag_i}{\tvarA_i}{\Qualifier_i~\SessionTypeT'_i}.\SessionTypeT_i}_{i\in
    I}$ and $\SessionTypeS \seq
  \ExternalChoice{\xmsg{\Tag_j}{\tvarB_j}{\Qualifier'_i~\SessionTypeS'_j}.\SessionTypeS_j}_{j\in
    J}$ and $I \subseteq J$)
  From the hypothesis $\SessionTypeT \subt \SessionTypeS$ we know that
  for every $i\in I$ there exists $\tvarC_i$ such that
  $\Qualifier_i \leq \Qualifier'_i$ and
  $\SessionTypeT'_i\subst{\tvarC_i}{\tvarA_i} \subt
  \SessionTypeS'_i\subst{\tvarC_i}{\tvarB_i}$ and
  $\SessionTypeT_i\subst{\tvarC_i}{\tvarA_i} \subt
  \SessionTypeS_i\subst{\tvarC_i}{\tvarB_i}$.
  Since $\SessionTypeT, \SessionTypeS \in \instances(\varmap,
  \SessionTypeT_0, \SessionTypeS_0)$ we know that $\tvarD_i =
  \varmap(\tvarA_i, \tvarB_i) \not\in \ftv(\SessionTypeT'_i) \cup
  \ftv(\SessionTypeT_i) \cup \ftv(\SessionTypeS'_i) \cup
  \ftv(\SessionTypeS_i)$.\Luca{Serve?}
  We deduce $\SessionTypeT'_i\subst{\tvarD_i}{\tvarA_i} \subt
  \SessionTypeS'_i\subst{\tvarD_i}{\tvarB_i}$ and
  $\SessionTypeT_i\subst{\tvarD_i}{\tvarA_i} \subt
  \SessionTypeS_i\subst{\tvarD_i}{\tvarB_i}$ for every $i\in
  I$.\Luca{Questo \`e un passaggio banale?}
  By induction hypothesis we derive that $\srel \cup \{(\SessionTypeT,
  \SessionTypeS)\} \vdash_\varmap \SessionTypeT'_i \asubt
  \SessionTypeS'_i$ and $\srel \cup \{(\SessionTypeT, \SessionTypeS)\}
  \vdash_\varmap \SessionTypeT_i \asubt \SessionTypeS_i$ are derivable
  for every $i \in I$.
  Also, $\srel \cup \{(\SessionTypeT, \SessionTypeS)\} \vdash_\varmap
  \Qualifier_i~\SessionTypeT'_i \asubt \Qualifier'_i~\SessionTypeS'_i$
  is derivable with an application of \rulename{S-Type} for every
  $i\in I$.
  We conclude $\srel \vdash_\varmap \SessionTypeT \asubt
  \SessionTypeS$ with an application of \rulename{S-Input}.

\item ($\SessionTypeT \seq
  \InternalChoice{\xmsg{\Tag_i}{\tvarA_i}{\TypeT_i}.\SessionTypeT_i}_{i\in
    I}$ and $\SessionTypeS \seq
  \InternalChoice{\xmsg{\Tag_j}{\tvarB_j}{\TypeS_j}.\SessionTypeS_j}_{j\in
    J}$ and $J \subseteq I$) Analogous to the previous case.
  \qedhere
\end{iteMize}
\end{proof}

\subsection{Type Weight}

We begin by proving that the weight algorithm is unaffected by
foldings/unfoldings of recursive terms.

\begin{proposition}[Proposition~\ref{prop:unfold_weight}]
  $\aweight(\BoundContext_0, \emptyset, \trec\tvarA.\SessionType) =
  \aweight(\BoundContext_0, \emptyset,
  \SessionType\subst{\trec\tvarA.\SessionType}{\tvarA})$.
\end{proposition}
\begin{proof}
\newcommand{\AWS}{\aweight(\BoundContext_0, \BoundContext, \SessionTypeS)}
\newcommand{\AWSS}{\aweight(\BoundContext_0, \BoundContext \setminus \{ \tvarA \},
  \SessionTypeS\subst{\trec\tvarA.\SessionTypeT}{\tvarA})}

  Let $\aweight(\BoundContext_0, \emptyset, \trec\tvarA.\SessionType)
  = w$.
  We prove a more general statement, namely that for every
  $\SessionTypeS$ and $\BoundContext$ such that
  $\aweight(\BoundContext_0, \BoundContext, \SessionTypeS) \leq w$ and
  $\btv(\SessionTypeS) \cap \ftv(\SessionTypeT) = \emptyset$ we have:
\begin{enumerate}[(1)]
\item $\tvarA \in \BoundContext$ implies $\AWS \leq \AWSS \leq
  \max\{w, \AWS\}$;

\item $\tvarA \not\in \BoundContext$ implies $\AWS = \AWSS$.
\end{enumerate}

The statement then follows from (1) by taking $\SessionTypeS =
\SessionTypeT$ and $\BoundContext = \{ \tvarA \}$ and noting that
$\aweight(\BoundContext_0, \emptyset, \trec\tvarA.\SessionTypeT) =
\aweight(\BoundContext_0, \{ \tvarA \}, \SessionTypeT)$ by definition
of algorithmic weight. We proceed by induction on $\SessionTypeS$
assuming, without loss of generality, that $(\{ \tvarA \} \cup
\ftv(\SessionTypeT)) \cap \btv(\SessionTypeS) = \emptyset$:
\begin{iteMize}{$\bullet$}
\item ($\SessionTypeS \seq \SessionEnd$ or $\SessionTypeS \seq
  \InternalChoice{\xmsg{\Tag_i}{\tvar_i}{\Type_i}{\SessionType_i}}_{i\in
    I}$) Clear as $\AWS = \AWSS = 0$.

\item ($\SessionTypeS \seq \tvarA$)
  We have $\AWSS = \aweight(\BoundContext_0, \BoundContext \setminus
  \{ \tvarA \}, \trec\tvarA.\SessionTypeT) = w$\Luca{qualche
    assunzione su $\BoundContext$ e le variabili libere di
    $\SessionTypeT$} therefore we conclude:
\begin{enumerate}[(1)]
\item $\AWS = 0 \leq w = \AWSS = \max\{w, \AWS\}$;

\item $\AWS = \infty = w = \AWSS$.
\end{enumerate}

\item ($\SessionTypeS \seq \tvarB \ne \tvarA$) Trivial since $\AWS =
  \AWSS$.

\item ($\SessionTypeS \seq \trec\tvarB.\SessionTypeS'$)
  By induction hypothesis we deduce:
\begin{enumerate}[(1)]
\item $\tvarA \in \BoundContext$ implies \[\quad\enspace\aweight(\BoundContext_0,
  \BoundContext \cup \{ \tvarB \}, \SessionTypeS') \leq
  \aweight(\BoundContext_0, (\BoundContext \cup \{ \tvarB \})\!
  \setminus\! \{ \tvarA \},
  \SessionTypeS'\subst{\trec\tvarA.\SessionTypeT}{\tvarA}) \leq
  \max\{w, \aweight(\BoundContext_0, \BoundContext \cup \{ \tvarB \},
  \SessionTypeS')\};\]

\item $\tvarA \not\in \BoundContext$ implies
  $\aweight(\BoundContext_0, \BoundContext \cup \{ \tvarB \},
  \SessionTypeS') = \aweight(\BoundContext_0, (\BoundContext \cup \{
  \tvarB \}) \setminus \{ \tvarA \},
  \SessionTypeS'\subst{\trec\tvarA.\SessionTypeT}{\tvarA})$.
\end{enumerate}
We conclude by definition of algorithmic weight, since:
\begin{iteMize}{$-$}
\item $\AWS = \aweight(\BoundContext_0, \BoundContext,
  \trec\tvarB.\SessionTypeS') = \aweight(\BoundContext_0,
  \BoundContext \cup \{ \tvarB \}, \SessionTypeS')$, and

\item $\AWSS = \aweight(\BoundContext_0, \BoundContext \setminus \{
  \tvarA \},
  (\trec\tvarB.\SessionTypeS')\subst{\trec\tvarA.\SessionTypeT}{\tvarA})
  = \aweight(\BoundContext_0, \BoundContext \setminus \{ \tvarA \},
  \trec\tvarB.(\SessionTypeS'\subst{\trec\tvarA.\SessionTypeT}{\tvarA}))
  = \aweight(\BoundContext_0, (\BoundContext \cup \{ \tvarB \})
  \setminus \{ \tvarA \},
  \SessionTypeS'\subst{\trec\tvarA.\SessionTypeT}{\tvarA})$.
\end{iteMize}

\item ($\SessionTypeS \seq
  \ExternalChoice{\xmsg{\Tag_i}{\tvar_i}{\Type_i}.\SessionType_i}_{i\in
    I}$)
  By induction hypothesis on $\Type_i$ and $\SessionTypeT_i$ for $i\in
  I$ we deduce:
\begin{enumerate}[(1)]
\item $\aweight(\BoundContext_0, \emptyset, \Type_i) =
  \aweight(\BoundContext_0, \emptyset,
  \Type_i\subst{\trec\tvarA.\SessionTypeT}{\tvarA})$;

\item $\tvarA \in \BoundContext$ implies \[\quad\enspace\aweight(\BoundContext_0,
  \BoundContext \setminus \{ \tvar_i \}, \SessionTypeT_i) \leq
  \aweight(\BoundContext_0, \BoundContext \setminus \{ \tvar_i, \tvarA
  \}, \SessionTypeT_i\subst{\trec\tvarA.\SessionTypeT}{\tvarA}) \leq
  \max\{w, \aweight(\BoundContext_0, \BoundContext \setminus \{
  \tvar_i \}, \SessionTypeT_i)\};\]

\item $\tvarA \not\in \BoundContext$ implies
  $\aweight(\BoundContext_0, \BoundContext \setminus \{ \tvar_i \},
  \SessionTypeT_i) = \aweight(\BoundContext_0, \BoundContext \setminus
  \{ \tvar_i, \tvarA \},
  \SessionTypeT_i\subst{\trec\tvarA.\SessionTypeT}{\tvarA})$.
\end{enumerate}
If $\tvarA \in \BoundContext$ we conclude:
\[
\begin{array}{@{}r@{~}c@{~}l@{}}
  \AWS & = &
  \max\{
  1 + \aweight(\zvars, \emptyset, \Type_i),
  \aweight(\zvars, \BoundContext \setminus \{ \tvar_i \}, \SessionType_i)
  \}_{i\in I}
  \\
  & = & 
  \max\{
  1 + \aweight(\BoundContext_0, \emptyset,
  \Type_i\subst{\trec\tvarA.\SessionTypeT}{\tvarA}),
  \aweight(\zvars, \BoundContext \setminus \{ \tvar_i \},
  \SessionType_i)
  \}_{i\in I}
  \\
  & \leq &
  \max\{
  1 + \aweight(\BoundContext_0, \emptyset,
  \Type_i\subst{\trec\tvarA.\SessionTypeT}{\tvarA}),
  \aweight(\BoundContext_0, \BoundContext \setminus \{ \tvar_i, \tvarA
  \}, \SessionTypeT_i\subst{\trec\tvarA.\SessionTypeT}{\tvarA})
  \}_{i\in I}
  \\
  & = & \AWSS
  \\
  & \leq & 
  \max\{
  1 + \aweight(\BoundContext_0, \emptyset,
  \Type_i\subst{\trec\tvarA.\SessionTypeT}{\tvarA}),
  w,
  \aweight(\zvars, \BoundContext \setminus \{ \tvar_i \},
  \SessionType_i)
  \}_{i\in I}
  \\
  & = &
  \max\{
  1 + \aweight(\BoundContext_0, \emptyset,
  \Type_i),
  w,
  \aweight(\zvars, \BoundContext \setminus \{ \tvar_i \},
  \SessionType_i)
  \}_{i\in I}
  \\
  & = & \max\{w, \AWS\}
\end{array}
\]
If $\tvar \not\in \BoundContext$ we conclude:
\[
\begin{array}{@{}r@{~}c@{~}l@{}}
  \AWS & = &
  \max\{
  1 + \aweight(\zvars, \emptyset, \Type_i),
  \aweight(\zvars, \BoundContext \setminus \{ \tvar_i \}, \SessionType_i)
  \}_{i\in I}
  \\
  & = &
  \max\{
  1 + \aweight(\BoundContext_0, \emptyset,
  \Type_i\subst{\trec\tvarA.\SessionTypeT}{\tvarA}),
  \aweight(\BoundContext_0, \BoundContext \setminus \{ \tvar_i, \tvarA
  \}, \SessionTypeT_i\subst{\trec\tvarA.\SessionTypeT}{\tvarA})
  \}_{i\in I}
  \\
  & = & \AWSS
  \,.
\end{array}
\]
\qedhere
\end{iteMize}
\end{proof}

\noindent The next lemma states that, if the weight algorithm determines a
weight $n$ for some endpoint type $\SessionTypeT$, then $n$ is a
weight bound for $\SessionTypeT$.

\begin{lemma}
\label{lem:weight_correctness}
If $\aweight(\BoundContext, \emptyset, \SessionTypeT) = n \in
\natset$, then $\BoundContext \vdash \SessionTypeT \wbound n$.
\end{lemma}
\begin{proof}
  It is enough to show that
  \[ {\wrel} \eqdef \{ (\BoundContext, \SessionTypeT, n) \mid
  \aweight(\BoundContext, \emptyset, \SessionTypeT) \leq n \in \natset \}
\]
is a coinductive weight bound. Let $(\BoundContext, \SessionTypeT, n)
\in {\wrel}$. Then \hypo{h} $\aweight(\BoundContext, \emptyset,
\SessionTypeT) \leq n \in \natset$.
Without loss of generality we may assume that $\SessionTypeT$ does
\emph{not} begin with a recursion. If this were not the case, by
contractivity of endpoint types we have $\SessionTypeT =
\SessionTypeT'$ where $\SessionTypeT'$ does not begin with a
recursion.
Now, by Proposition~\ref{prop:unfold_weight} we deduce
$\aweight(\BoundContext, \emptyset, \SessionTypeT') =
\aweight(\BoundContext, \emptyset, \SessionTypeT) = n$ and therefore
$(\BoundContext, \SessionTypeT', n) \in {\wrel}$ by definition of
$\wrel$.

We reason by cases on $\SessionTypeT$:
\begin{iteMize}{$\bullet$}
\item ($\SessionTypeT \seq \SessionEnd$ or $\SessionTypeT \seq
  \InternalChoice{\xmsg{\Tag_i}{\tvar_i}{\Type_i}.\SessionTypeT_i}_{i\in
    I}$) There is nothing to prove.

\item ($\SessionTypeT \seq \tvar$) From \hypo{h} we deduce $\tvar \in
  \BoundContext$ and there nothing left to prove.

\item ($\SessionTypeT \seq
  \ExternalChoice{\xmsg{\Tag_i}{\tvar_i}{\Qualifier_i~\SessionTypeS_i}.\SessionTypeT_i}_{i\in
    I}$)
  Then $0 < \aweight(\BoundContext, \emptyset, \SessionTypeT) \leq n$.
  From \hypo{h} we deduce $\aweight(\BoundContext, \emptyset,
  \SessionTypeS_i) \leq n - 1$ and $\aweight(\BoundContext, \emptyset,
  \SessionTypeT_i) \leq n$ for every $i\in I$.
  We conclude $(\BoundContext, \SessionTypeS_i, n - 1) \in {\wrel}$
  and $(\BoundContext, \SessionTypeT_i, n) \in {\wrel}$ for every
  $i\in I$.
  \qedhere
\end{iteMize}
\end{proof}

\noindent The last auxiliary result proves that the weight algorithm computes
the least upper weight bound for an endpoint type. We use
$\substitution$ to range over arbitrary substitutions of endpoint
types in place of type variables, we write
$\SessionTypeT\substitution$ for $\SessionTypeT$ where the
substitutions in $\substitution$ have been applied, and
$\dom(\substitution)$ for the domain of $\substitution$ (the set of
type variables that are instantiated).

\begin{lemma}
\label{lem:weight_completeness}
If $\BoundContext \vdash \SessionTypeT\substitution \wbound n$, then
$\aweight(\BoundContext, \dom(\substitution), \SessionTypeT) \leq n$.
\end{lemma}
\begin{proof}
  By induction on $\SessionTypeT$:
\begin{iteMize}{$\bullet$}
\item ($\SessionTypeT \seq \SessionEnd$ or $\SessionTypeT \seq
  \InternalChoice{\xmsg{\Tag_i}{\tvar_i}{\Type_i}.\SessionTypeT_i}_{i\in
    I}$) Easy since $\aweight(\BoundContext, \dom(\substitution),
  \SessionTypeT) = 0$.

\item ($\SessionTypeT \seq \tvar$) From the hypothesis $\BoundContext
  \vdash \SessionTypeT\substitution :: n$ we deduce $\tvar \in
  \BoundContext \cup \dom(\substitution)$. By definition of
  algorithmic weight we conclude $\aweight(\BoundContext,
  \dom(\substitution), \SessionTypeT) = 0$.

\item ($\SessionTypeT \seq \trec\tvar.\SessionTypeS$)
  Let $\substitution' = (\substitution \setminus \tvar),\{ \tvar
  \mapsto \SessionTypeT \}$ where $\substitution \setminus \tvar$ is
  the restriction of $\substitution$ to $\dom(\substitution) \setminus
  \{ \tvar \}$. We have $\BoundContext \vdash
  \SessionTypeS\substitution' \wbound n$.
  By induction hypothesis we deduce $\aweight(\BoundContext,
  \dom(\substitution) \cup \{ \tvar \}, \SessionTypeS) \leq n$.
  By definition of algorithmic weight we conclude
  $\aweight(\BoundContext, \dom(\substitution), \SessionTypeT) =
  \aweight(\BoundContext, \dom(\substitution),
  \trec\tvar.\SessionTypeS) = \aweight(\BoundContext,
  \dom(\substitution) \cup \{ \tvar \}, \SessionTypeS) \leq n$.

\item ($\SessionTypeT \seq
  \ExternalChoice{\xmsg{\Tag_i}{\tvar_i}{\Qualifier_i~\SessionTypeS_i}.\SessionTypeT_i}_{i\in
    I}$)
  For every $i\in I$ let $\substitution_i = \substitution \setminus \{
  \tvar_i \}$.
  From the hypothesis $\BoundContext \vdash \SessionTypeT\substitution
  \wbound n$ we deduce $\BoundContext \vdash
  \SessionTypeS_i\substitution_i \wbound n - 1$ and $\BoundContext
  \vdash \SessionTypeT_i\substitution_i \wbound n$ for every $i \in I$.
  By induction hypothesis we deduce $\aweight(\BoundContext,
  \dom(\substitution_i), \SessionTypeS_i) \leq n - 1$ and
  $\aweight(\BoundContext, \dom(\substitution_i), \SessionTypeT_i)
  \leq n$.
  We conclude $\aweight(\BoundContext, \dom(\substitution),
  \SessionTypeT) \leq n$ by definition of algorithmic weight.
  \qedhere
\end{iteMize}
\end{proof}

\noindent Correctness of the weight algorithm is simply a combination of the two
previous lemmas.

\begin{theorem}[Theorem~\ref{thm:weight_algorithm}]
  $\xweight\BoundContext\SessionType = \aweight(\BoundContext,
  \emptyset, \SessionType)$.
\end{theorem}
\begin{proof}
  From Lemma~\ref{lem:weight_correctness} we deduce
  $\xweight\BoundContext\SessionType = \min \{ n \in \natset \mid
  \BoundContext \vdash \SessionType \wbound n \} \leq
  \aweight(\BoundContext, \emptyset, \SessionType)$.
  From Lemma~\ref{lem:weight_completeness}, by taking $\substitution =
  \emptyset$ (the empty substitution), we conclude
  $\aweight(\BoundContext, \emptyset, \SessionType) \leq
  \xweight\BoundContext\SessionType$.
\end{proof}


\end{document}